\documentclass[letterpaper,english,11pt]{scrartcl}

 \newif\ifreview\reviewfalse

 \newif\ifjournal\journalfalse

\usepackage{amsmath,amsfonts,amsthm,amssymb}

\allowdisplaybreaks

\usepackage{relsize}
\usepackage{mathtools}
\usepackage{braket}
\usepackage{aligned-overset}

\usepackage[margin=1in,bottom=1.3in]{geometry}
\usepackage{authblk}
\usepackage{enumitem,framed}

\usepackage{longtable}

\usepackage[linesnumbered,ruled,nokwfunc,longend]{algorithm2e}
\DontPrintSemicolon
\SetKwInOut{Input}{Input}\SetKwInOut{Output}{Output}
\SetKwProg{fetry}{foreach }{ try}{} 
\SetKwBlock{catch}{catch }{end foreach try}

\usepackage{tikz}
\usetikzlibrary{tikzmark}
\usetikzlibrary{nfold}
\usetikzlibrary{graphs,fadings,arrows,calc}
\usepackage{pgfplots}
\usepackage{graphicx}
\usepgfplotslibrary{statistics} 
\usetikzlibrary{pgfplots.statistics}
\pgfplotsset{compat=1.10}
\usepgfplotslibrary{fillbetween}
\usepackage{subfigure}

\usetikzlibrary{arrows.meta}
\usetikzlibrary{cd}
\tikzstyle{res}=[draw=black, fill=white, minimum size=1.3em, font={\tiny}]
\usepgfplotslibrary{fillbetween}

\usepackage{thm-restate}
\usepackage{xspace}
\usepackage{color, colortbl}
\usepackage{hyperref}
\usepackage{cleveref}

\newcommand{\refsym}[1]{%
	\ensuremath{(%
		\ifthenelse{\equal{#1}{1}}{\ast}%
		{\ifthenelse{\equal{#1}{2}}{\#}%
		{\ifthenelse{\equal{#1}{3}}{\triangle}%
		{\ifthenelse{\equal{#1}{4}}{\bigcirc}%
		{\ifthenelse{\equal{#1}{5}}{\diamond}
		{\color{red}NaN}}}}}%
	)}%
}
\newcommand{\symoverset}[2]{\toverset{\refsym{#1}}{#2}}
\newcommand{\toverset}[2]{\overset{\text{#1}}{#2}}

\newcommand{\Croverset}[2]{\overset{\text{\Crefshort{#1}}}{#2}}
\newcommand{\Crefshort}[1]{%
	{%
		\Crefname{definition}{Def.}{Def.}%
		\Crefname{lemma}{Lem.}{Lem.}%
		\Crefname{proposition}{Prop.}{Prop.}%
		\Crefname{corollary}{Cor.}{Cor.}%
		\Crefname{theorem}{Thm.}{Thm.}%
		\Crefname{observation}{Obs.}{Obs.}%
		\Crefname{claim}{Cl.}{Cl.}%
		\Crefname{section}{Sec.}{Sec.}%
		\Crefname{subsection}{Sec.}{Sec.}%
		\Crefname{example}{Ex.}{Ex.}%
		\Crefname{equation}{Eq.}{Eq.}%
		\Crefname{algorithm}{Alg.}{Alg.}%
		\Cref{#1}%
	}%
}

\DeclareRobustCommand{\Crefnosort}[1]{%
  \begingroup\@cref@sortfalse\Cref{#1}\endgroup
}

\usepackage{stmaryrd}
\usepackage[font={small,it}]{caption}

\definecolor{darkgreen}{rgb}{0.0, 0.5, 0.0}

\newcommand{\R}{\mathbb{R}}
\newcommand{\Rnn}{\mathbb{R}_{\geq 0}}
\newcommand{\N}{\mathbb{N}}

\newcommand{\comment}[1]{}

\newcommand{\func}{A}

\newcommand{\objfunc}{\vartheta}

\newcommand{\eflow}{{\g}}
\newcommand{\wflow}{{h}}
\NewDocumentCommand{\sysop}{O{\ }}{system-optimal{#1}}
\NewDocumentCommand{\umini}{O{\ }}{\ensuremath{\ell}-minimal{#1}}
\NewDocumentCommand{\mini}{O{\ }}{\ensuremath{\ell}-minimal{#1}}
\newcommand{\tmin}{\tau^{\min}}

\NewDocumentCommand{\CharF}{O{}}{%
	\ifthenelse{\equal{#1}{}}%
	{1}%
	{1_{#1}}%
}

\usepackage{bm}

\definecolor{LightCyan}{rgb}{0.88,1,1}
\DeclareMathOperator*{\argmin}{arg\,min}

\def\mybig#1{{\hbox{$\left#1\vbox to23\p@{}\right.\n@space$}}}

\newcommand{\dup}[2]{\langle#1,#2\rangle}

\theoremstyle{definition}
\newtheorem{definition}{Definition}[section]

\theoremstyle{plain}
\newtheorem{theorem}[definition]{Theorem}

\newtheorem{lemma}[definition]{Lemma}

\newtheorem{corollary}[definition]{Corollary}
\newtheorem{claim}[theorem]{Claim}
\newtheorem{subclaim}[theorem]{Subclaim}

\newtheorem{parttheorem}[definition]{Parts of Theorem}
\Crefname{claim}{Claim}{Claims}

\newtheorem{innerrsttheorem}{Theorem}
\newenvironment{rsttheorem}[1]
{\renewcommand\theinnerrsttheorem{#1}\innerrsttheorem}
{\endinnerrsttheorem}

\theoremstyle{remark}
\newtheorem{remark}[definition]{Remark}
\newtheorem{example}[definition]{Example}

\newlist{thmparts}{enumerate}{1}
\setlist[thmparts]{
	label=\alph*)
}

\usepackage{xstring}
\usepackage{etoolbox}
\apptocmd{\cref@getref}{\xdef\@lastusedlabel{#1}}{}{error}

\Crefformat{thmpart}{%
	\StrCount{\@lastusedlabel}{:}[\LastColonPos]%
	\StrBefore[\LastColonPos]{\@lastusedlabel}{:}[\ThmLabel]%
	\nameCref{\ThmLabel}~#2\ref*{\ThmLabel}#1#3%
}
\Crefmultiformat{thmpart}{%
	\StrCount{\@lastusedlabel}{:}[\LastColonPos]%
	\StrBefore[\LastColonPos]{\@lastusedlabel}{:}[\ThmLabel]%
	\nameCref{\ThmLabel}~#2\ref*{\ThmLabel}#1#3%
}{ and~#2#1#3}{, #2#1#3}{ and~#2#1#3}

\newenvironment{proofClaim}[1][]{\ifthenelse{\equal{#1}{}}{\begin{proof}}{\begin{proof}[#1]}}{\end{proof}}

\newlist{proofbycases}{enumerate}{1}
\setlist[proofbycases]{
	leftmargin=0em,
	labelwidth=-.5em,
    parsep=0pt,
    listparindent=\parindent,
	label=\boldmath\bfseries\sffamily\arabic*. Case: \protect\casedescr:,
	ref=\arabic*,
	align=left
}

\newcommand{\proofitem}[1]{\def\pidescr{#1}%
	\item}
\newlist{structuredproof}{enumerate}{3}
\setlist[structuredproof]{
	leftmargin=1em,
    labelwidth=.3em,
    parsep=0pt,
    listparindent=\parindent,
	label=\boldmath\bfseries\sffamily\protect\pidescr:,
	align=left
}
 \setlist[structuredproof,2]{
     leftmargin=1em
 }
 \setlist[structuredproof,3]{
     leftmargin=1em
 }

\newenvironment{proofbyinduction}{\begin{description}[leftmargin=0em,parsep=0pt,listparindent=\parindent]}{\end{description}}

\newcommand{\basecase}[1]{\item[Base Case ({\boldmath\bfseries#1}):]}
\newcommand{\inductionstep}[1]{\item[Induction Step ({\boldmath\bfseries#1}):]}

\DeclareFontFamily{OT1}{pzc}{}
\DeclareFontShape{OT1}{pzc}{m}{it}{<-> s * [1.10] pzcmi7t}{}
\DeclareMathAlphabet{\mathpzc}{OT1}{pzc}{m}{it}
 
\newcommand{\myparagraph}[1]{\paragraph{#1.}}

\newcommand{\fixed@sra}{$\vrule height 2\fontdimen22\textfont2 width 0pt\shortrightarrow$}
\newcommand{\shortarrow}[1]{%
  \mathrel{\text{\rotatebox[origin=c]{\numexpr#1*45}{\fixed@sra}}}
}

\newcommand{\norm}[1]{\lVert #1 \rVert}
\newcommand{\abs}[1]{\lvert#1\rvert}
\newcommand{\ceil}[1]{\lceil #1 \rceil}

\newcommand{\wlg}{w.l.o.g.\ }

\newcommand{\wrt}{w.r.t.\ }

\usepackage{comment}

\newcommand{\Meas}{\mathcal{M}}

\newcommand{\di}{\;\mathrm{d}}

\newcommand{\Feas}{\Lambda^*}
\newcommand{\arc}{e}

\newcommand{\g}{g}

\newcommand{\exit}{T}
\newcommand{\arr}{A}
\newcommand{\trav}{D}

\newcommand{\hori}{\R}

\newcommand{\Routes}{\mathcal{W}}
\newcommand{\wa}{w}
\newcommand{\inflow}{r}

\newcommand{\id}{\mathrm{id}}

\newcommand{\GA}{E}
\newcommand{\GV}{V}
\newcommand{\source}{s}
\newcommand{\dest}{{d}}
\newcommand{\edgesFrom}[1]{\delta^+(#1)}
\newcommand{\edgesTo}[1]{\delta^-(#1)}
\newcommand{\sink}{destination}

\NewDocumentCommand{\stwalk}{O{\source}O{\dest}}{\ensuremath{#1},\ensuremath{#2}-walk}
\NewDocumentCommand{\stpath}{O{\source}O{\dest}}{\ensuremath{#1},\ensuremath{#2}-path}

\newcommand{\refrunind}[2]{%
\ifthenelse{\equal{#2}{1}}{
	\ensuremath{%
		\ifthenelse{\equal{#1}{1}}{n}%
		{\ifthenelse{\equal{#1}{2}}{m}%
		{\ifthenelse{\equal{#1}{3}}{j}%
		{\ifthenelse{\equal{#1}{4}}{l}%
		{\ifthenelse{\equal{#1}{5}}{i}
		{\color{red}NaN}}}}}%
	}%
 }
{
\ensuremath{%
		\ifthenelse{\equal{#1}{1}}{N}%
		{\ifthenelse{\equal{#1}{2}}{M}%
		{\ifthenelse{\equal{#1}{3}}{J}%
		{\ifthenelse{\equal{#1}{4}}{L}%
		{\ifthenelse{\equal{#1}{5}}{I}
		{\color{red}NaN}}}}}%
	}%
}
}
 
\newcommand{\n}[1]{\ifthenelse{\equal{#1}{}}{\refrunind{1}}{\refrunind{#1}{1}}}
\newcommand{\capn}[1]{\ifthenelse{\equal{#1}{}}{\refrunind{1}}{\refrunind{#1}{2}}}

\newcommand{\diffe}{\Delta}

\newcommand{\wto}{\rightharpoonup}

\newcommand{\tarc}{\Tilde{\arc}}
\newcommand{\teflow}{\Tilde{\eflow}}
\newcommand{\tg}{\tilde{\g}}
\newcommand{\tv}{\Tilde{v}}
\newcommand{\eset}{\Tilde{\mathcal{E}}}
\NewDocumentCommand{\seql}{O{1}O{\Routes}O{L(\hori)}}{\otimes^{#1}_{#2}#3}

\newcommand{\ofeas}{\mathcal{M}}
\NewDocumentCommand{\auto}{O{\ }}{autonomous#1}
\NewDocumentCommand{\Aauto}{O{\ }}{An autonomous#1}
\newcommand{\aauto}{an autonomous }
\newcommand{\Auto}{Autonomous }
\newcommand{\startint}{{(-\infty,\,}}

 \tikzstyle{vertex} = [shape=circle,draw=black]
\tikzstyle{namedVertex} = [shape=circle,draw=black]
\tikzstyle{namedVertexF} = [shape=circle,draw=black,fill=white]
\tikzstyle{edge} = [draw,->,thick]
\tikzstyle{labeledNodeS}=[circle, color=black!75!white, draw, inner sep = 0.1em, minimum size = 1.5em, scale=1.25]
\tikzstyle{normalEdge}=[very thick, >=stealth]

 \usepackage{apptools}
 \usepackage{thmtools}
\usepackage{thm-restate}

\ifreview
    \usepackage[disable]{todonotes}
\else
    \usepackage[textsize=tiny,textwidth=2cm,shadow,loadshadowlibrary]{todonotes}
\fi
\usepackage{booktabs}
\newcommand{\lgcom}[2][]{\ifthenelse{\equal{#1}{journal}}{\todo[color=green!90!blue!30]{LG (for future version): #2}}{\todo[color=green!70!blue!60]{LG: #2}}} 

 \newcommand{\measfunc}{f}
 \newcommand{\leb}{\sigma}
  \newcommand{\eqperdef}{\overset{\text{def}}{=}}
 \newcommand{\defpereq}{\eqperdef}
 \newcommand{\sdRoutes}{\hat{\Routes}_{\source,\dest}}
 \newcommand{\SimpCyc}{\mathcal{C}}
 \newcommand{\Indi}{\mathbf{1}}
  \NewDocumentCommand{\op}{O{}}{\nabla^{#1}}
\usepackage{nicefrac}
\usepackage{adjustbox}
\NewDocumentCommand{\Nl}{O{\trav(\cdot,\cdot)}}{\ell^{#1}} 
\newcommand{\FeasSol}{\mathcal{F}}
\newtheorem{innerrstlemma}{Lemma}
\newenvironment{rstlemma}[1]
{\renewcommand\theinnerrstlemma{#1}\innerrstlemma}
{\endinnerrstlemma}
\NewDocumentCommand{\edom}{mO{}}{\mathcal H_{#1}^{#2}}

\ifreview
    \newcommand{\BigPicture}[2][0]{#2}
\else
    
    \newcommand{\BigPicture}[2][0]{#2}
\fi

\title{A Decomposition Theorem for Dynamic Flows}

\ifreview
    \author{}
    \date{\vspace{-2cm}}
\else
    \author{Lukas Graf, Tobias Harks and Julian Schwarz} 

    \affil{\small University of Passau, Faculty of Computer Science and Mathematics, 94032 Passau\\
    \href{mailto:julian.schwarz@uni-passau.de}{\{\texttt{lukas.graf,tobias.harks,julian.schwarz\}@uni-passau.de}}}
\fi

\begin{document}

\maketitle
\begin{abstract}
The famous  flow decomposition theorem of Gallai~\cite{Gallai58}
states that any static edge $\source$,$\dest$-flow in a directed graph can be decomposed into a
nonnegative linear combination of incidence vectors of paths and cycles.
In this paper, we study the decomposition problem for the 
setting of \emph{dynamic} edge $\source$,$\dest$-flows assuming a quite general dynamic flow
propagation model. We prove the following decomposition theorem: For any integrable  dynamic  edge $\source$,$\dest$-flow, 
there exists a decomposition into a nonnegative linear combination of \stwalk{} inflows and cycles of zero transit time. 
We show that a variant of the classical algorithmic approach of  iteratively
subtracting walk inflows from the current dynamic edge flow converges to a dynamic circulation and that every such circulation can be induced by inflows into cycles of zero transit time.
The algorithm terminates in finite time, if
there is a lower bound on the minimum edge travel times and the flow is finitely supported. We further characterize those
dynamic edge flows which can be decomposed purely into nonnegative linear combinations of \stwalk{} inflows. 

The proofs rely on the new concept of \auto network loadings which allows us to describe how particles of a different walk flow would hypothetically propagate throughout the network under the fixed travel times induced by the given edge flow. 
We show several technical properties of this type of network loading and, as a byproduct, we also derive some general results on dynamic flows which could be of interest outside the context of this paper as well.
\end{abstract} 

\ifreview
\else
\clearpage
\tableofcontents
\clearpage
\fi

\section{Introduction}\label{sec:intro}  
Dynamic network flows are an important mathematical concept in network flow theory with applications in the areas of dynamic traffic assignment, production systems and communication networks. 
As one of the earliest works in this area, Ford and Fulkerson~\cite{Ford62,FordFulkersonDynamicFlows} proposed dynamic flows as a generalization of static flows incorporating a time component.
The dynamic nature arises by assuming that flow particles require a certain amount of time to travel through each edge and when flow is injected into paths at certain points in time, the flow propagation leads to later effects in other parts of the network -- this flow propagation is often called \emph{network loading}. While Ford and Fulkerson~\cite{Ford62,FordFulkersonDynamicFlows}  and also further works in the area (see the survey of Skutella~\cite{Skutella08})
assumed flow-independent, constant travel times,
more realistic network loading models come with \emph{flow-dependent} and time-vaying travel times.
Such models  have been considered extensively in the dynamic traffic assignment community, see for instance the various link-delay formulations~\cite{Friesz93,ZhuM00,KohlerS05}, the Vickrey model with point queues~\cite{CominettiCL15,GrafH22,GHKM23,GrafHS20,Koch11,MeunierW10,OlverSK21,OlverSK23} or
the  Lighthill-Whitham-Richards (LWR) model~\cite{FrieszHanPedro13,LW55,Richards56}.

For a dynamic flow model with flow-dependent and time-varying travel times, the network loading problem asks for the
evolution of dynamic edge flows and corresponding edge travel times for given inflow rates into the paths or walks of the network.
The \emph{flow decomposition problem}, on the other hand, asks for the inverse: 
Given a dynamic  edge $\source$,$\dest$-flow (i.e., inflow functions per edge satisfying balance constraints at all vertices except at the source and the destination), can we decompose the edge flows into walk inflow rates and circulations so that these
walk inflow rates and circulations result in the given edge flow?
This question plays a prominent role for static flows and is answered by the static flow decomposition
theorem (see Gallai~\cite[Satz~1.2.7]{Gallai58}) stating that any static edge $\source$,$\dest$-flow can be decomposed into a nonnegative linear combination
of incidence vectors of paths and cycles. The decomposition property comes into play
at various places: for proving optimality conditions for minimum cost flows, transhipments and more, see Schrijver~\cite[Chapter 11]{Schrijver03} for a comprehensive treatment.
An analogue of this decomposition theorem for dynamic flows with flow dependent transit time is -- to the best of our knowledge -- not known so far.

\subsection{Our Results}\label{sec:results}
For a quite general network-loading model, which includes the linear edge-delay model and the Vickrey queueing model as special cases, we consider the following decomposition
problem: Is an  integrable dynamic   edge $\source$,$\dest$-flow decomposable into a nonnegative linear combination of \stwalk-inflows
and circulations?
Our main result settles this problem:

\begin{rsttheorem}{\ref{thm: FLowDecompModel} (informal)}
    Every integrable edge $\source$,$\dest$-flow admits a flow decomposition into nonnegative linear combinations of \stwalk{\footnote{Note that for a decomposition theorem for dynamic edge flows, it is in general necessary to allow for $\source$,$\dest$-walk inflows instead of only considering (simple) path inflows. To see this, just send inflow  for the time interval $[0,1]$ from  $s$ to $d$ along an $\source$,$\dest$ walk containing a simple cycle. Further, suppose that all travel times are constant $1$.
If we view the resulting dynamic edge flow as the input of the decomposition problem, the only possible decomposition is exactly the described walk inflow rate we started with. 
}} inflows and cycles of zero transit time.    
\end{rsttheorem}

A decomposition into walk-inflows and a dynamic circulation, i.e.~an edge flow fulfilling flow conservation at every node, can be found by the natural inflow reduction algorithm (\Cref{alg: FlowDecompositionPseudo}). The latter can then by decomposed further into cycle inflows of zero transit time.

 \begin{algorithm}
		\caption{Flow Decomposition Algorithm -- Pseudocode}\label{alg: FlowDecompositionPseudo}
		\Input{An   edge $\source$,$\dest$-flow $\g$}
		\Output{Walk inflow rates $h$  such that the difference of $\g$ and the corresponding edge flow of $h$ is a dynamic circulation}
		enumerate all \stwalk s $\hat{\Routes} = \{\wa_k\}_{k \in \N}$  
        and set $\g^1 \leftarrow \g$   
        
        \For{$k \in \N$}{
        Subtract from $\g^k$ as much flow as possible via a walk inflow rate $h_{\wa_k}$ into walk $\wa_k$ and set the remaining flow to $\g^{k+1}$
        \label{line:MainStep}
        }        
    \KwRet{$h_{\wa_k}, k \in \N$}
\end{algorithm}
 
Of particular interest are those dynamic edge flows that are decomposable into
$\source$,$\dest$-walk inflows only. Here we give the following  combinatorial characterization of  this property:

\begin{rsttheorem}{\ref{thm: PureFLowDecompModel} (informal)}
    An integrable edge $\source$,$\dest$-flow admits a flow decomposition
    purely into $\source$,$\dest$-walk inflows if and only if for any cycle and for (almost) all times where this cycle has zero transit time and carries flow, at least one of the following two properties is satisfied: 
    \begin{itemize}
        \item The destination is contained in the cycle and has positive net inflow at that time. 
        \item There is an edge leaving the cycle with positive edge inflow at that time.
    \end{itemize}
\end{rsttheorem}
Let us remark that Koch gave 
a similar characterization in~\cite[Lemma~3.47]{Koch11} 
requiring that every flow-carrying cycle is connected to a flow carrying walk under the starting flow decomposition. 
This characterization, however, is incorrect as the stated condition is too strong and, thus, only yields a sufficient but not a necessary condition, cf.~the counter example in \Cref{fig:KochPureFD-CounterExample}. 
Moreover, Koch's proof of sufficiency of the above condition is also incomplete: He only demonstrates the transformation of a given flow decomposition satisfying the sufficient conditions into a pure walk decomposition \emph{for a single cycle, walk and point in time}. To actually obtain a complete pure decomposition, however, one needs to do this \emph{for all cycles, walks and times at once}, which is technically much more involved.

Finally, we demonstrate in \Cref{sec:EquivalenceOfDE} the usefulness of the above decomposition result by showing  that the two natural definitions for dynamic equilibria, namely the edge- and walk-based definitions, respectively, do in fact describe the same solution concept (\Cref{thm:def-equilibrium-equivalent}). This equivalence is shown to hold even for generalized cost functions that may include tolls, rebates or emission costs. 

\subsection{Challenges and Technical Contributions}\label{sec:techniques}

The proof of the above decomposition theorems mainly rests on analyzing \Cref{alg: FlowDecompositionPseudo}. We start by briefly discussing the main challenges here as well as giving a high-level overview of our solutions to them.

\begin{description}
    \item[Formalization:]
    While the intuitive idea of the above algorithm is clear, it turns out that it is not trivial how to formalize the main step (\Cref{line:MainStep}) of \Cref{alg: FlowDecompositionPseudo} in a mathematically precise way. 
    In particular: What are the objects considered here and what does it mean to subtract a walk flow from an edge flow? To address these questions we introduce the concept of \auto network loadings, that is, the hypothetical flow propagation  of some walk inflow under any fixed travel time function which varies absolutely continuous in time. 
    The \auto network loading allows us to  translate walk inflow rates into \emph{\auto flows}, i.e.~edge inflow rates fulfilling flow conservation \wrt the fixed time-varying travel times. This way, we can view all intermediate flows $\g^k$ occurring during the algorithm as such \auto flows with respect to the fixed travel times induced by the initially given edge flow. 
    Within this framework, we can then subtract \auto flows induced by walk inflow rates from these intermediate edge flows and state a suitable optimization problem for \Cref{line:MainStep} that characterizes the maximal possible walk inflow among all possible walk inflows. 
    \item[Well-definedness:]
     To demonstrate that the algorithm is now well-defined we have to show that this optimization problem is itself well-posed and guaranteed to have an optimal solution. For the former, one has to be careful as not every walk inflow induces an  edge flow under fixed travel times (see \Cref{exa: noarcflow}). Hence, the feasible domain of the optimization problem has to be chosen in exactly the right way to include all possible walk inflows and exclude all others. 
     Regarding the existence of an optimal solution, we aim to apply the Weierstrass Theorem, 
     where the main challenge lies in showing that the set of feasible solutions is compact \wrt an appropriate topology.
    \item[Correctness:]
    As the central tool for showing correctness we  show that any flow satisfying flow conservation at all nodes other than the source and sink either has a flow carrying \stwalk{} or is a dynamic circulation. Showing that the limit of the intermediate flows~$\g^k$ satisfies this property then allows us to deduce the correctness of the algorithm since it removes the maximal amount of flow from every \stwalk.
\end{description}

\begin{figure}
 	\centering
 	\begin{adjustbox}{max width=\textwidth,max height=\textheight-9em}
 		\input{tikz/Overview} 	
 	\end{adjustbox}
 	\caption{A schematic overview over the structure of this paper: The top part contains our main results (for \auto flows in the middle and for non-\auto flows on the right) as well as the main building blocks for their proofs (on the left). Implication arrows indicate that one result is the main ingredient for proving another. Single arrows indicate that a result is used within the corresponding proof. The bottom half contains more technical results on \auto network loadings, which are used in many proofs. Not all relations between these results and other results are explicitly drawn. The following abbreviations and notations are used: auton.\ for \auto[,] flow cons.\ for flow conservation, flow decomp.\ for flow decomposition and $\ell_{\wa}(h_\wa)$ and $\ell_\Routes(h)$ for the \auto network loading of a single walk inflow~$h_\wa$ and a whole walk inflow vector~$h$, respectively.}
 	\label{fig:Overview}
 \end{figure}

In the following, we describe the technical contributions in more detail (see also \Cref{fig:Overview} for a schematic overview of our main results and their proof structure).

\paragraph{\Auto Network Loadings.}
As our key concept for formalizing \Cref{alg: FlowDecompositionPseudo}, we introduce  \auto network loadings (i.e., walk inflows that are loaded with respect to fixed travel time functions which vary absolutely continuous in time) and study
their properties. 
We start  with a characterization of walk-inflows that admit a corresponding \auto edge flow by the following property (\Cref{lem: elluExistenceProperties}): 
For any set of times such that a positive measure of flow enters some fixed walk, the set of \auto arrival times  at any edge along that walk must have positive measure as well. 
 This result allows us to restrict the feasible
space of walk-inflows which have to be considered in the flow decomposition algorithm. With this we can define a corresponding optimization problem and show that it admits an optimal solution (\Cref{thm: ExistenceOptSol}).

Next, we consider the concept of \auto node balances which allows us to define flow conservation with respect to fixed time-varying travel times. We then call flows satisfying this type of flow conservation at all nodes, except the source and destination,  \auto $\source$,$\dest$-flows and show in \Cref{lem: flowconW'} that 
subtracting an $\source$,$\dest$-walk from \aauto edge $\source$,$\dest$-flow results in another  \auto edge $\source$,$\dest$-flow.

Finally, based on 
several structural insights into \auto network loadings, we formulate the two main ingredients for the proof of our decomposition theorem:
\Cref{lem: ZeroCycleDecomposition} shows that any \auto dynamic circulation 
has a decomposition into zero-cycle inflows. 
 \Cref{lem: ExistenceOfFlowCarryingWalk} states that  any \auto $\source$,$\dest$-flow  with positive outflow at the source admits a flow-carrying \stwalk . Note that finding such a walk is much harder in the dynamic case compared to the static case where a simple breath-first search in the subnetwork of flow carrying edges starting at the source suffices. This is because we have to find such a walk not just for a single particle but for a positive measure of particle at once (i.e.\ a set of starting times of positive measure) and the ``time-expanded'' graph in which this search has to take place is of infinite size and, hence, termination of the search procedure is not obvious.
 To address these issues we devise an algorithm (\Cref{algo: sdwalk}) that pushes flow along (flow carrying) outgoing edges starting with the positive network inflow at the source. The flow receiving nodes together with 
the pushed flow  are then recorded in a tree structure. We show termination of this algorithm by using the tree structure and a potential argument tracking the total volume of pushed flow across the layers of the tree.

\paragraph{Flow Decomposition.}
With the above results on \auto network loadings at hand, we then turn back to the problem of flow decomposition. 
For the existence of a flow decomposition (\Cref{thm: FlowDecomp}), we mainly have to show the correctness of the (formal) decomposition algorithm (\Cref{alg: FlowDecomposition}). 
The aforementioned \Cref{lem: flowconW'}  ensures that the limit of the sequence of remaining flows $(\g^k)$ is still  \aauto $\source$,$\dest$-flow. 
 \Cref{lem: ExistenceOfFlowCarryingWalk} then guarantees that this limit fulfills flow conservation at the source as well, since the existence of 
 a flow carrying \stwalk{} $\wa_k$ would lead to a contradiction to the maximality of the removed walk inflow $h_{\wa_k}$. 
 From this, it follows by \Cref{lem: FlowConEveryNode} that we have flow conservation at~$\dest$ as well, i.e., a dynamic circulation. Hence, the algorithm is correct.
 Thus, \Cref{lem: ZeroCycleDecomposition} is applicable, implying that the dynamic circulation of the algorithm's output can be further decomposed into cycle-inflows with zero transit times.

In \Cref{sec:FlowDecomp:Pure} we then consider pure flow decomposition, i.e.\ one where flow is sent only via $\source$,$\dest$-walks. 
We characterize 
in \Cref{thm: PureFlowDecompIntuitive} the flows admitting such a pure decomposition as 
those flows where 
flow is sent into a zero-cycle $c$ only if there exists a flow carrying edge leaving~$c$ or $c$ contains the destination and the  network outflow rate is positive. 
Proving necessity is relatively straightforward. For sufficiency, we start with a general flow decomposition (which exists by \Cref{thm: FlowDecomp}) and then adjust it by incorporating any inflows into zero-cycles into some \stwalk s. Note that this step is technically challenging as the zero-cycles might not be directly connected to any flow carrying \stwalk{} but only indirectly via other zero-cycles.
Moreover, the flow rates may not match directly. 

 Finally, we deduce from  \Cref{thm: PureFlowDecompIntuitive} the existence of maximally pure flow decompositions, that is, flow decompositions that only use inflow into zero-cycles when it is unavoidable (\Cref{cor: PureFlowDecompIntuitive}).

\subsection{Applications of the Flow Decomposition Theorem}
Decomposing dynamic edge flows into walk inflows is an important task in the traffic assignment literature, see Peeta and Ziliaskopoulos~\cite[Section 3.1.5]{Peeta01} for an overview. 
A central problem here is to reverse engineer from given dynamic edge flow measurements an underlying walk inflow distribution which -- after network loading -- can explain the edge flow measurements, see Cascetta et al~\cite{CascettaIM93}. 
Our flow decomposition theorem shows that this problem is actually well-posed, that is, dynamic edge flow measurements obeying conservation laws at all nodes (except $\source$,$\dest$) can be decomposed into a nonnegative linear combination of \stwalk{} inflows and circulations.

The flow decomposition problem 
also plays a fundamental role for showing the
equivalence of edge and path-based definitions of dynamic equilibria. In the past decade, there has been exciting progress on 
our understanding of dynamic equilibrium flows
for the Vickrey queuing model~\cite{Vickrey69}. 
The main results of this literature include the existence and characterization  of dynamic equilibria (Koch and Skutella~\cite{Koch11}, Cominetti et al.~\cite{CominettiCL15}, Sering and Vargas-Koch~\cite{Sering2019}),  their uniqueness, stability and convergence properties (Olver et al.~\cite{OlverSK21}), their long term behavior with infinitely lasting inflow (Cominetti et al.~\cite{CCO22})
and their price of anarchy (Correa et al.~\cite{CCOPoAforNashFlows}). 

The crucial point is that these works
use
an \emph{edge-based} definition of dynamic equilibria, that is, 
a dynamic edge $\source$,$\dest$-flow is a dynamic equilibrium if at almost all times, positive edge-inflow only occurs, if the edge lies on a shortest path (\wrt induced travel times) to the \sink. 
Dynamic equilibrium flows are, however,  fundamentally \emph{path-based} objects, because paths are exactly the strategies of the flow particles.  
Under the assumption that
travel time is the only payoff-relevant measure for the flow particles, the arrival times of particles at intermediate nodes
are synchronized for the Vickrey queuing model leading to the notion of (static) \emph{thin flows}
as first discovered by Koch and Skutella~\cite{Koch11}.
These static thin flows can be decomposed into paths (using the classical static flow decomposition theorem) and, thus,
gives the equivalence between edge- and path-based dynamic equilibria (see \cite[Section~2.7]{CominettiCL15}). 

For other physical flow models and other types of equilibria, however, 
a comparable thin flow formulation is not known. For example, 
in the traffic assignment community, it is quite standard to assume that, instead of travel times, flow particles minimize a \emph{generalized cost function}
which might include tolls, rebates or other desiderata such as emissions (cf.~\cite{FrieszHanPedro13}). For such models, the synchronization property required for thin flow
formulations breaks down and, thus, the equivalence between 
edge- and path-based equilibrium definitions for dynamic $\source$,$\dest$-flows is unclear.  
Our flow decomposition theorem allows us 
to fill this gap and we show in \Cref{sec:EquivalenceOfDE}
that edge- and walk-based equilibrium definitions for dynamic $\source$,$\dest$-flows are in fact equivalent, even for generalized cost functions and general physical flow models (cf.~\Cref{thm:def-equilibrium-equivalent}). 

Another application of our decomposition theorems can be found in the recent paper~\cite{GHS24TollSODA}, where it is used to derive a strong duality result for an infinite dimensional optimization problem
which in turn is used to derive implementing tolls for dynamic equilibrium flows.

\subsection{Related Work}
After the initial work of Ford and Fulkerson~\cite{Ford62},
several papers considered dynamic flow optimization problems such as the maximum flow over time problem (see Anderson and Philpott~\cite{ContinuousMaxFlowMinCutEarlier}, Philpott~\cite{Philpott90}, Fleischer and Tardos~\cite{FleischerT98}, Koch and Nasrabadi~\cite{KochN14} and Koch et al.~\cite{KochNS11}), the earliest arrival flow problem (see Gale~\cite{Gale59}), the quickest transshipment problem (see Hoppe and Tardos~\cite{HoppeT00}, Schloter et al.~\cite{SchloterST22} and Skutella~\cite{Skutella23} for an introduction into the topic) and minimum cost dynamic flows (Klinz and Woeginger~\cite{KlinzW04}).
A common characteristic of the above works is a simplified network loading model, i.e., they assume constant travel times.

The only reference we are aware of addressing the flow
decomposition problem for dynamic flows with
general flow-dependent transit times is the PhD thesis
by Ronald Koch~\cite{KochThesis}, who also explicitly mentioned the lack of literature on this topic (\cite[page~113]{KochThesis}): ``Unfortunately, it seems that there is no contribution addressing dynamic flow decomposition so far.''
In~\cite[Chapter~3.5]{KochThesis}, he 
formally defined the flow decomposition problem under a fairly general dynamic flow model.
He gave an example (\cite[Example~3.48]{KochThesis}) of an infinite time horizon, non-integrable dynamic edge flow  which does not admit a decomposition into walk inflows and circulations.
This counter example, however, crucially uses the fact that the cumulative edge inflow  for some edges is unbounded,  
leaving the existence of a solution to the  decomposition problem for the more realistic case of finite  cumulative edge inflows  unaffected.
In~\cite[Chapter~3.5, page~94]{KochThesis}, he then sketched the natural algorithm for finding a flow decomposition (similar to \Cref{alg: FlowDecompositionPseudo} above) which consists of 
first subtracting  dynamic circulations
with zero transit time from the input edge flow and then
iteratively subtracting walk inflows from the remaining dynamic edge flow. 
However, he gave no proof of correctness for this algorithm and only provided some intuition
for why the algorithm should work under the hypothesis that
the input edge flow vector admits a decomposition. Quoting from~\cite[Chapter 3.5, page 94]{KochThesis}: ``As already mentioned, a flow decomposition of an edge flow over time may not exist. However, the Flow Decomposition algorithm converges to a flow decomposition if the underlying edge flow over time is decomposable. For observing this, we only give a proof idea which is strongly based on intuition.'' 
His intuition is built on the key invariant that once a path inflow function is subtracted, the resulting reduced  edge flow is still decomposable.
His (intuitive) explanation why  this invariant should be correct uses the hypothesis that the initial edge dynamic flow is decomposable. However, this starting assumption (the underlying edge flow over time is decomposable) is the  key open question which we answer in this paper.

\section{The Model}\label{sec:Model}
\subsection{Network}

We consider single-source, single-\sink{} networks given by a directed graph $G=(\GV,\GA)$ with nodes $\GV$ and edges $\GA \subseteq V \times V$, a source node $s \in \GV$ and a destination node $\dest \in \GV$. We always assume $\GA$ to not contain any loops (i.e.\ edges $(v,v)$) and for $\source$ to be different to~$\dest$. We denote by $\hat{\Routes}$ the countable set of (finite) \stwalk s in~$G$. Here, an \stwalk~$\wa$ is a tuple 
of edges $\wa = (\arc_1,\ldots,\arc_k) \in \hat{\Routes}$ with $\arc_j=(v_j,v_{j+1})\in \GA$ for all $j\in [k]:=\{1,\ldots,k\}$ for some $(v_j)_{j\in[k+1]}\in \GV^{k+1}$. We denote by $\abs{\wa} \in \N_0$ the length (=number of edges) of~$\wa$ and use $\wa[j] \coloneqq \arc_j$ to refer to the $j$-th edge on walk~$\wa$. For a node $v \in V$ or an edge~$\arc \in \GA$ we write $v\in \wa$ or $\arc \in \wa$ to say that this node/edge lies on that walk, i.e.\ that there exist  some $j \in [k]$ and $\hat{v},v'\in \GV$ with $(\hat{v},v') = \wa[j]$ and $v \in \{\hat{v},v'\}$ or that there exists some $j \in [k]$ with $\arc=\wa[j]$, respectively. 
By $\edgesFrom{v}$ we denote the set of edges leaving a node $v$ and by $\edgesTo{v}$ the set of edges entering $v$. 
A walk $\wa$ is called simple, if it does not visit any node twice except possibly the starting node, i.e.,~for all $v\in \GV$ there exists at most one $j \in [\abs{\wa}]$ with $\wa[j] \in\delta^+(v)$.
Furthermore, we call a walk $c=(\gamma_1,\ldots,\gamma_m)$ a cycle if $\gamma_1 \in \edgesFrom{v}$ and $\gamma_m \in \edgesTo{v}$ for some node  $v \in \GV$. 
We denote by $\mathcal{C}$ the finite set of simple cycles. 
For a walk $\wa$ and $j\leq |\wa|$, we denote by  $\wa_{\geq j}$ and $\wa_{>j}$ the sub-walks of $\wa$ starting with $\wa[j]$ and $\wa[j+1]$, respectively. Analogously, we define $\wa_{\leq j}$ and $\wa_{<j}$. 
We denote by $\edgesFrom{\hat{\wa}}$ the edges leaving a walk, i.e.~$\edgesFrom{\hat{\wa}}\coloneq \Set{(v,v') \in \GA\mid v\in \hat{\wa},v'\notin\hat{\wa}}$. 
Furthermore, for two walks $\wa^1 = (\arc_1^1,\ldots,\arc_{k_1}^1),\wa^2=(\arc_1^2,\ldots,\arc^2_{k_2})$
with $\wa^1$ ending in a node $v$ and $\wa_2$ starting in it, 
we write $(\wa^1,\wa^2):=(\arc_1^1,\ldots,\arc^1_{k_1},\arc^2_1,\ldots,\arc^2_{k_2})$.

We consider $\hori$  as our planning horizon   during which flow particles can traverse the network. 
Since dynamic flows will be described by Lebesgue-integrable functions on~$\hori$, we equip $\hori$ with its Borel $\sigma$-algebra $\mathcal{B}(\hori)$. 
We denote by $\sigma$ the Lebesgue measure on $\hori$ and by $L(\hori)$ the space of ($\sigma$-equivalence classes of) $\sigma$-integrable real-valued functions over~$\hori$ equipped with the standard norm induced topology and the partial order induced by $L_+(\hori)$, i.e.\ the subsets of nonnegative integrable functions. Analogously, we denote by $L^\infty(\hori)$ the space of essentially bounded real-valued functions over~$\hori$ with the partial order induced by $L_+^\infty(\hori)$.
For any countable set~$M$, we denote by $\seql[1][M][L(\hori)]$ the set of 
vectors $(h_m)_{m \in M} \in L(\hori)^M$ whose sum $\sum_{m \in M}h_m \in L(\hori)$ is well-defined and exists, i.e.
\begin{align*}
    \seql[1][M][L(\hori)] &:= \big\{ h \in L(\hori)^{M} \mid \norm{h} := \sum_{m \in M} \norm{h_m} < \infty \big\}.
\end{align*}
 This defines again a Banach space (cf.~\cite[Section 16.11]{guide2006infinite}) with topological dual
 \begin{align*}
     \seql[\infty][M][L^\infty(\hori)]:= \big\{ f \in L^\infty(\hori)^M \mid \norm{f} := \sup_{m \in M} \norm{f_m}_\infty < \infty \big\}. 
\end{align*} 
We denote the bilinear form between this dual pair by $\dup{f}{h} \coloneqq \sum_{m \in M}\int_\hori f_m\cdot h_m\di\sigma$ for $f\in \seql[\infty][M][L^\infty(\hori)],h\in\seql[1][M][L(\hori)]$.  Here, we use $\int_\hori f\di\sigma$ to denote the integral of $f$ over~$\hori$ with respect to the Lebesgue measure~$\sigma$.\footnote{We use this notation instead of writing $\int_\hori f(t)\di t$ to stay consistent with the proofs of some of the more technical lemmas where we also have to consider integrals with respect to other measures.} 
Similarly, we denote for any two vectors of nonnegative measurable functions  (not necessarily contained in $L^\infty(\hori)^M$ or  $L(\hori)^M$)    $\tilde{f},\tilde{h}:\hori \to \R^M_+$ the sum over their integrals via $\dup{\tilde f}{\tilde h} \coloneqq \sum_{m \in M}\int_\hori \tilde f_m\cdot \tilde h_m\di\sigma \in \R_+ \cup\{\infty\}$. 
We say that a sequence $(h_n)_{n \in \N}$ converges weakly in $\seql[1][M][L(\hori)]$ to some $h \in \seql[1][M][L(\hori)]$ and write $h_n \wto h$, if $\dup{f}{h_n} \to \dup{f}{h}$ for any $f\in\seql[\infty][M][L^\infty(\hori)]$. 
Analogously, we define $\seql[1][M][L(\hori)^\GA]$ where we use $\norm{\g} := \sum_{\arc \in \GA}\norm{\g_\arc}$ for $\g \in L(\hori)^\GA$.

\subsection{Dynamic Flows}

The main concept underlying dynamic flows are the traversal time functions:  

\myparagraph{Traversal time functions} In our model, any vector of edge inflow rates $\g\in L_+(\hori)^\GA$ induces    corresponding nonnegative   edge traversal time functions
$\trav_\arc(\g,\cdot),\arc \in \GA$ 
 with 
$\trav_\arc(\g,t) \in \R_+$ denoting the time needed to traverse $\arc$ 
when entering the latter at time~$t$. 
For any such edge traversal time function, we also define two related functions: 
Firstly, we introduce
edge exit time functions $\exit
_\arc(\g,t):= t + \trav_\arc(\g,t)$ denoting the time a particle exits edge $\arc$ when entering at $t$. 
Secondly, we define edge arrival time functions $\arr_{\wa,j}(\g,\cdot)$ denoting the time a particle arrives at the tail of the $j$-th edge of some walk~$\wa$ when entering $\wa$ at time $t$. More precisely, for any walk $\wa$ we define $\arr_{\wa,1}(\g,\cdot):= \id$ and then, recursively, $\arr_{\wa,j}(\g,\cdot):= \exit_{\wa[{j-1}]}(\g,\cdot) \circ\arr_{\wa,j-1}(\g,\cdot)$ for $j\in\{2,\ldots,|\wa| + 1\}$. Here,  $\arr_{\wa,|\wa|+1}(\g,\cdot)$ is  interpreted as 
 the arrival time at the end node of the walk.

We assume that $\trav_\arc(\g,\cdot)$ is (locally) absolutely continuous\footnote{We omit from now on the term ``locally'' and call a function $\trav:\hori \to \R$ absolutely continuous  if it is absolutely continuous on every closed interval $[a,b]\subseteq \hori$ in the sense of \cite[Definition 5.3.1]{Bogachev2007I}.  We remark that the main properties of absolutely continuous functions carry directly over to the locally absolutely continuous ones. All properties we require are gathered and proved in \Cref{lem: PropAbsCon}.}  and adheres to the first-in first-out principle (FIFO), that is, $\exit_\arc(\g,\cdot)$ is a monotonic increasing function. Note, that this also implies that both $\exit_\arc(\g,\cdot)$ and $\arr_{\wa,j}(\g,\cdot)$ are absolutely continuous as well (\Cref{lem: PropAbsCon:Conca}).   

Observe, that both the Vickrey queuing model and the linear edge delay model, two prominent flow propagation models, fit into our framework. Here, let us remark that both models usually only use $\R_+$ as planning horizon. However, any flow propagation model that is defined only on a subinterval $[t_0,t_f]$ of $\hori$ can be extended to a model on the whole $\hori$ by simply extending the flows by $0$ on $\hori \setminus [t_0,t_f]$ and the travel times constantly on $\hori \setminus [t_0,t_f]$, i.e.~$\trav(\g,t) = \trav(\g,t_0)$ for $t<  t_0$ and $\trav(\g,t) = \trav(\g,t_f)$ for $t > t_f$. 

\begin{example}\label{ex:VickreyModel}
    In both the  Vickrey queuing and the linear edge delay model each edge~$\arc \in \GA$ comes with a free flow travel time~$\tau_\arc > 0$ and a service rate~$\nu_\arc > 0$. The traversal time function~$\trav_\arc(\g,\cdot):\R_+\to\R_+$ for an edge inflow rate vector $g \in L_+(\hori_+)^\GA$ is then defined as the unique solution to a system of equations in terms of the corresponding edges flows:

    For linear edge delays, this system is
        \[\trav_\arc(\g,t) = \tau_\arc + \frac{x_\arc(\g,t)}{\nu_\arc} \quad\text{ and }\quad x_\arc(\g, t) = \int_{[0,t]} \g_\arc\di\sigma-\int_{\exit_\arc(\g,\cdot)^{-1}([0,t])} \g_\arc \di\sigma,\]
    together with the condition that $x_\arc(\g,t)$ is always nonnegative. Here, $x_\arc(\g,t)$ represents the flow volume on edge~$\arc$ at time~$t$. For the Vickrey queuing model, it is
         \[\trav_\arc(\g,t) = \tau_\arc + \frac{q_\arc(\g,t)}{\nu_e} \quad\text{ and }\quad q_\arc(\g, t) = \int_{[0,t]} \g_\arc\di\sigma-\int_{\exit_\arc(\g,\cdot)^{-1}([0,t+\tau_\arc])} \g_\arc \di\sigma,\]
    together with the conditions that the queue is always nonnegative and that the derivative of $t \mapsto \int_{\exit_\arc(\g,\cdot)^{-1}([0,t])} \g_\arc \di\sigma$ (i.e.\ the outflow rate of edge~$\arc$) is bounded by $\nu_\arc$ almost everywhere. Here $q_\arc(\g,t)$ represents the flow volume in the queue of edge~$\arc$ at time~$t$.

    For both models it is known that for any edge flow~$\g$, 
    there exists a unique traversal time function $\trav_\arc(\g,\cdot)$ solving the corresponding system (cf.\ e.g.\ \cite{CominettiCL15,ZhuM00}).  Moreover, this travel time function satisfies the assumptions from above.
\end{example}

With this, we can now formally describe dynamic flows. We will use two types of these flows: 
Edge flows and walk flows: 

\myparagraph{Walk Flows}
For a countable collection of (not necessarily $\source$,$\dest$-)walks $\Routes'$, a 
  \emph{walk flow} or \emph{walk-inflow function} is a vector $h\in L_+(  \hori)^{{\Routes}'}$   
with $h_{{\wa}}(t)$ representing the walk inflow rate at time $t\in \hori$ into the walk ${\wa}\in  {\Routes'}$.
 
\myparagraph{Edge Flows} 
An edge flow is any vector $\g \in L_+(\hori)^\GA$ where $\g_\arc(t)$ then denotes the inflow rate into edge~$\arc$ at time~$t$. For any such edge flow we then define the concept of node balances which is a measure associating with every measurable subset of times the difference between cumulative outflow from and inflow into a node during these times. In particular, for an interval $(-\infty,t]$ with a corresponding negative node balance, its absolute value can then be interpreted as the amount of flow particles waiting at that node at time~$t$.
\begin{definition}\label{def: FlowBalaSDFlow}
	The \emph{node balance} (\wrt $\trav(\cdot,\cdot)$) at node $v \in \GV$ for an arbitrary edge inflow vector $\g\in L_+(\hori)^\GA$ is given by the measure  $\op[\trav(\cdot,\cdot)]_v\g$, which assigns to every $\mathfrak T \in \mathcal{B}(\hori)$ the value
	\begin{align}\label{eq: FlowBalance}
		\op[\trav(\cdot,\cdot)]_v \g(\mathfrak T) \coloneq  
		\sum_{\arc \in \delta^+(v)} \int_{\mathfrak T} \g_\arc \di \leb -  \sum_{\arc \in \delta^-(v)} \int_{\exit_\arc(\g,\cdot)^{-1}(\mathfrak T)} \g_\arc \di \leb. 
	\end{align}
	If the Radon-Nikodym derivative ${\inflow}_v$ of $\op_v \g$ exists, i.e.~a function $\inflow_v \in L(\R)$ satisfying for all $\mathfrak T \in \mathcal{B}(\hori)$
	\begin{align}\label{eq: FlowBalanceDerivative}
		\int_{\mathfrak T} {\inflow}_v\di\leb =  
		\sum_{\arc \in \delta^+(v)} \int_{\mathfrak T} \g_\arc \di \leb -  \sum_{\arc \in \delta^-(v)} \int_{\exit_\arc(\g,\cdot)^{-1}(\mathfrak T)} \g_\arc \di \leb,
	\end{align}
	we say that $\g$ has the  net (node)  outflow rate ${\inflow}_v$ at $v$, or equivalently, the  net inflow rate $-{\inflow}_v$. 
	If the latter is equal to zero almost everywhere, we 
	say that $\g$ fulfills \emph{flow conservation at~$v$} (\wrt $\trav(\cdot,\cdot)$). 
\end{definition}

We then call an edge flow $\g \in L_+(\hori)^\GA$ an (edge) \emph{$\source$,$\dest$-flow} (\wrt $\trav(\cdot,\cdot)$), if it has a net outflow rate $\inflow_s \in L_+(\hori)$ at $\source$, fulfills flow conservation at 
	all $v \neq \source,\dest$ and has a nonpositive node balance at $\dest$, i.e., \eqref{eq: FlowBalanceDerivative}. 
	If $\g$ fulfills flow conservation at all nodes, we refer to $\g$ as \emph{dynamic circulation} (\wrt $(\trav(\cdot,\cdot))$).

\subsection{Flow Decomposition} 
In the following we formally define (pure) flow decomposition.

\begin{definition} \label{def: FlowDecomp}
    Let $\g \in L_+(\hori)^\GA$ be an edge $\source$,$\dest$-flow. We call a walk-inflow function $h \in L_+(\hori)^{\Routes'}$ for $\Routes'= \hat{\Routes}$ a \emph{pure flow decomposition} for~$\g$ if the following holds for all $\arc \in \GA$: 
  \begin{align}\label{eq: DefFlowDecomp}
    \int_{\startint t]} \g_\arc\di\sigma = \sum_{\wa \in \Routes'}\sum_{j: \wa[j] = \arc}\int_{\arr_{\wa,j}(\g,\cdot)^{-1}({\startint}t])}
h_\wa\di\sigma \text{ for all } t\in \hori. 
\end{align}  
 If the latter statement holds for $\Routes'= \hat{\Routes}\cup \mathcal{C}$ and all $h_c$ with $c \in \mathcal{C}$ are zero-cycle inflow rates (\wrt $\trav(\g,\cdot)$), we simply speak of a \emph{flow decomposition} of~$\g$.  Here, we call an inflow rate $h_c$ into a cycle~$c$ a  zero-cycle inflow rate if it fulfills the implication  $h_c(t)>0 \implies \trav_\arc(\g,t)=0$ for all $\arc \in c$ and almost all times $t \in \hori$. 
    \end{definition}
    In regard of~\eqref{eq: DefFlowDecomp}, remark that in case of convergence of  $\sum_{k \in \N}\sum_{j: \wa^k[j] = \arc}\int_{\arr_{\wa^k,j}(\g,\cdot)^{-1}({\startint}t])}
h_{\wa^k}\di\sigma$ for an arbitrary ordering $\Routes'=\{\wa^k\}_{k \in \N}$, the latter is absolute convergent by the nonnegativity of $h_{\wa^k}$ and hence the order of the summands does not matter, that is, writing $\sum_{\wa \in \Routes'}$ is well-defined.

This definition leads to the following natural question(s):

\begin{framed}\centering
    Which edge $\source$,$\dest$-flows $\g$  admit a  (pure) flow decomposition?
 \end{framed}

Our main theorems give complete answers to this.

\begin{theorem}\label{thm: FLowDecompModel}
    Every edge $\source$,$\dest$-flow  $\g \in L_+(\hori)^\GA$ admits a flow decomposition.
\end{theorem}
\begin{theorem}\label{thm: PureFLowDecompModel}
    An edge $\source$,$\dest$-flow  $\g \in L_+(\hori)^\GA$ with an outflow rate $\inflow_\dest$ at $\dest$  admits a pure flow decomposition if and only if is satisfies the following property: 
    For every zero-cycle inflow rate $h_c'\in L_+(\hori)$ into \emph{any} (not necessary simple)  cycle $c$ with $\sum_{j\leq \abs{c}:c[j] = \arc}h_c' \leq \eflow_\arc,\arc \in c$, we have for almost all $t \in \hori$ with $h_c'(t)>0$ that (at least) one of the following conditions is satisfied:
    \begin{thmparts}
        \item  $\dest \in c$ and  $ \inflow_\dest (t)<0$.  
        \item  There exists an edge $\arc\in \edgesFrom{c}$ with $\eflow_\arc(t)>0$.  
    \end{thmparts}
\end{theorem}

As mentioned earlier, these theorems follow from the analogous \Cref{thm: FlowDecomp,thm: PureFlowDecompIntuitive} for so-called \auto flows. 
We formally introduce this key concept in the following section, along with several key properties we require for our main results.

 \section{\Auto Network Loadings}\label{sec:uBasedNetworkLoadings}

Let us consider an edge $\source$,$\dest$-flow $\g$ for which we want to compute a flow decomposition via \Cref{alg: FlowDecompositionPseudo}. 
In order to formulate the latter mathematically precise, we require the concept of \emph{\auto network loadings}. 
These are edge flows
${\tilde{\g}} \in L_+(\hori)^\GA$  which can be induced by a walk flow vector $\tilde{h} \in L_+(\hori)^{\Routes'}$ 
for some  countable collection of walks $\Routes'$ under the \emph{fixed} traversal time functions $\trav(\g,\cdot)$,  i.e.\ flows fulfilling for all $\arc \in \GA$:
\begin{align}\label{eq: DefUEdgeFlowU}
    \int_{\startint t]} \tilde{\g}_\arc\di\sigma = \sum_{\wa\in \Routes'} \sum_{j: \wa[j] = \arc}\int_{\arr_{\wa,j}(\g,\cdot)^{-1}({\startint}t])}\tilde{h}_\wa\di\sigma \text{ for all } t \in \hori. 
\end{align}
We investigate this concept of \auto network loadings in an even more general form, that is, we fix an arbitrary, flow-independent but time-dependent travel time function $\trav:\hori \to \R_+^\GA, t \mapsto (\trav_\arc(t))_{\arc \in \GA}$ and analyze the existence and properties of 
edge flows ${\g} \in L_+(\hori)^\GA$ which can be induced 
by a walk flow vector $h \in L_+(\hori)^{\Routes'}$ sending into all walks $\wa \in \Routes'$  flow $h_\wa$  under the \emph{fixed} traversal time functions $\trav$ for an arbitrary countable collection of walks $\Routes'$, i.e.~flows fulfilling for all $\arc \in \GA$: 
\begin{align}\label{eq: DefUEdgeFlow}
    \int_{\startint t]} \g_\arc\di\sigma =\sum_{\wa\in \Routes'} \sum_{j: \wa[j] = \arc}\int_{\arr_{\wa,j}^{-1}({\startint}t])}h_\wa\di\sigma \text{ for all } t \in \hori. 
\end{align}
Here, we denote by $\exit_\arc$ and $\arr_{\wa,j}$ the (flow independent) exit and arrival time function corresponding to $\trav$ and 
 merely assume that  $\trav_\arc$ is absolutely continuous and fulfills FIFO. Note that, as above, this implies that
 $\exit_\arc$ and $\arr_{\wa,j}$ are absolutely continuous as well.

In the above situation \eqref{eq: DefUEdgeFlow}, 
we write $\ell_{\Routes'}(h) =  (\ell_{\Routes',\arc}(h))_{\arc \in \GA} := (\g_\arc)_{\arc \in \GA}$, say that $\ell_{\Routes'}(h)$ is well-defined, exists and that the walk inflow rate $h$ induces~$\g$ under $\trav$ or that $g=\ell_{\Routes'}(h)$ is the \auto network loading of~$h$ under~$\trav$.

Analogous to \Cref{def: FlowDecomp}, we say that $h\in L_+(\hori)^{\Routes'}$ for $\Routes' \in \{\hat\Routes,\hat\Routes\cup \mathcal{C}\}$ is \aauto (pure) flow decomposition  of a $\g \in L_+(\hori)^\GA$ if 
the \auto  network loading  $\ell_{\Routes'}(h)$ of $h$ exists and coincides with $\g$. 
In this regard, let us state a first important  observation for \auto network loadings relating \auto flow decompositions to (non-\auto[)] ones:
\begin{lemma}\label{lem: NonParaIsPara}
    Consider an edge $\source$,$\dest$-flow $\g \in L_+(\hori)^\GA$ and a walk inflow $h\in L_+(\hori)^{\Routes'}$ for $\Routes' \in \{\hat\Routes,\hat\Routes\cup \mathcal{C}\}$. Then, $h$ is a   (pure) flow decomposition  of $\g$ if and only if $h$ is \aauto (pure) flow decomposition  of $\g$ \wrt the fixed travel times $\trav := \trav(\g,\cdot)$.
\end{lemma}
\begin{proof}
We prove both directions separately: 
\begin{structuredproof}
    \proofitem{``$\Leftarrow$''} If $h$ is \aauto (pure) flow decomposition  of $\g$ \wrt $\trav := \trav(\g,\cdot)$, then $\ell_{\Routes'}(h) = \g$, that is, \eqref{eq: DefUEdgeFlow} holds for $\trav := \trav(\g,\cdot)$ which is exactly \eqref{eq: DefFlowDecomp}.

    \proofitem{``$\Rightarrow$''} If $h$ is a (pure) flow decomposition  of $\g$, then \eqref{eq: DefFlowDecomp} holds. For $\trav := \trav(\g,\cdot)$, this implies by definition that $\g$ is the \auto network loading to $h$ \wrt $\trav$, i.e.~$\ell_{\Routes'}(h) = \g$ and, thus, $h$ is the \auto (pure) flow decomposition  of $\g$ \wrt $\trav$. \qedhere 
\end{structuredproof}
\end{proof}

In the following, we  briefly state the main structural properties on \auto network loadings needed for our flow decomposition results in the subsequent sections. 
The proofs, together with more technical properties of network loadings, can be found in \Cref{sec: ANLResultsFlowDecomp}.

\paragraph{Existence of \Auto Network Loadings} 

As already mentioned in the introduction, not every walk inflow rate necessarily induces an edge flow \wrt arbitrary fixed traversal times $\trav$ that can be described via a vector ${\g} \in L_+(\hori)^\GA$. This is true even in case of traversal times $\trav$ being induced by some other flow in a well known flow propagation model like the Vickrey model, cf.~\Cref{exa: noarcflow}. 
The following \namecref{lem: elluExistenceProperties} determines for which inflow rates $h$ the vector $\ell_{\Routes'}(h)$ exists and gives a complete characterization. 
To this end, we introduce 
for any walk $\wa$, $j \in [|\wa|]$ and $h_\wa \in L_+(\hori)$ 
the \auto flow induced by $h_{\wa}$ on the  $j$-th edge of $\wa$ 
(if it exists) as the function 
$\ell_{\wa,j}(h_\wa) \in L_+(\hori)$ satisfying
\begin{align*}
    \int_{\startint t]} \ell_{\wa,j}(h_\wa) \di\sigma=    \int_{\arr_{\wa,j}^{-1}(\startint t])}h_\wa\di\sigma \quad \quad \text{ for all }t \in \hori.
\end{align*}
In the following, we denote by $\edom{\Routes'}$ for any collection of walks $\Routes'$ 
the maximal domain of the corresponding \auto network loading operator. 
Equivalently to writing $h \in \edom{\Routes'}$, we will also say that $\ell_{\Routes'}(h)$ exists.  
We adopt the analogous notation and convention for the \auto network loading operator $\ell_{\wa,j}$ and its maximal domain $\edom{\wa,j}$.
In \Cref{lem: elluExistenceProperties} we provide a characterization of those walk flows that induce \aauto edge flow, alongside with several topological properties of the \auto network loading operator. Here, we only state the parts of \Cref{lem: elluExistenceProperties} concerning the existence characterization:

\begin{parttheorem}\label{lem: elluExistencePropertiesShort}
	 Consider \aauto travel time operator $\trav(\cdot)$, an arbitrary countable collection of walks $\Routes'$, $h \in L_+(\hori)^{\Routes'}$,  $\wa\in \Routes'$, $j \in[|\wa|+1]$ and $\arc \in \GA$. 
	Then, the following holds: 
	\begin{itemize}

        \item[\ref{lem: elluExistenceProperties:ExistenceInducedFlowOnJthEdge}] $h_\wa \in \edom{\wa,j}$  if and only if $h_\wa$ satisfies 
        \begin{align*}
			h_\wa = 0 \text{ on } \arr_{\wa,j}^{-1}(\mathfrak T) \text{ for every Borel-measurable $\leb$-null set }\mathfrak T \in \mathcal{B}(\hori).
		\end{align*}
       In this case $\ell_{\wa,j}(h_\wa)$ is uniquely determined.                 
		\item[\ref{lem: elluExistenceProperties:ExistenceInducedFlowOnAllEdges}] $h_\wa \in \edom{\wa}[]$ if and only if  $h_\wa \in \edom{\wa,j}$ for all $j\leq \abs{\wa}$. In this case, $\ell_{\wa}(h_\wa)$ is uniquely determined by $\ell_{\wa,\arc}(h_\wa) = \sum_{j:\wa[j]=\arc}\ell_{\wa,j}(h_\wa)$ for all $\arc \in \GA$.
		
		\item[\ref{lem: elluExistenceProperties:ExistenceInducedFlow}] $h\in \edom{\Routes'}$ if and only if  $h_\wa \in \edom{\wa}$  for all $\wa \in \Routes'$ and $(\Nl[]_{\wa}(h_\wa))_{\wa\in \Routes'} \in \seql[1][\Routes'][L_+(\hori)^\GA]$ holds.  In this case, $\Nl[]_{\Routes'}(h)$ is uniquely determined by $\Nl[]_{\Routes'}(h) = \sum_{\wa \in \Routes'}\Nl[]_\wa(h_\wa)$. 
	\end{itemize}
	
\end{parttheorem}

\paragraph{Optimization Problems Involving \Auto Network Loadings}

With the existence characterization of \auto network loadings at hand, we are in the position to formulate 
the optimization problem needed in \Cref{alg: FlowDecompositionPseudo} and show the existence of optimal solutions under suitable assumptions.  In fact, we 
do this for 
a whole class of  optimization problems involving \auto network loadings which will contain the aforementioned problem as a special case. 
We consider general optimization problems of the following form: 

\begin{align} 
	\max_{h}\;    &\objfunc(h)  \tag{P} \label{opt: GeneralOverview} \\
	\text{s.t.: } &\Nl[]_{\Routes'}(h) \leq \g  \label{eq: ExOptSolLeqOverview}\\
	&h \in \ofeas \nonumber
\end{align}
Here, 
$\Routes'$ denotes an arbitrary countable collection of walks which may contain individual walks multiple but at most finitely many times. 
The constraint vector  $\g$ is an arbitrary  element in $L_+(\hori)^\GA$. 
The objective~$\objfunc$ is some real-valued function on $\ofeas$,  which, in turn, is some subset of $\edom{\Routes'}$ containing at least one $h$ fulfilling~\eqref{eq: ExOptSolLeqOverview}, i.e.~the set of feasible solutions 
$\FeasSol\coloneq \{h \in \ofeas\mid \Nl[]_{\Routes'}(h)\leq \g\}$
is non-empty. 
Remark that 
\eqref{opt: GeneralOverview} is well-defined as $\ofeas \subseteq \edom{\Routes'}$ ensures that $\Nl[]_{\Routes'}(h)$ is well-defined. 
In the following theorem, $\FeasSol$ and 
$\edom{\Routes'}$ are equipped with the subspace topology induced by $\seql[1][\Routes'][L_+(\hori)]$.\footnote{The inclusion $\edom{\Routes'} \subseteq \seql[1][\Routes'][L_+(\hori)]$ is proven in \Cref{lem: elluContinuity:Subset}.}

 \begin{theorem}\label{thm: ExistenceOptSol}
 	Assume that $\objfunc:\FeasSol\to \R$  is sequentially weakly upper semi-continuous and $\ofeas$ is sequentially weakly closed in~$\edom{\Routes'}[]$.   Then, 
 	the optimization problem~\eqref{opt: General} has an optimal solution.
 \end{theorem}

 For the proof of \Cref{thm: ExistenceOptSol}, we aim to apply a Weierstrass-type of argument, requiring us to show that the set of feasible solutions is weakly compact in $\seql[1][\Routes'][L_+(\hori)]$. This is not a trivial task as 
 $\ell_{\Routes'}$ is not weakly continuous in general (cf.~\Cref{exa: DiscontinuityNL}), and hence, we first have to establish in \Cref{lem: elluContinuity} that $\ell_{\Routes'}$ satisfies a suitable form of upper semicontinuity.

\paragraph{\Auto Node Balances} 
We introduce the concept of the \auto node balance, i.e.\ the node balance of an arbitrary vector $\g \in L_+(\hori)^\GA$ as defined in \Cref{def: FlowBalaSDFlow} \wrt the fixed traversal time function~$\trav$.
With this, we can formally define \auto (edge) $\source$,$\dest$-flows as those  $\eflow \in L_+(\hori)^\GA$ who have a nonnegative \auto net outflow rate at $\source$, \aauto net outflow rate of $0$ at all $v \neq \source,\dest$ (i.e.\ satisfy \auto flow conservation) and a nonpositive \auto node balance at $\dest$. This concept and the following insights play a key role in the formalization of \Cref{alg: FlowDecompositionPseudo} 
as a crucial aspect of the latter is   that any appearing $\g^k$ during the execution of the algorithm and, in particular, their limit is \aauto $\source$,$\dest$-flow.

\begin{lemma}
    \label{lem: flowconW'}
    	Consider an arbitrary countable collection of walks $\Routes'$, a corresponding walk inflow rate vector $h \in \edom{\Routes'}$ with   $\g\coloneq \ell_{\Routes'}(h)$  and a node $v \in \GV$. 
	Then we have for all $\mathfrak T \in \mathcal{B}(\hori)$:
	\begin{align*} 
		\op_v \g(\mathfrak T) =  \sum_{\wa \in \Routes'_{v+}}\int_{\mathfrak T}  h_\wa  \di\leb  - \sum_{\wa \in \Routes'_{v-}}\int_{\arr_{\wa,|\wa|+1}^{-1}(\mathfrak T)}  h_\wa \di\leb ,
	\end{align*}
	where $\Routes'_{v+}$ denotes the set of walks in $\Routes'$ starting at $v$ while $\Routes'_{v-}$ denotes the set of walks in $\Routes'$ ending in  $v$. 
	
	If, additionally, $\ell_{\wa,\abs{\wa}+1}(h_\wa)$ exist for all $\wa \in \Routes'_{v-}$, we even  have for all $\mathfrak T \in \mathcal{B}(\hori)$:
	\begin{align*}
		\op_v \g(\mathfrak T) = \sum_{\wa \in \Routes'_{v+}}\int_{\mathfrak T} h_\wa \di\leb - \sum_{\wa \in \Routes'_{v-}} \int_{\mathfrak T} \ell_{\wa,\abs{\wa}+1}(h_\wa) \di\leb.
	\end{align*} 
\end{lemma}

\paragraph{Properties of \Auto $\source$,$\dest$-Flows}

The two main ingredients 
of the proof of the flow decomposition theorem  state that \aauto $\source$,$\dest$-flow  has either a positive net outflow rate at $\source$ and  admits 
a  flow-carrying \stwalk{} (\Cref{lem: ExistenceOfFlowCarryingWalk}), or, is a dynamic circulation (\Cref{lem: FlowConEveryNode}) and can be decomposed into zero-cycle inflow rates (\Cref{lem: ZeroCycleDecomposition}). 
From this, the correctness of  \Cref{alg: FlowDecompositionPseudo}  follows as the limit of the \auto $\source$,$\dest$-flows  $\g^k,k\in \N$ can not admit any flow-carrying \stwalk~$\wa_k$ due to the maximality of the corresponding  $h_{\wa_k}$. 

\begin{theorem}\label{lem: ZeroCycleDecomposition}
	Any  \aauto dynamic circulation $\eflow$ can be decomposed into zero-cycle inflow rates ${h}_c \in L_+(\hori), c \in \SimpCyc$ via $\eflow_\arc = \sum_{c \in \SimpCyc} \ell_{c,\arc}({h}_c) = \sum_{c \in \SimpCyc:\arc \in c}  {h}_c$ for all $\arc \in \GA$. 
\end{theorem}

We show this \namecref{lem: ZeroCycleDecomposition} by applying \Cref{lem: FlowConEveryNode} from which it follows that the  edge travel times are (almost) always zero when there is inflow into an edge. 
This reduces the problem to the static decomposition problem.

\begin{theorem}\label{lem: ExistenceOfFlowCarryingWalk}
Let $\eflow \in L_+(\hori)^\GA$ be an \auto edge $\source$,$\dest$-flow with a positive node outflow rate ${\inflow}_s\in L_+(\hori) \setminus\{0\}$ at $s$. 
    Then, there exist 
    an \stwalk{} $\wa \in \hat{\Routes}$ and a walk inflow rate $\wflow_\wa \in \edom{\wa}\setminus\{0\}$ with $\ell_{\wa}(\wflow_\wa) \leq \eflow$, $\wflow_\wa \leq {\inflow}_s$ as well as $\op_\dest\ell_\wa(\wflow_\wa) \geq \op_\dest \eflow$.
\end{theorem}
For this \namecref{lem: ExistenceOfFlowCarryingWalk}, we refer to the introduction for a proof sketch.

\section{Flow Decomposition}\label{sec:FlowDecomp}

We now come to our main decomposition question and first consider the case of general flow decompositions  into \stwalk s and zero-cycles. Afterwards, we turn to our characterization of those edge flows that even admit a pure flow decomposition.

\subsection{General Flow Decomposition}\label{sec:FlowDecomp:General}

In this subsection, we show the promised flow decomposition \Cref{thm: FLowDecompModel}. In fact, we will prove the analogous version (\Cref{thm: FlowDecomp}) for \auto flows, stating that any \auto $\source$,$\dest$-flow has \aauto flow decomposition. 
We remark again that any (non-\auto[)] edge $\source$,$\dest$-flow $\g$ is in particular \aauto   $\source$,$\dest$-flow \wrt $\trav(\g,\cdot)$ and, hence, \Cref{thm: FLowDecompModel} follows immediately from \Cref{thm: FlowDecomp} using \Cref{lem: NonParaIsPara}.    
 Note that proving the \auto version \Cref{thm: FlowDecomp} of \Cref{thm: FLowDecompModel} does not add any additional layer of complexity to the proof. This is because from the second iteration onwards the flow decomposition algorithm used for the proof has to compute \aauto flow decomposition of \aauto $\source$,$\dest$-flow anyway (namely, of $\g^2$  \wrt $\trav(\g,\cdot)$).

\begin{algorithm}[H]
	\caption{Flow Decomposition Algorithm}\label{alg: FlowDecomposition}
	\Input{\Aauto $\source$,$\dest$-flow  $\g\in L_+(\hori)^\GA$  with  node outflow rate ${\inflow}_s\in L_+(\hori)$ at $s$}
	\Output{Walk inflow rates $\wflow \in L_+(\hori)^{\hat{\Routes}}$  such that $ \eflow-\sum_{\wa \in \hat{\Routes}}\ell_\wa(\wflow_\wa)$ is nonnegative and fulfills flow conservation at all nodes}
			
		fix some order on the set of all \stwalk s $\hat{\Routes} = \{\wa_k\}_{k \in \N}$ \label{line: enumerate}
  
        set $\eflow^1 \leftarrow \eflow$, $\inflow_s^1 \leftarrow \inflow_s$
        
        \For{$k \in \N$}{
        
        find an optimal solution $\wflow_{\wa_k}$ of
    \begin{align} 
        \max\; & \int_\hori\wflow_{\wa_k}\di\sigma  \tag{$\mathrm{FD}^k$}\label{opt: FlowDecomp}\\
        \text{s.t.: }&\ell_{\wa_k}(\wflow_{\wa_k}) \leq \eflow^k \label{ineq: FD} \\
        &\wflow_{\wa_k} \leq \inflow^k_s \label{ineq: FD2}\\
        &\op_\dest \big(\ell_{\wa_k}(\wflow_{\wa_k})\big) \geq \op_\dest \eflow^k  \label{ineq: FD3}\\
                    &\wflow_{\wa_k} \in \edom{\wa_k} \nonumber
    \end{align}
        
        $\eflow^{k+1} \leftarrow \eflow^k - \ell_{\wa_k}(\wflow_{\wa_k})$ 
       
         $\inflow^{k+1}_s \leftarrow \inflow^k_s- \wflow_{\wa_k}$

        }

    \KwRet{$\wflow_{\wa_k}, k \in \N$}
\end{algorithm}

\begin{theorem}\label{thm: FlowDecomp}
    Every \auto $\source$,$\dest$-flow has \aauto flow decomposition.
\end{theorem}

The proof of this \namecref{thm: FlowDecomp} mainly consists of showing the correctness of \Cref{alg: FlowDecomposition}, that is, showing that after countably many steps the algorithm returns a walk flow $\wflow$ such that the difference between $\eflow$ and the  edge flow induced by~$\wflow$ under $\trav$ is an  edge flow satisfying flow conservation at all nodes. The \namecref{thm: FlowDecomp} then follows by applying \Cref{lem: ZeroCycleDecomposition} which allows us to decompose the remaining flow into  zero-cycle inflows.

\begin{proof}
We split the main part of the proof into two separate steps, showing the well-definedness and correctness of \Cref{alg: FlowDecomposition} separately:
\begin{structuredproof}
    \proofitem{Well-definedness}
        We show that \eqref{opt: FlowDecomp} always has an optimal solution by verifying that \Cref{thm: ExistenceOptSol} is applicable here: The objective is clearly (sequentially) weakly continuous. Furthermore, the set 
 \begin{align*}
     \ofeas:= \big\{ \wflow_{\wa_k} \in \edom{\wa_k}\mid \wflow_{\wa_k} \leq \inflow^k_s, \op_\dest \big(\ell_{\wa_k}(\wflow_{\wa_k})\big) \geq \op_\dest \eflow^k  \big\}
 \end{align*}
 is  sequentially weakly closed in $\edom{\wa_k}$. 
To see this, consider a weakly converging sequence  $(\wflow_{\wa_k}^n)_{n\in \N}$ contained in   $\ofeas$  with 
$\wflow_{\wa_k}^n \wto \wflow_{\wa_k}^* \in \edom{{\{\wa_k\}}}$. Then, we have for arbitrary $\mathfrak T \in \mathcal{B}(\hori)$ 
\begin{align*}
   \int_{\mathfrak T}\inflow^k_s\di\sigma \geq   \int_{\mathfrak T}\wflow_{\wa_k}^n\di\sigma \to \int_{\mathfrak T}\wflow_{\wa_k}^*\di\sigma 
\end{align*}
showing that $\wflow_{\wa_k}^*\leq \inflow^k_s$. 
Similarly, we get by \Cref{lem: elluContinuity} that $\ell_{\wa_k}(\wflow_{\wa_k}^n)\wto \ell_{\wa_k}(\wflow_{\wa_k}^*)$ and hence for an arbitrary $\mathfrak T \in \mathcal{B}(\hori)$:
\begin{align*}
         \op_\dest \eflow^k (\mathfrak T) &\leq \op_\dest\big(\ell_{\wa_k}(\wflow_{\wa_k}^n)\big) (\mathfrak T) \eqperdef \sum_{\arc \in \delta^+(\dest)} \int_{\mathfrak T} \ell_{\wa_k,\arc}(\wflow_{\wa_k}^n) \di \sigma -  \sum_{\arc \in \delta^-(\dest)} \int_{\exit_\arc^{-1}(\mathfrak T)}\ell_{\wa_k,\arc}(\wflow_{\wa_k}^n)\di \sigma  \\
      &\to \sum_{\arc \in \delta^+(\dest)} \int_{\mathfrak T} \ell_{\wa_k,\arc}(\wflow_{\wa_k}^*) \di \sigma -  \sum_{\arc \in \delta^-(\dest)} \int_{\exit_\arc^{-1}(\mathfrak T)}\ell_{\wa_k,\arc}(\wflow_{\wa_k}^*)\di \sigma \defpereq  \op_\dest\big(\ell_{\wa_k}(\wflow_{\wa_k}^*)\big) (\mathfrak T)
\end{align*}
which shows that $\op_\dest \big(\ell_{\wa_k}(\wflow_{\wa_k}^*)\big) \geq \op_\dest \eflow^k  $.

Finally, the problem always has a feasible solution given by $0$ since we always have $\inflow^k_s \in L_+(\hori), \op_\dest \eflow^k \leq 0$ as well as $\eflow^k\in L_+(\hori)^\GA$ for all $k \in \N$. All these properties follow by a straight forward induction over~$k$ and the feasibility of $\wflow_{\wa_{k-1}}$ for~(\hyperref[opt: FlowDecomp]{$\mathrm{FD}^{k-1}$})  (i.e.~the fulfillment of the inequalities in~\eqref{ineq: FD}, \eqref{ineq: FD2} and~\eqref{ineq: FD3}). 

Therefore, \Cref{thm: ExistenceOptSol} ensures that \eqref{opt: FlowDecomp} always has an optimal solution and, consequently, the algorithm is well-defined.
        \proofitem{Correctness} 
We have to show that $\sum_{\wa \in \hat{\Routes}}\ell_\wa(h_\wa)$ is well-defined, i.e.\ converges absolutely, and that $\eflow-\sum_{\wa \in \hat{\Routes}}\ell_\wa(h_\wa)$ is nonnegative and fulfills flow conservation at all nodes. 

For the well-definedness, we observe that the series $\sum_{k \in \N}\ell_{\wa_k}(\wflow_{\wa_k})$ is bounded by~$0$ and~$\eflow$. The lower bound holds since all $h_{\wa_k}$ are nonnegative and, subsequently, so are $\ell_{\wa_k}(\wflow_{\wa_k})$ by \Cref{lem: ellOrderPreserving}. The upper bound follows from the fact that we have for any $k^*\in \N$ that 
$\eflow-\sum_{k < k^*}\ell_{\wa_k}(\wflow_{\wa_k})= \eflow^{k^*} \in L_+(\hori)^\GA$ and, hence, 
$ \sum_{k < k^*}\ell_{\wa_k}(\wflow_{\wa_k})\leq \eflow$. 
Thus, the pointwise limit of the series exists almost everywhere and, by Lebesgue's dominated convergence theorem,  the series itself converges (absolutely) in $L_+(\hori)^\GA$. Moreover, we get
\begin{align*}
    \eflow^* \coloneqq \lim_{k^*\to \infty}\eflow^{k^*}= \lim_{k^*\to \infty}\eflow - \sum_{k < k^*}\ell_{\wa_k}(\wflow_{\wa_k}) =  \eflow - \sum_{k \in \N}\ell_{\wa_k}(\wflow_{\wa_k}) \in L_+(\hori)^\GA,
\end{align*} 
which, in particular, also implies that $\eflow^* = \eflow-\sum_{\wa \in \hat{\Routes}}\ell_\wa(h_\wa)$ is nonnegative.

Hence, it only remains to show flow conservation at all nodes. 
By \Cref{lem: flowconW'}, 
the series $\sum_{k \in \N}\ell_{\wa_k}(\wflow_{\wa_k})$ fulfills flow conservation at all $v\neq s,\dest$ and, subsequently, so does $\eflow^*$. Next, we  
make the following observation: From the above, 
we can deduce that the net outflow rate $\inflow_s^*$ of $\eflow^*$ at $s$ is nonnegative and the outflow 
$\op_\dest\eflow^*$  is nonpositive. This is due to the fact that by  \Cref{lem: flowconW'}  the net outflow rate of the series $\sum_{k \in \N}\ell_{\wa_k}(\wflow_{\wa_k})$  at $s$ equals $\sum_{k \in \N}\wflow_{\wa_k}$ while the net outflow    at $\dest$ equals $\sum_{k \in \N} \op_\dest \big(\ell_{\wa_k}(\wflow_{\wa_k})\big)$ and we have $\sum_{k \leq k^*}\wflow_{\wa_k} \leq \inflow_s$ by $\inflow^{k^*+1}_s\in L_+(\hori)$ and 
$\sum_{k \leq k^*}  \op_\dest \big(\ell_{\wa_k}(\wflow_{\wa_k})\big) \geq \op_\dest\eflow$ by $\op_\dest \eflow^{k^*+1} \leq 0$  
for all $k^* \in \N$. Therefore, $\eflow^*$ is an \auto edge $\source$,$\dest$-flow. 

Now assume for the sake of a contradiction that flow conservation does not hold at $\source$, i.e.~$\inflow^*_s \in L_+(\hori)\setminus\{0\}$. 
By \Cref{lem: ExistenceOfFlowCarryingWalk}, there exists a $\tilde{k} \in \N$ and  $\Tilde{\wflow}_{\wa_{\tilde{k}}}\in L_+(\hori)\setminus\{0\}$ with $\ell_{\wa_{\tilde{k}}}(\Tilde{\wflow}_{\wa_{\tilde{k}}}) \leq \eflow^*$ and $\Tilde{\wflow}_{\wa_{\tilde{k}}} \leq \inflow_\source^*$ as well as 
$\op_\dest\ell_{\wa_{\tilde{k}}}(\Tilde{\wflow}_{\wa_{\tilde{k}}})  \geq \op_\dest \eflow^*$.
Since $\eflow^*\leq \eflow^{{\tilde{k}}+1}= \eflow^{\tilde{k}} - \ell_{\wa_{\tilde{k}}}(\wflow_{\wa_{\tilde{k}}})$, $\inflow^*_s\leq \inflow^{{\tilde{k}}+1}_s = \inflow^{{\tilde{k}}}_s - \wflow_{\wa_{\tilde{k}}} $ and $ \op_\dest \eflow^* \geq   \op_\dest \eflow^{\tilde{k}+1} = \op_\dest \eflow^{\tilde{k}} - \op_\dest\ell_{\wa_{\tilde{k}}}(\Tilde{\wflow}_{\wa_{\tilde{k}}}) $, 
 the sum $\wflow_{\wa_{\tilde k}} + \tilde{\wflow}_{\wa_{\tilde k}}$ is feasible for~(\hyperref[opt: FlowDecomp]{$\mathrm{FD}^{\tilde k}$}), contradicting the optimality of $\wflow_{\wa_{\tilde{k}}}$. 

Thus, $\eflow^*$ fulfills flow conservation at all nodes $v\neq \dest$ and hence, by \Cref{lem: FlowConEveryNode}, at $\dest$ as well, which concludes the correctness of the algorithm. 
\end{structuredproof}
With this, we can now apply \Cref{lem: ZeroCycleDecomposition} to $\eflow^*= \eflow- \sum_{\wa \in \hat{\Routes}}h_{\wa}$ and get zero-cycle inflow rates $h_c, c \in \mathcal{C}$ such that $\sum_{c \in \mathcal{C}:\arc \in c}h_c = \eflow^*$. By combining these zero-cycle inflow rates $(h_c) \in L_+(\hori)^\mathcal{C}$ with the \stwalk{} inflow rates $(h_{\wa_k})_{k \in \N} \in L_+(\hori)^{\hat{\Routes}}$ produced by the algorithm we obtain a walk inflow vector~$h \in L_+(\hori)^{\hat{\Routes}\cup\mathcal{C}}$ which is a flow decomposition of~$\eflow$.
\end{proof}

\begin{remark}[Finite Execution of \Cref{alg: FlowDecomposition}]
    In case that  $\g$ has finite support and all travel times are lower bounded by some $\tmin > 0$ on that support,
    we can adjust 
    \Cref{alg: FlowDecomposition} suitably 
    such that it terminates after a finite number of steps:  
    By enumerating the \stwalk s in \Cref{line: enumerate} in an ascending order with respect to their number of edges (i.e.~$k< k' \implies \abs{\wa_k} \leq \abs{\wa_{k'}}$), we ensure that there exists some $k^* \in \N$ such that any particle entering a walk $\wa_k,k\geq k^*$ during the support of~$\g$ would enter the last edge of~$\wa_k$ only after the support of~$\g$ ended. This results in \eqref{opt: FlowDecomp} for $k\geq k^*$ to have $\wflow_{\wa_k}=0$ as the only feasible solution. Hence, \Cref{alg: FlowDecomposition} can be stopped at the $k^*$-th iteration.  
\end{remark}

\subsection{Pure Flow Decomposition }\label{sec:FlowDecomp:Pure}

In this section, we investigate the question of when \aauto $\source$,$\dest$-flow admits a pure \auto $\source$,$\dest$-flow decomposition. 
Remark again that, by \Cref{lem: NonParaIsPara}, this then also answers the same question for non-\auto flows (\Cref{thm: PureFLowDecompModel}) as every (non-\auto[)] edge $\source$,$\dest$-flow $\g$ is in particular \aauto $\source$,$\dest$-flow \wrt $\trav(\g,\cdot)$. 
We will now prove the following:

\begin{theorem}\label{thm: PureFlowDecompIntuitive}
    \Aauto $\source$,$\dest$-flow $\eflow \in L_+(\hori)^\GA$ with net outflow rate $\inflow_\dest$ at $\dest$ has \aauto pure flow decomposition   if and only if the following property holds: 
    For every zero-cycle inflow rate $h_c'\in L_+(\hori)$ into \emph{any} (not necessary simple) cycle $c$ with $\ell_c(h'_c) \leq \eflow$, we have for almost all $t \in \hori$ with $h_c'(t)>0$ that (at least) one of the following conditions is satisfied:
    \begin{thmparts}
        \item  $\dest \in c$ and  $ \inflow_\dest (t)<0$. \label[thmpart]{thm: PureFlowDecompIntuitive: Dest}
        \item  There exists an edge $\arc=(v,v') \in \edgesFrom{c}$ with  $\eflow_\arc(t)>0$. \label[thmpart]{thm: PureFlowDecompIntuitive: NotDest}
    \end{thmparts}
\end{theorem}

Intuitively, the above two conditions \ref{thm: PureFlowDecompIntuitive: Dest} and \ref{thm: PureFlowDecompIntuitive: NotDest} are necessary and sufficient conditions for the zero-cycle inflow $h_c$ to not be disconnected from the remaining flow of $\g$. It is clear that \ref{thm: PureFlowDecompIntuitive: NotDest} ensures the connectedness. For 
\ref{thm: PureFlowDecompIntuitive: Dest}, note that in case of its fulfillment, the positive \emph{net} inflow of $\dest$ implies that 
there has to arrive flow from $\g$ that does not belong to $h_c$ as the latter has no impact on the net inflow of $\dest$. 
We remark that it does not suffice to only consider simple zero-cycles in \Cref{thm: PureFlowDecompIntuitive} as the following 
 example shows: 

\begin{figure}
    \centering
    \BigPicture[0]{
    \begin{tikzpicture}
        \node[namedVertex](s)at(0,0){$\source$};
        \node[namedVertex](d)at(9,0){$\dest$};
        \node[namedVertex](v1)at(3,0){$v_1$};
        \node[namedVertex](v2)at(6,0){$v_2$};
        \node[namedVertex](v3)at(4.5,2){$v_3$};
        \node[namedVertex](v4)at(4.5,-2){$v_4$};

        \draw[edge](s)--node[above]{$0$}(v1);
        \draw[edge](v1)--node[above]{$2\cdot\CharF[[0,1]]$}(v2);
        \draw[edge](v2)--node[above]{$0$}(d);
        \draw[edge](v2)--node[anchor=south west]{$\CharF[[0,1]]$}(v3);
        \draw[edge](v3)--node[anchor=south east]{$\CharF[[0,1]]$}(v1);
        \draw[edge](v2)--node[anchor=north west]{$\CharF[[0,1]]$}(v4);
        \draw[edge](v4)--node[anchor=north east]{$\CharF[[0,1]]$}(v1);
    \end{tikzpicture}
    }
    \caption{An edge $\source$,$\dest$-flow that has no pure flow decomposition even though every \emph{simple} zero-cycle has a flow-carrying outgoing edge. All edges in this network have a constant traversal time of zero. The edge-labels denote the edge inflow rates of the given flow.}
    \label{fig:PureFlowDecompNonSimpleNec}
\end{figure}

\begin{example}
    In \Cref{fig:PureFlowDecompNonSimpleNec} we given an example for an edge $\source$,$\dest$-flow wherein every simple zero-cycle (namely the cycles $((v_1,v_2),(v_2,v_3),(v_3,v_1))$ and $((v_1,v_2),(v_2,v_4),(v_4,v_1))$ during the interval $[0,1]$) has an outgoing flow-carrying edge ($(v_2,v_4)$ and $(v_2,v_3)$, respectively). Nevertheless, the given flow clearly has no pure flow decomposition (as there is never any outflow from the source).
\end{example}

Instead of proving \Cref{thm: PureFlowDecompIntuitive} directly, we will first show with \Cref{thm: PureFlowDecomp} another characterization using connected components of zero-cycles from which the above \namecref{thm: PureFlowDecompIntuitive} will then follow almost immediately.

In order to formally state this \namecref{thm: PureFlowDecomp}, we require some additional terminology: 
Consider a set of zero-cycle inflow rates $h_c,c \in \mathcal{C}$ and fix some (arbitrary) representatives of those. We define 
for all $t \in \hori$ the set $\mathcal{C}(t):=\{c \in \mathcal{C}\mid  {h}_c(t)>0 \text{ and }\trav_{\arc}(t) = 0, \arc \in c\}$. 
Let  $C_1^t,\ldots,C_{m(t)}^t$   for  a $\n2(t) \in \N$ be the 
partition of $\mathcal{C}(t)$ into maximal connected components, that is, a partition with the following two properties:
\begin{itemize}
    \item For every $j \in\{1,\ldots,m(t)\}$ the edge set $\GA_{C_j^t}:=\set{\arc \in \GA | \exists c \in C_j^t: e \in c}$ induces a connected subgraph of~$G$ and 
    \item the node sets $\GV_{C_j^t}:=\set{v \in \GV | \exists c \in C_j^t: v \in c}$ are disjoint, i.e.\ $\GV_{C_j^t} \cap \GV_{C_{j'}^t} = \emptyset$ for all $j'\neq j$. 
\end{itemize} 
Next, we consider the set $\{C \subseteq \mathcal{C} \mid \sigma(\mathfrak T_C)>0\} \subseteq 2^\mathcal{C}$ where $\mathfrak T_{C} := \{t \in \hori \mid  \exists\, j: C = C_j^t \}$. We denote this set via $\{C_{\n1}\}_{\n1 \in \capn1}$ 
where   $\capn1\subseteq \N$ is a finite family of indices.  Note that the sets $\mathfrak T_{C}$  are measurable as they can be written as follows: 
    $\mathfrak T_{C} = \bigcap_{c \in C}\mathfrak T_{c} \cap \bigcap_{c \in \bar{C}\setminus C}(\hori \setminus\mathfrak T_c)$ where
    $\mathfrak T_c := \{t \in \hori \mid {h}_c(t)>0 \text{ and }\trav_{\arc}(t) = 0, \arc \in c\}$ and $\Bar{C} = \{c \in \mathcal{C}\mid \exists \,c' \in C: c \text{ shares a node with }c'\}$ with $\mathfrak T_c$ being measurable due to ${h}$ and $\trav$ being measurable. Furthermore, we denote for any $\n1 \in \capn1$ by $\GV_{C_{\n1}}:=\{v \in \GV\mid\exists c \in C_{\n1}: v \in c \}$ the nodes contained in $C_{\n1}$ and analogously by  $\GA_{C_{\n1}}:=\{\arc \in \GA\mid\exists c \in C_{\n1}: \arc \in c \}$ the edges  contained in $C_{\n1}$.
    We denote by $\edgesFrom{C_{\n1}}$ the edges leaving the connected component induced by $C_{\n1}$, i.e., $\edgesFrom{C_{\n1}}\coloneq \Set{(v,v') \in \GA\mid v \in\GV_{C_{\n1}}, (v,v')\notin  \GA_{C_{\n1}}}$.
    Finally, we remark that the above sets are uniquely determined (up to enumeration) by the support of the representative $h_c,c \in \mathcal{C}$. 
    
\begin{definition}\label{def: ConnectedComp}
    In the situation as described above, we call $\mathcal{C}(t)$ the set of active cycles at $t \in \hori$,    $\{C_{\n1}\}_{\n1 \in \capn1} $ the resulting  connected components and   $\mathfrak T_{C_{\n1}},\n1\in\capn1$ the set of times at which the connected components are active. Similarly, $\mathfrak T_{c}$ for any $c \in \mathcal{C}$ is the set of times at which the cycle~$c$ is active.
\end{definition}
 With this notation at hand, we can state in the following another characterization of edge flows with flow decompositions purely into \stwalk s, from which \Cref{thm: PureFlowDecompIntuitive} will follow almost immediately.  
 The proof of the following theorem is given after the proof of \Cref{thm: PureFlowDecompIntuitive}.

 \begin{theorem}\label{thm: PureFlowDecomp}
    Consider \aauto $\source$,$\dest$-flow $\eflow \in L_+(\hori)^\GA$ with a corresponding 
    flow decomposition $h_\wa,\wa \in  {\hat{\Routes}},h_c,c \in \mathcal{C}$, an outflow rate $\inflow_\dest$ and an arbitrary representative of $h$ together with the  sets defined in \Cref{def: ConnectedComp}. 
    Then  $\eflow$ has \aauto flow decomposition   purely into \stwalk s if and only if for every $\n1 \in \capn1$ and almost all $t \in \mathfrak T_{C_{\n1}}$ (at least) one of the following statements is true 
\begin{thmparts}
    \item $\dest \in \GV_{C_{\n1}}$ and  $\inflow_\dest (t)<0$. \label[thmpart]{thm: PureFlowDecomp: Dest} 
    \item there exists an edge $\arc=(v,v') \in \edgesFrom{C_{\n1}}$ with $\eflow_\arc(t)>0$. \label[thmpart]{thm: PureFlowDecomp: NotDest}
\end{thmparts}
\end{theorem}
 
With additional help from \Cref{thm: FlowDecomp}, we can now deduce \Cref{thm: PureFlowDecompIntuitive} from the above  \Cref{thm: PureFlowDecomp} as follows:
If the conditions of \Cref{thm: PureFlowDecompIntuitive} are satisfied, we can use, for any $\n1 \in \capn1$, a cycle~$c$ traversing all edges in $E_{C_n}$ and an inflow rate $h'_c$ which is positive on $\mathfrak T_{C_{\n1}}$ to deduce from this that the corresponding conditions of \Cref{thm: PureFlowDecomp} are satisfied as well. This then gives us the existence of a pure flow decomposition. For the other direction, we are given an edge flow $\eflow$ for which we know that there exists a pure flow decomposition and some zero-cycle inflow rate $h'_c$ with an induced edge flow upper bounded by~$\eflow$. After removing this induced flow on~$c$ from~$\eflow$, 
we can get a (non-pure) flow decomposition for the remaining flow from \Cref{thm: FlowDecomp} while, at the same time, decomposing the removed flow into zero-cycle inflows into simple cycles only. By recombining these two decomposition, we get an alternative decomposition of the given edge flow~$\eflow$ in which it is guaranteed that for every time where $h'_c$ is positive, all edges of~$c$ are part of some common component $C_{\n1}$. Since we know that $\eflow$ has a pure flow decomposition, we get from \Cref{thm: PureFlowDecomp}, that the two conditions of that \namecref{thm: PureFlowDecomp} must be satisfied for the constructed decomposition. From this, we can then deduce that the corresponding conditions in \Cref{thm: PureFlowDecompIntuitive} hold as well.

\begin{proof}[Proof of \Cref{thm: PureFlowDecompIntuitive}] 
 We show both directions separately:
     \begin{structuredproof}
        \proofitem{``$\Leftarrow$"} 
        Let $h$ be any flow decomposition of~$\g$, which exists by \Cref{thm: FlowDecomp}. Now, consider an arbitrary representative of~$h$ with the corresponding sets defined in \Cref{def: ConnectedComp}. We will verify that the conditions \ref{thm: PureFlowDecompIntuitive: Dest} and \ref{thm: PureFlowDecompIntuitive: NotDest} stated in \Cref{thm: PureFlowDecomp} are fulfilled. Let $ \n1 \in \capn1$ be arbitrary 
        and let $c$ be a cycle that contains every cycle in $C_{\n1}$ once. Moreover, consider a corresponding inflow rate $h_c'$ that fulfills $\ell_c(h_c') \leq \sum_{c \in C_{\n1}:\arc \in c}h_c \leq \eflow_\arc$ for all $\arc \in c$ and $h_c'(t)>0$ for a.e.~$t\in \mathfrak T_{C_{\n1}}$. Note that this is possible as $h_c(t)>0$ for a.e.~$t \in \mathfrak T_{C_{\n1}}$ and all $c \in C_{\n1}$. 
        The fulfillment of  \Cref{thm: PureFlowDecomp: Dest} or \Cref{thm: PureFlowDecomp: NotDest} is then a direct consequence of our starting assumption that  either \Cref{thm: PureFlowDecompIntuitive: Dest} or \Cref{thm: PureFlowDecompIntuitive: NotDest} is fulfilled. 
        Hence, we can apply \Cref{thm: PureFlowDecomp} to get the existence of a pure flow decomposition which finishes this direction of the proof. 
        \proofitem{``$\Rightarrow$"} 
         Consider a  zero-cycle inflow rate $h_c'\in L_+(\hori)$ with $\ell_c(h'_c) \leq \eflow$. Then, $\hat{\eflow}:= \eflow- \ell_c(h'_c)$ has a flow decomposition $\hat{h}$  by \Cref{thm: FlowDecomp}. Now, choose simple cycles $c_1,\dots,c_m \in \mathcal{C}$ together with multiplicities $n_1,\dots,n_m \in \N$ such that taken together they make up the cycle~$c$ in the sense that we have $\sum_{k: \arc \in c_k}n_k = \abs{\set{j \in [\abs{c}] | c[j]=\arc}}$ for every edge $\arc \in \GA$. Moreover, we define zero-cycle inflow rates $h'_{c_k} \coloneqq n_k \cdot h'_c$. Then, we clearly have 
            \[\sum_{j \leq \abs{c}:c[j]=\arc} h'_c = \sum_{k \leq m: \arc \in c_k}h'_{c_k} \text{ for all } \arc \in \GA\]
        as well as $h'_c(t) > 0 \Leftrightarrow (\forall k: h'_{c_k}(t)>0)$ for almost all $t \in \R$. We choose for the remainder of the proof representatives that fulfill this property for all $t \in \hori$. 
        By adding, for any $k\leq \n2$, to (an arbitrary representative of) $\hat{h}_{c_{k}}$ the respective (above representative of the) zero-cycle inflow rate ${h}'_{c_{k}}$, we arrive at a (representative of a) flow 
        decomposition~$h$ of~$\eflow$ with ${h}_{c_{k}}(t)\geq {h}'_{c_{k}}(t)$ for all $t\in\hori,k\leq \n2$. 
        The latter, in particular, fulfills ${h}_{c_{k}}(t)\geq h'_{c_{k}}(t)>0,k\leq \n2$ for all $t$ with $h'_c(t)>0$ by the choice of the representative of~$h'$. 
        Since there exists a pure flow decomposition of~$\eflow$, we can  apply \Cref{thm: PureFlowDecomp} \wrt the representative of $h$ we constructed. 
        Choose, then, a representative of $\eflow$ that fulfills $h_{\tilde{c}}(t)>0 \implies \eflow_\arc(t)>0$ for all $\arc \in {\tilde{c}}$ and ${\tilde{c}} \in \mathcal{C}$. 
 		For every~$t$ with $h_c'(t)>0$, we now find $\n1 \in \capn1$ with $c_{k} \in C_{\n1}$ for all $ k\leq \n2$ and  such that $t \in \mathfrak T_{C_{\n1}}$. 
		Hence, $\{t\in \hori \mid h_c'(t)>0 \} \subseteq \bigcup_{\n1 \in \capn1}\mathfrak T_{C_{\n1}}$ and, therefore, 
		for almost all $t$ with $h_c'(t)>0$, either \Cref{thm: PureFlowDecomp: Dest} or \Cref{thm: PureFlowDecomp: NotDest} is fulfilled for some $\n1 \in \capn1$, i.e.\  either $\dest \in \GV_{C_{\n1}}$ and $\inflow_\dest(t) <0$ or 
		there exists $\arc = (v,v')\in \edgesFrom{C_{\n1}}$ with $\eflow_\arc(t) >0$. 
		Depending on which condition of \Cref{thm: PureFlowDecomp: Dest} or \Cref{thm: PureFlowDecomp: NotDest} holds, 
		we argue in the following that one of the conditions in \Cref{thm: PureFlowDecompIntuitive} is fulfilled:
		\begin{structuredproof}
			\proofitem{\Cref{thm: PureFlowDecomp: Dest} holds} If $\dest \in c$, it follows immediately that \Cref{thm: PureFlowDecompIntuitive: Dest}
			is fulfilled.  If $\dest \notin c$, there must be at least one edge contained in $E_{C_{\n1}}$ but not in $c$. Since, by the definition of $C_{\n1}$, the graph induced by $C_{\n1}$ is connected and build only out of (simple) cycles, there must, in fact, be a different cycle $c' \in C_{\n1}$ with an edge $\arc'=(v,v') \in c'$ leaving $c$, i.e.~$\arc'\in\edgesFrom{c}$. Moreover, again by the definition of~$C_{\n1}$, there must be positive inflow $h_{c'}(t)>0$.          
			Hence, \Cref{thm: PureFlowDecompIntuitive: NotDest} is fulfilled for $\arc'$ as  $\g_{\arc'}(t)>0$ follows by the 
			choice of the representative of $\g$  together with $h_{c'}(t)>0$ and $\arc' \in c'$. 
			\proofitem{\Cref{thm: PureFlowDecomp: NotDest} holds} If $v \in c$, then \Cref{thm: PureFlowDecompIntuitive: NotDest} is fulfilled for $\arc$. If $v \notin c$, the exact same argumentation as for the case of \Cref{thm: PureFlowDecomp: Dest} being valid, applies, showing that \Cref{thm: PureFlowDecompIntuitive: NotDest} is fulfilled. 
		\end{structuredproof}
        Thus, the proof of this direction is finished as well.  \qedhere
     \end{structuredproof}
 \end{proof}


We now turn back to \Cref{thm: PureFlowDecomp} and its proof. 
We first note that \cite[Lemma~3.47]{KochThesis} claims a similar characterization for purely \stwalk-decomposable flows among all decomposable flows stating that an edge $\source$,$\dest$-flow with a walk decomposition also has a \emph{pure} walk decomposition if and only if every zero-cycle of the decomposition is \emph{directly} connected to an \stwalk{} used in the given flow decomposition at the same time. However, this characterization is incorrect as the stated condition is too strong and, thus, only yields a sufficient but not a necessary condition.
A simple example which does not satisfy this condition even though it has a decomposition purely into \stwalk s is given in \Cref{fig:KochPureFD-CounterExample}. 

\begin{figure}\centering
    \BigPicture[0]{%
        \begin{tikzpicture}
            \coordinate(s)at(0,0);
            \coordinate(v1)at(3,0);
            \coordinate(v2)at(3,2);
            \coordinate(v3)at(3,4);
            \coordinate(d)at(6,0);

            \draw[red,line width=2pt,->,rounded corners]($(s)+(0,-.1)$) -- ($(d)+(-.3,-.1)$);
            \draw[blue,line width=2pt,->,rounded corners]($(v1)+(-.1,.3)$) to[out=120,in=60,looseness=22] ($(v1)+(.1,.3)$);
            \draw[green,line width=2pt,->,rounded corners]($(v2)+(-.1,.3)$) to[out=120,in=60,looseness=22] ($(v2)+(.1,.3)$);

            \node[namedVertexF](s)at(s){$\source$};
            \node[namedVertexF](v1)at(v1){$v_1$};
            \node[namedVertexF](v2)at(v2){$v_2$};
            \node[namedVertexF](v3)at(v3){$v_3$};
            \node[namedVertexF](d)at(d){$\dest$};

            \draw[edge](s) --node[above]{$\CharF[[0,1]]$} (v1);
            \draw[edge](v1) --node[above]{$\CharF[[0,1]]$} (d);
            \draw[edge](v1) to[bend left=40]node[left]{$\CharF[[0,1]]$} (v2);
            \draw[edge](v2) to[bend left=40]node[left]{$\CharF[[0,1]]$} (v3);
            \draw[edge](v3) to[bend left=40]node[right]{$\CharF[[0,1]]$} (v2);
            \draw[edge](v2) to[bend left=40]node[right]{$\CharF[[0,1]]$} (v1);
        \end{tikzpicture}
    }
    \caption{A network with time and flow independent travel times of $0$ for all edges. The labels on the edges denote an $\source$,$\dest$-flow and the arrows provide a non-pure flow decomposition. Note that the edge flow clearly also has a pure flow decomposition even though the topmost zero-cycle (green) of the given flow decomposition is not directly connected to the only \stwalk{} (red) used in this decomposition.}\label{fig:KochPureFD-CounterExample}
\end{figure}

Before we come to the  proof of \Cref{thm: PureFlowDecomp}, let us give a brief proof sketch first:

The only if direction is quite straightforward and exploits the fact that for any connected component~$C_{\n1}$ and (a.e.) point in time $t \in \mathfrak T_{C_{\n1}}$, the flow induced on an edge contained in the component is induced by some \stwalk{} $\wa$ under the flow  decomposition   purely into \stwalk s. Tracking this flow along the walk until it leaves the component implies the fulfillment of \ref{thm: PureFlowDecomp: Dest} or \ref{thm: PureFlowDecomp: NotDest}. 

In contrast, the if direction is technically quite involved. We start by showing that for any connected component $C_{\n1}$ we can construct another flow decomposition with the same connected components and the additional property  that each cycle in $C_{\n1}$ is directly connected to a flow-carrying \stwalk and, thus, satisfies the stronger condition stated in \cite[Lemma~3.47]{Koch11}. From here, the idea is then to add sufficiently many copies of this cycle to the corresponding walk such that the flow requirement on the cycle is met. Note that Koch only shows how to do this for a single cycle, walk and point in time. To actually obtain a complete pure decomposition, we need to do this for all cycles, walks and times at once, which is technically much more involved.

\begin{proof}[Proof of \Cref{thm: PureFlowDecomp}]
  	We prove both directions separately. 
	\begin{structuredproof}
		\proofitem{``$\Rightarrow$''} 
		Consider the starting flow decomposition $h_\wa,\wa \in  {\sdRoutes},h_c,c \in \SimpCyc$ of $\g$ with the corresponding sets from \Cref{def: ConnectedComp}. 
		Let  $h'_\wa,\wa \in  {\sdRoutes},h'_c,c \in \SimpCyc$ with $h'_c = 0,c \in \SimpCyc$ be a pure flow decomposition of $\eflow$. 
		Assume for the sake of a contradiction that there exists $\n1\in \capn1$ and  $\mathfrak T^*\subseteq \mathfrak T_{C_{\n1}},\leb(\mathfrak T^*)>0$ such that for almost every $t \in \mathfrak T^*$ 
		neither~\ref{thm: PureFlowDecomp: Dest} nor~\ref{thm: PureFlowDecomp: NotDest} is fulfilled. 
		
		Consider an arbitrary $c \in C_{\n1}$ and $\arc \in c$. 
		By definition of $ C_{\n1}$, we have  
		\begin{align*}
			\sum_{\wa\in \sdRoutes}\sum_{j:\wa[j] = \arc}\ell_{\wa,j}(h'_\wa)(t) = \eflow_\arc(t) \geq h_c(t)>0 \quad \quad  \text{ for a.e.\ } t \in \mathfrak T_{C_{\n1}}.
		\end{align*}
		 Hence,  there has to exist a walk $\wa \in \sdRoutes$, $j \leq |\wa|$ with $\wa[j] = \arc$ and a measurable set $\mathfrak T \subseteq \mathfrak T^*, \leb(\mathfrak T)>0$ such that  $h_c(t)$ and $\ell_{\wa,j}(h'_\wa)(t)$ are bigger than $0$ for a.e.~$t \in \mathfrak T$.

		Now let $j'$ be the first index in $\{j+1,\ldots,|\wa|\}$ with $\wa[j'] \notin \GA_{C_{\n1}}$ (and subsequently $\wa[j'] \in \edgesFrom{C_{\n1}}$) in case such an index exists or set $j' := |\wa|+1$ otherwise. 
		Since we have 
		$\trav_{\tilde{\arc}}(t) = 0,\tilde{\arc} \in \GA_{C_{\n1}}$ for all $t\in \mathfrak T\subseteq \mathfrak T_{C_{\n1}}$, we get  
		by \Cref{lem: FLowOnZeroTrav} applied to the walk  $\wa_{\geq j}$, the set $\mathfrak D:= \mathfrak T$ and indices $j_1 := 1$ and $j_2 := j'-j+1$ that 
		\begin{align*} 
			 \ell_{\wa,j}(h_\wa') \cdot \Indi_{\mathfrak T} &\Croverset{lem: elluPropagation}{=} \ell_{\wa_{\geq j},1}(\ell_{\wa,j}(h_\wa'))\cdot \Indi_{\mathfrak T} \Croverset{lem: FLowOnZeroTrav}{=}  \ell_{\wa_{\geq j},j'-j+1}(\ell_{\wa,j}(h_\wa')\cdot \Indi_{\mathfrak T}) \\
			 &\Croverset{lem: elluindi}{=}  \ell_{\wa_{\geq j},j'-j+1}(\ell_{\wa,j}(h_\wa')) \cdot \Indi_{\arr_{\wa_{\geq j},j'-j+1}(\mathfrak T)}
			 \Croverset{lem: elluPropagation}{=} \ell_{\wa,j'}(h_\wa')  \cdot \Indi_{\mathfrak T}.
		\end{align*}
		  Remark for the above that $\arr_{\wa_{\geq j},j'-j+1}(t) = t$ for all $t \in \mathfrak T$ as $\trav_{\wa_{\geq j}[\tilde{j}]}(t) =0$ by $\wa_{\geq j}[\tilde{j}] = \wa[j+\tilde{j}-1] \in \GA_{C_{\n1}}$ for all $\tilde{j} < j'-j+1$ and all $t \in \mathfrak T$ by the choice of $j'$. 
		The above now shows that $ \ell_{\wa,j'}(h_\wa')(t)>0$ for almost all $t \in \mathfrak T$ since $\ell_{\wa,j}(h_\wa') (t)>0$ for almost all $t\in \mathfrak T$. From this, we derive now a contradiction depending on which value $j'$ attains:

		In the case that $j' \leq |\wa|$, we get a contradiction to~\ref{thm: PureFlowDecomp: NotDest}   not being fulfilled for almost all  $t \in\mathfrak T^*$,  since we have by the choice of $j'$ that  $\wa[j']  \in \edgesFrom{C_{\n1}}$  and 
		$0<\ell_{\wa,j'}(h_\wa')(t) \leq \eflow_{\wa[j']}(t)$ for a.e.~$t \in\mathfrak T$. 
		
		In the case of $j'= \abs{\wa}+1$, we get $\dest \in \GV_{C_{\n1}}$ since $\wa$ is an \stwalk. Then, by \Cref{lem: flowcon}, we get $\inflow_\dest \leq -\ell_{\wa,\abs{\wa}+1}(h_\wa') < 0$ on $\mathfrak T$, contradicting that~\ref{thm: PureFlowDecomp: Dest} is not fulfilled  for almost all  $t \in\mathfrak T^*$. 
		Remark that the existence of $\inflow_\dest$ implies the existence of $\ell_{\wa,\abs{\wa}+1}(h_\wa') $.

		\proofitem{``$\Leftarrow$''}
		We start by proving the following claim which will allow us to assume \wlg that any zero-cycle flow is connected to a flow carrying \stwalk{}  under the flow decomposition $h$ of $\g$. 
		
		\begin{claim}\label{claim: FlowDecompZshKomps}
			Fix a representative of $h$ with corresponding sets as defined in \Cref{def: ConnectedComp}. Then, for any $\n1\in \capn1$ we can construct another (representative of a) flow decomposition $h'_\wa,\wa \in  {\sdRoutes},h'_c,c \in \SimpCyc$ of $\eflow$ such that
			\begin{itemize}
				\item we can use the same sets to satisfy \Cref{def: ConnectedComp} for $h'$,
				\item for any $\n1' \in \capn1$ satisfying
				\begin{align}\label{eq:ZeroCyclesConnected}
					h_c(t) >0 \implies \sum_{\wa \in \sdRoutes}\ell_{\wa,\arc}(h_\wa) (t)>0 \quad \text{for almost all } t \in \mathfrak T_{C_{\n1'}}, \text{ all } \arc \in c \text{ and } c \in C_{\n1'},
				\end{align}
				the same still holds under $h'$ and
				\item under $h'$ the implication \eqref{eq:ZeroCyclesConnected} holds for $\n1$ as well.
			\end{itemize}    
		\end{claim}
		
		\begin{proofClaim} 
			Let $\n1 \in \capn1$ be arbitrary. 
			We construct in the following a countable set of walks $\wa^{l},l \in L$ with corresponding starting points $\mathfrak D^{l}\subseteq \hori$ with the property that flow is sent over $\wa^{l}$ during $\mathfrak D^{l}$ under $h$ and arrives at a node shared with $C_{\n1}$ during  $\mathfrak T_{C_{\n1}}$. Additionally, the union of these arrival times at a node of $C_{\n1}$ over all walks~$\wa^l$ is disjoint and equals $\mathfrak T_{C_{\n1}}$ up to a null set. By adding a cycle containing each edge in $C_{\n1}$ to these walks, we will then be able to construct  walk inflow rates that fulfill the desired condition for all $c \in C_{\n1}$. 
			
			\begin{subclaim}\label{subclaim: SetL}
				There exists a countable set $L\subseteq \N$ with corresponding walks $\wa^{l} \in \sdRoutes, j_{l}\leq \abs{\wa^{l}}+1$ and departure time sets $\mathfrak D^{l} \in \mathcal{B}(\hori)$ with $\leb(\mathfrak D^{l})> 0$ such that for all $l \in L$
				\begin{itemize}
					\item either $j_l = \abs{\wa^l}+1$ and $\dest \in \GV_{C_{\n1}}$ or  $\wa^{l}[j_{l}] \in \delta^+(v)$ for a node $v \in \GV_{C_{\n1}}$, 
					\item $h_{\wa^l}(t)>0$ for a.e.~$t \in \mathfrak D^l$, 
					\item the union  $\bigcup_{l' \in L}  \arr_{\wa^{l'},j_{l'}}(\mathfrak D^{l'})$ equals $\mathfrak T_{C_{\n1}}$ up to a null set,
					\item the union  $\bigcup_{l' \in L}  \arr_{\wa^{l'},j_{l'}}(\mathfrak D^{l'})$ is disjoint and
					\item $\arr_{\wa^{l},j_{l}}(\mathfrak D)$ has positive measure for all $\mathfrak D\subseteq \mathfrak D^{l}$ with $\leb(\mathfrak D)>0$.           
				\end{itemize} 
			\end{subclaim}
			\begin{proofClaim}    
				Consider an arbitrary ordering $\{\wa^k\}_{k \in \N}$ on the set $\sdRoutes$. 
				We aim to apply \Cref{lem: Relations:h>0u>0:u>0ExistsCountableMZwischen} \wrt $\mathfrak T_{C_n}$ and the set $\hat{L}\subseteq \N^2$ consisting 
				of all $(k,j) \in \N^2$ for which either $j= \abs{\wa^k}+1$ and $\dest \in \GV_{C_{\n1}}$ or 
				$\wa^k[j]=(v,v')$ fulfills $v \in \GV_{C_{\n1}}$. 
				Hence, we have to verify that for almost every $t \in \mathfrak T_{C_n}$ there exists an $l \in \hat{L}$ with corresponding $\tilde{t} \in \arr_{\wa^l,j_l}^{-1}(t)$ such that $h_\wa(\tilde{t})>0$. 
				
				We know that for almost all $t \in \mathfrak T_{C_{\n1}}$, either~\ref{thm: PureFlowDecomp: Dest} or \ref{thm: PureFlowDecomp: NotDest} holds.
				Hence, the union of the sets 
				\begin{align*}
					\mathfrak T_1 \coloneqq \set{t \in \mathfrak T_{C_{\n1}} | \text{\ref{thm: PureFlowDecomp: Dest} holds at } t} \text{ and } \mathfrak T_2^\arc \coloneqq \set{t \in \mathfrak T_{C_{\n1}} | \text{\ref{thm: PureFlowDecomp: NotDest} holds at } t \text{ for } \arc}, \arc \in \GA
				\end{align*}
				equals $\mathfrak T_{C_n}$ up to a null set. It is, thus, sufficient to show the required property for the individual sets separately:  
				\begin{structuredproof}
					\proofitem{$\mathfrak T_1$} 
					We know by \Cref{lem: Relations:h>0u>0:u>0ExistsHw>0D} that for $\op_\dest \g$-almost every 
					$t$, there exists $k\in \N$ and $\tilde{t} \in \arr_{\wa^k,\abs{\wa}+1}^{-1}(t)$ with $h_{\wa^k}(\tilde{t})>0$. Since $\op_\dest \g= \inflow_\dest\cdot \leb$ and $\inflow_\dest(t) < 0$ for almost every  $t\in\mathfrak T_1$, it follows that any $\op_\dest \g$-null subset  of $\mathfrak T_1$ is also a  $\leb$-null set. In particular, we get that for $\leb$-almost every 
					$t\in \mathfrak T_1$, there exists $k\in \N$ and $\tilde{t} \in \arr_{\wa^k,\abs{\wa}+1}^{-1}(t)$ with $h_{\wa^k}(\tilde{t})>0$. Since $\dest  \in \GV_{C_{\n1}}$ by \ref{thm: PureFlowDecomp: Dest} being satisfied for almost all $t \in \mathfrak T_1$, the required property holds on $\mathfrak T_1$. 
					
					\proofitem{$\mathfrak T_2^\arc$ for arbitrary $\arc \in \GA$} 
					This follows immediately by \Cref{lem: Relations:h>0u>0:u>0ExistsHw>0} since 
					$\g_\arc(t)>0$ for all $t \in \mathfrak{T}_2^\arc$.  \qedhere
				\end{structuredproof}
			\end{proofClaim}
			
			With the walks $\wa^l$, indices $j_l$ and starting times $\mathfrak D^l$ constructed in \Cref{subclaim: SetL} at hand, we can now take each of these walks and shift some inflow (during $\mathfrak D^l$) into $\wa^l$ to the walk that first follows $\wa^l$ until the $j^l$-th edge (where the walk intersects $C_{\n1}$), then traverses each of the simple cycles in $C_{\n1}$ exactly once and, finally, follows the second half of $\wa^l$ towards the destination. At the same time we reduce the inflow rates into the cycles contained in $C_{\n1}$ in such a way as to keep the induced edge flow the same. This way we obtain a new walk-decomposition $h'$ in which \eqref{eq:ZeroCyclesConnected} holds for $C_{\n1}$ as well, i.e.\ whenever $C_{\n1}$ is active, each of its edges is part of a used \stwalk.
			
			To formalize this, define for any $l \in L$ the walk 
			$\hat{\wa}^{l}:=(\wa^{l}_{< j_{l}},c_{\n1}^{l}, \wa^{l}_{\geq j_{l}}) $ where $c_{\n1}^{l}$ is a cycle composed of all cycles in $C_{\n1}$ (i.e.~an Eulerian circuit in the directed multi-graph containing each edge from $\GA$ as many times as there are cycles in $C_{\n1}$ that contain that edge) and staring with $\dest$ if $j_l = \abs{\wa^l}+1$ or starting with the tail of $\wa^{l}[j_{l}]$ otherwise.  Note that, by construction of $L$, in the former case we have $\dest \in \GV_{C_{\n1}}$ while in  the latter case the tail of $\wa^{l}[j_{l}]$ is contained in $\GV_{C_{\n1}}$. Hence,  $c_{\n1}^{l}$ exists and $\hat{\wa}^l$ is a well-defined \stwalk. 
			
			Next, take (and fix an arbitrary representative of) $\hat{h}^{l} \in L_+(\hori)$   with $\ell_{\wa^{l},j_{l}}(\hat{h}^{l}) = \min_{c \in C_{\n1}} h_c \cdot \Indi_{ \arr_{\wa^{l},j_{l}}(\mathfrak{D}^{l})}$. That is, $\hat{h}^l$ is the maximal amount of flow that can be sent into $\hat\wa^l$ such that the induced flow on any of the cycles $c \in C_{\n1}$ does not exceed the flow sent into it via $h_c$. Such a walk inflow $\hat h^l$ exists by~\Cref{lem: 1to1:h-f}. 
			We  observe that $\hat{h}^{l}(t)>0$ for almost every $t \in \mathfrak D^{l}$ since we have for arbitrary measurable non-null set $\mathfrak D \subseteq \mathfrak D^{l}$:
			\begin{align*}
				0\symoverset{1}{<}   &\int_{\arr_{\wa^{l},{j_{l}}}(\mathfrak D)}  \min_{c \in C_{\n1}} {h}_c \cdot \Indi_{ \arr_{\wa^{l},j_{l}}(\mathfrak{D}^{l})} \di\leb 
				= \int_{\arr_{\wa^{l},{j_{l}}}(\mathfrak D)}  \ell_{\wa^{l},j_{l}}(\hat{h}^{l}) \di\leb  \\
				= &\int_{ \arr_{\wa^{l},{j_{l}}}^{-1}(\arr_{\wa^{l},{j_{l}}}(\mathfrak D))}  \hat{h}^{l}  \di\leb \\
				\symoverset{2}{=} &\int_{\mathfrak D}   \hat{h}^{l}  \di\leb 
			\end{align*}
			Here, note for  the  inequality \refsym{1}   that   $\arr_{\wa^{l},{j_{l}}}(\mathfrak D)$  is not a null set by \Cref{subclaim: SetL}  and  $\min_{c \in C_{\n1}} {h}_c \cdot \Indi_{ \arr_{\wa^{l},j_{l}}(\mathfrak{D}^{l})}(t)>0$ for a.e.~$t\in \arr_{\wa^{l},{j_{l}}}(\mathfrak D) $ since  $\arr_{\wa^{l},{j_{l}}}(\mathfrak D^l)\subseteq \mathfrak T_{C_{\n1}}$. 
			The   equality \refsym{2} is due to \Cref{lem: elluinj}. Thus, we have $\hat{h}^{l}(t)>0$ for almost all $t \in \mathfrak D^{l}$.
			A similar argument shows that $\hat{h}^{l}(t)=0$ on $\hori\setminus\mathfrak D^{l}$.
			
			With this we can now define a new walk inflow~$h'$ by setting for all $\wa \in \sdRoutes$ and $c \in \SimpCyc$: 
			\begin{align}\label{eq:FlowDecompZshKomps:Defh'}
				h_\wa' := h_\wa + \sum_{l\in L: \hat{\wa}^{l} = \wa} \frac{1}{2^l} \rho^{l}  - \sum_{l\in L: {\wa}^{l} = \wa} \frac{1}{2^l} \rho^{l}  && \text{ and } &&  h_c' := h_c - \Indi_{  C_{\n1}}(c) \cdot \sum_{l \in L} \frac{1}{2^l} \ell_{\wa^l,j_l}(\rho^{l})
			\end{align}
			where  $\rho^{l}:=  \min\{0.5\cdot \hat{h}^{l},h_{\wa^{l}}\}$. 
			Here, intuitively, $\rho^l$ is the amount of flow we can safely shift from walk $\wa^l$ to walk~$\hat\wa^l$ without either creating negative walk inflow rates into one of the walks~$\wa^l$ or making one of the cycles $c \in C_{\n1}$ inactive (and, hence, changing the sets in \Cref{def: ConnectedComp}). The former is ensured by $\rho^l\leq {h}_{\wa ^l}$, the latter by $\rho^l\leq 0.5\cdot \hat{h}^l$. In the actual definition of~$h'$ we then shift even less flow than $\rho^l$ (namely, $2^{-l}\rho^l$) since it is possible that we have $\wa^l = \wa^{l'}$ for $l \neq l'$ and, consequently, we remove flow from the same walk multiple times.

			The following \namecref{subclaim: FlowDecompZshKomps: h'props} now states that the walk inflow~$h'$ satisfies all the desired properties. \Cref{claim: FlowDecompZshKomps} will then be an immediate consequence of this \namecref{subclaim: FlowDecompZshKomps: h'props}:
			\begin{subclaim} \label{subclaim: FlowDecompZshKomps: h'props}
				The walk inflow~$h'$ defined by~\eqref{eq:FlowDecompZshKomps:Defh'} fulfills the following:
				\begin{enumerate}[label = \roman*)]
					\item $h'$ is well-defined and $h'\in L_+(\hori)^{\sdRoutes\cup \SimpCyc}$. \label{subclaim: FlowDecompZshKomps: h'props: h'wd}
					\item For all $c \in \SimpCyc$ and almost all $t \in \hori$ we have $h_c(t)>0 \Leftrightarrow h_c'(t)>0$.  \label{subclaim: FlowDecompZshKomps: h'props: h_ch'_c} 
					\item $ \ell_\wa({h}'_\wa)$ exists for every ${\wa \in \sdRoutes}$. \label{subclaim: FlowDecompZshKomps: h'props: l(h')ex} 
					\item $\sum_{\wa \in \sdRoutes} \ell_{\wa,\arc}({h}'_\wa) \geq \sum_{\wa \in \sdRoutes} \ell_{\wa,\arc}({h}_\wa)$ for all $\arc \in \GA$.\label{subclaim: FlowDecompZshKomps: h'props: GeqPureFlow}
					\item $\sum_{\wa \in \sdRoutes} \ell_{\wa,\arc}({h}'_\wa) + \sum_{c \in \SimpCyc: \arc \in c}h_c' = \eflow_\arc$ for all $\arc \in \GA$. \label{subclaim: FlowDecompZshKomps: h'props: equalEflow}
					\item For all $c \in C_{\n1}$, $\arc \in c$ and almost all $t \in \mathfrak T_{C_{\n1}}$, we have $\sum_{\wa \in \sdRoutes}\ell_{\wa,\arc}(h_\wa')(t)>0$.  \label{subclaim: FlowDecompZshKomps: h'props: g'_c>0}
				\end{enumerate}
			\end{subclaim}
			\begin{proofClaim}
				\begin{structuredproof}
					\proofitem{\ref{subclaim: FlowDecompZshKomps: h'props: h'wd}} 
					We start by verifying for any $\wa \in \sdRoutes$ that 
					$\sum_{l\in L: \hat{\wa}^{l} = \wa} \frac{1}{2^l} \rho^{l} \text{ and } \sum_{l\in L: {\wa}^{l} = \wa} \frac{1}{2^l} \rho^{l}$ 
					are well-defined integrable functions. 
					This follows by Beppo Levi's monotone convergence theorem (cf.~\cite[Theorem 2.8.2/Corollary 2.8.6]{Bogachev2007I}) as we can bound 
					\begin{align*}
						\sum_{l \in L} \frac{1}{2^l} \rho^{l} \leq  \sum_{l\in L}  \frac{1}{2^l} h_{\wa^l} \leq \sum_{l\in L}  \frac{1}{2^l} \inflow_s \leq  \inflow_s \sum_{l \in \N} \frac{1}{2^l}   =  \inflow_s. 
					\end{align*}
					Hence, we get the integrability $h_\wa' \in L(\hori)$. 
					In particular, we get for any $\wa \in \sdRoutes$ that 
					\begin{align*}
						{h}_\wa' \geq {h}_\wa - \sum_{l\in L: {\wa}^{l} = \wa}\frac{1}{2^l}\rho^{l} \geq   {h}_\wa  - \sum_{l\in L: {\wa}^{l} = \wa}\frac{1}{2^l}h_{\wa^l} \geq  {h}_\wa   - \sum_{l\in \N}\frac{1}{2^l}h_{\wa} 
						=  {h}_\wa -     {h}_\wa    = 0 ,
					\end{align*}
					that is, $h_\wa'$ is nonnegative.

					Finally, for $c \in \SimpCyc$, we observe that
					\begin{align}\label{eq: h_c'}
						h_c' &=  h_c -  \sum_{l \in L} \frac{1}{2^l} \ell_{\wa^l,j_l}(\rho^{l}) \symoverset{1}{\geq} h_c -  \sum_{l \in L} \frac{1}{2^l} \ell_{\wa^l,j_l}\big(\frac{1}{2}\hat{h}^l\big)
						= h_c -  \frac{1}{2} \sum_{l \in L} \frac{1}{2^l} \min_{\tilde{c} \in C_{\n1}}h_{\tilde{c}} \cdot \Indi_{ \arr_{\wa^{l},j_{l}}(\mathfrak{D}^{l})} \nonumber
						\\&\geq h_c -  \frac{1}{2} \sum_{l \in L} \frac{1}{2^l} \min_{\tilde{c} \in C_{\n1}}h_{\tilde{c}}  \geq 
						h_c -  \frac{1}{2} \min_{\tilde{c} \in C_{\n1}}h_{\tilde{c}} \geq  h_c -  \frac{1}{2} h_c  =   \frac{1}{2} h_c 
					\end{align}
					where the first inequality \refsym{1} holds by \Cref{lem: ellOrderPreserving} and the definition of $\rho^l$. 
					Thus, we may infer that $h'_c\geq 0$. Moreover, it is clear that $h'_c \leq h_c$ and hence $h_c' \in L_+(\hori)$ follows by $h_c \in L_+(\hori)$.

					\proofitem{\ref{subclaim: FlowDecompZshKomps: h'props: h_ch'_c}} It is clear by definition of $h_c'$ that for almost every $t \in \hori$ the implication $h_c'(t) >0 \Rightarrow h_c(t)>0$ is true. The reverse direction is a direct consequence of the estimate shown in~\eqref{eq: h_c'}.

					\proofitem{\ref{subclaim: FlowDecompZshKomps: h'props: l(h')ex}} 
					Let $\wa \in \sdRoutes$ be arbitrary. 
					By the continuity and linearity of $\ell_\wa$, we get \begin{align*}
						&\ell_\wa\Big(\sum_{l\in L: \hat{\wa}^{l} = \wa}\frac{1}{2^l} \rho^{l}\Big) = \sum_{l\in L: \hat{\wa}^{l} = \wa}\frac{1}{2^l}  \ell_\wa( \rho^{l}) \quad \text{ and }\\
						&\ell_\wa\Big(\sum_{l\in L:  {\wa}^{l} = \wa}\frac{1}{2^l} \rho^{l}\Big) = \sum_{l\in L:  {\wa}^{l} = \wa}\frac{1}{2^l}  \ell_\wa( \rho^{l})
					\end{align*}
					in case that $\ell_\wa(\rho^{l})$ exists for all  $l\in L$ with either $\hat{\wa}^{l} = \wa$ or $\wa^{l} = \wa$. 
					For the latter case, we have $\rho^{l} \leq h_{\wa^{l}} = h_{\wa}$ and since $\ell_\wa(h_\wa)$ exists, it follows that 
					$\ell_\wa( \rho^{l})$ exists. 
					For the former case, let  $l\in L$ with $\hat{\wa}^{l} = \wa$ be arbitrary for the following.  
					Since $\wa^{l}$ and $\hat{\wa}^{l}$($=\wa$) have the first $j_{l}-1$ edges in common, $\rho^{l} \leq h_{\wa^{l}}$ and
					$\ell_{\wa^{l},j}( h_{\wa^{l}} )$ for $j \leq\abs{\wa^l}+1$ exist (for $j= \abs{\wa ^l}+1$ by existence of $\inflow_d$), it follows that 
					$\ell_{\wa,j}( \rho^{l})$ exists for all $j \leq j_{l}$. 
					For $j = j_{l} +z$ with $z \in \{0,\ldots,|c_{\n1}^{l}|\}$, \Cref{lem: FLowOnZeroTrav} applied to $\mathfrak D^l$ yields that
					$\ell_{\wa,j}( \rho^{l})$ exists and is equal to $\ell_{\wa,j_{l}}( \rho^{l}) $. Remark that $\rho^l = 0$ on  $\hori\setminus\mathfrak D^l$ by  $\hat{h}^l$ being $0$ on the latter set. Moreover, note that 
					\Cref{lem: FLowOnZeroTrav} is applicable for $\mathfrak D^l$ as we have 
					$\arr_{\wa,j_l}(t) = \arr_{\wa,j}(t)$ for almost every $t \in \mathfrak D^l$ by    
					$ \arr_{\wa,j_l}(t) = \arr_{\wa^l,j_l}(t) \in \mathfrak T_{C_{\n1}}$ and $\wa[j] = \hat{\wa}^l[j] \in \GA_{C_{\n1}}$ as well as  $\trav_\arc(t')= 0$ for all $t' \in \mathfrak T_{C_{\n1}}$ and $\arc \in \GA_{C_{\n1}}$.

					Finally, consider $j = j_{l} +|c_{\n1}^{l}| + z$ with $z\in \{1,\ldots,\abs{\wa} -j_{l}\}$. 
					We argue for the validity of the following equalities: 
					\begin{equation}\label{eq: identity}
							\begin{aligned}
							\ell_{\wa^l,j_l+z} ( \rho^{l})\overset{\text{\Crefshort{lem: elluPropagation}}}&{=} \ell_{\wa^l_{\geq j_l},z+1}(\ell_{\wa^l,j_{l}}( \rho^{l})) \symoverset{1}{=}\ell_{\wa^l_{\geq j_l},z+1}(\ell_{\wa,j_{l} +|c_{\n1}^{l}|}( \rho^{l})) \\
							&\symoverset{2}{=}\ell_{\wa_{\geq j_l+|c_{\n1}^{l}|},z+1}(\ell_{\wa,j_{l} +|c_{\n1}^{l}|}( \rho^{l})) \Croverset{lem: elluPropagation}{=}  	\ell_{\wa,j_{l} +|c_{\n1}^{l}| + z}( \rho^{l}) =	\ell_{\wa,j}( \rho^{l})
						\end{aligned}
					\end{equation}
					For the above, first remark that by $\rho^{l} \leq h_{\wa^{l}}$ and
					$\ell_{\wa^{l},j_l+z}( h_{\wa^{l}} )$ existing, it follows that 
					$\ell_{\wa^l,j_l+z}( \rho^{l})$ exists. 
					For equality \refsym{1}, we used the previous insights that 
					$\ell_{\wa,j_{l} +|c_{\n1}^{l}|}( \rho^{l}) = \ell_{\wa,j_{l}}( \rho^{l})$. 
					For \refsym{2}, we used that $\wa^l_{\geq j_l} = \wa_{\geq j_l+|c_{\n1}^{l}|}$. 
					Note that the equalities in which we used \Cref{lem: elluPropagation} imply in particular the existence of the expression on the right hand side of the equality.

					\proofitem{\ref{subclaim: FlowDecompZshKomps: h'props: GeqPureFlow} and \ref{subclaim: FlowDecompZshKomps: h'props: equalEflow}} 
					By the linearity of $\ell_{\wa,e}$ (\Cref{lem: elluExistenceProperties:Linearity}) and the identities derived in \ref{subclaim: FlowDecompZshKomps: h'props: l(h')ex}
					we calculate for 
					an arbitrary $\arc \in \GA$:
					\begin{align}
						\sum_{\wa \in \sdRoutes}\ell_{\wa,\arc}(h_\wa') &= \sum_{\wa \in \sdRoutes} \ell_{\wa,\arc}(h_\wa)+ \sum_{\wa \in \sdRoutes} \sum_{l: \hat{\wa}^{l} = \wa}\frac{1}{2^l}  \ell_{\wa,\arc}( \rho^{l}) -\sum_{l:  {\wa}^{l} = \wa}\frac{1}{2^l}  \ell_{\wa,\arc}( \rho^{l}) \nonumber\\
						&=  \sum_{\wa \in \sdRoutes} \ell_{\wa,\arc}(h_\wa) + \sum_{l \in L} \frac{1}{2^l} \Big(\ell_{\hat{\wa}^l,\arc}( \rho^{l}) - \ell_{\wa^l,\arc}( \rho^{l})\Big)\nonumber \\
						&\symoverset{1}{=}  \sum_{\wa \in \sdRoutes} \ell_{\wa,\arc}(h_\wa) +   \sum_{l \in L} \frac{1}{2^l}
						\Big(\sum_{j\in \{j_{l},\ldots,j_{l}+|c^{l}_{\n1}|-1 \}: \hat{\wa}^{l}[j] = \arc } \ell_{\hat{\wa}^{l},j}(\rho^{l})    \Big)\nonumber \\
						&\symoverset{2}{=}  \sum_{\wa \in \sdRoutes} \ell_{\wa,\arc}(h_\wa) +   \sum_{l \in L} \frac{1}{2^l}
						\Big(\abs{\{c \in C_{\n1}:\arc \in c\}} \cdot \ell_{{\wa}^{l},j_l}(\rho^{l})    \Big) \label{eq: identity2} 
					\end{align}
					where we used in the equality indicated by~\refsym{1} that 
					$\wa^l$ and $\hat{\wa}^l$ share the first $j_l-1$ edges as well as the identity in~\eqref{eq: identity}.
					For the equality specified by \refsym{2}, we used that $\ell_{\hat{\wa}^{l},j}(\rho^{l})  = \ell_{{\wa}^{l},j_l}(\rho^{l}) $ for $j\in \{j_{l},\ldots,j_{l}+|c^{l}_{\n1}|-1 \}$ as well as the fact that $c_{\n1}^l$ contains every cycle in $C_{\n1}$ exactly once and each cycle in $C_{\n1}$ is simple. 
					
					From this, \ref{subclaim: FlowDecompZshKomps: h'props: GeqPureFlow} now follows by $\ell_{\wa^l,j_l}(\rho^l)$ being nonnegative. 
					The statement in \ref{subclaim: FlowDecompZshKomps: h'props: equalEflow} follows by observing that 
					\begin{align*}
						\sum_{c \in \SimpCyc:\arc \in c} h_c' = \sum_{c \in \SimpCyc:\arc \in c} h_c - \abs{\{c \in C_{\n1}:\arc \in c\}} \cdot \sum_{l \in L} \frac{1}{2^l}\ell_{{\wa}^{l},j_l}(\rho^{l}) 
					\end{align*}
					which implies, together with the above \eqref{eq: identity2}, that  
					\begin{align*}
						\sum_{\wa \in \sdRoutes} \ell_{\wa,\arc}({h}'_\wa) + \sum_{c \in \SimpCyc: \arc \in c}h_c' =\sum_{\wa \in \sdRoutes} \ell_{\wa,\arc}({h}_\wa) + \sum_{c \in \SimpCyc: \arc \in c}h_c= \eflow_\arc.
					\end{align*}
					
					\proofitem{\ref{subclaim: FlowDecompZshKomps: h'props: g'_c>0}} 
					Since  $\hat{h}^l$ and $h_{\wa^l}$ are larger than $0$ almost everywhere on $\mathfrak D^l$, so is $\rho^l$. 
					Thus, \Cref{lem: ellOrderPreservingSharpened} implies that $\ell_{\wa^l,j_l}(\rho^l)$ is larger than $0$ almost everywhere on  $\arr_{\wa^{l},j_{l}}(\mathfrak{D}^{l})$. Since, furthermore, $\bigcup_{l \in L}  \arr_{\wa^{l},j_{l}}(\mathfrak{D}^{l})$ equals $\mathfrak T_{C_{\n1}}$ up to a null set, we may deduce that for $e \in c \in C_{\n1}$ the second term in~\eqref{eq: identity2}  is strictly larger than zero almost everywhere on  $\mathfrak T_{C_{\n1}}$, from which the claim follows immediately.\qedhere
				\end{structuredproof}
			\end{proofClaim}
			
			With this, we can now deduce that $h'$ has all the properties stated in \Cref{claim: FlowDecompZshKomps}: First, \labelcref{subclaim: FlowDecompZshKomps: h'props: h'wd,subclaim: FlowDecompZshKomps: h'props: h_ch'_c,subclaim: FlowDecompZshKomps: h'props: equalEflow,subclaim: FlowDecompZshKomps: h'props: l(h')ex} together show that $h'$ is a decomposition of~$\eflow$. Then, \ref{subclaim: FlowDecompZshKomps: h'props: h_ch'_c} ensures that we can use the same sets according to \Cref{def: ConnectedComp} for~$h'$ as for~$h$. Next, \ref{subclaim: FlowDecompZshKomps: h'props: GeqPureFlow} together with \ref{subclaim: FlowDecompZshKomps: h'props: h_ch'_c} guarantees that implication~\eqref{eq:ZeroCyclesConnected} remains true for all $\n1' \in \capn1$ for which it was already true under~$h$ and, finally, \ref{subclaim: FlowDecompZshKomps: h'props: g'_c>0} demonstrates that, under~$h'$, this implication holds for~$\n1$ as well.
		\end{proofClaim}
		With this \namecref{claim: FlowDecompZshKomps} at hand, we may assume now \wlg that $h$ fulfills the implication 
		\begin{align}\label{eq: FlowDecompEachCycle}
			h_c(t)>0 \implies \sum_{\wa \in \sdRoutes}\ell_{\wa,\arc}(h_\wa)(t)>0 \text{ for almost all } t \in \hori \text{ and all } \arc \in c, c \in \SimpCyc
		\end{align}
		as we can otherwise apply \Cref{claim: FlowDecompZshKomps} successively over all $\n1 \in \capn1$ and consider the resulting alternative flow decomposition. 
		Starting with such a walk inflow~$h$ we construct in the following for an arbitrary $c \in \SimpCyc$ another flow decomposition which does not use a zero-cycle inflow rate into~$c$. From this, the theorem itself then, finally, follows by successively applying this construction to~$h$ for all $c\in \SimpCyc$: 
		\begin{claim}\label{claim: FlowDecompSingleCycle} 
			Given a  flow $h$ fulfilling~\eqref{eq: FlowDecompEachCycle}, we can construct 
			for an arbitrary $c \in \SimpCyc$ 
			(a representative of) another flow decomposition $h'_\wa,\wa \in  {\sdRoutes},h'_c,c \in \SimpCyc$ of $\eflow$ which  still satisfies~\eqref{eq: FlowDecompEachCycle} and fulfills $h_c'=0$ as well as $h'_{\tilde{c}} = h_{\tilde{c}}$ for all $\tilde{c}\neq c$.
		\end{claim}
		\begin{proofClaim}
			The main idea for the construction is the same as the one we used for \Cref{claim: FlowDecompZshKomps} and, in particular, \Cref{subclaim: FlowDecompZshKomps: h'props}: Shift flow from walks $\wa$ intersecting $c$ to walks $\hat\wa$ that additionally traverse the cycle~$c$ in such a way as to replace all the zero-cycle inflow into~$c$. Note, however, that such a direct shift will not always be possible as the flow rate on a walk~$\wa$ might be smaller than the one on~$c$ at the relevant times. We will remedy this by using walks~$\hat\wa$ that traverse~$c$ not just once but multiple times. 
			
			We now start by obtaining the walks~$\wa$ from which we shift flow:
			Let $\arc \coloneqq c[1]$ be the first edge on~$c$. Since $h$ satisfies~\eqref{eq: FlowDecompEachCycle}, there exist, by \Cref{lem: Relations:h>0u>0:u>0ExistsCountableM}, a countable set  $M$ (denoted by $L$ in \Cref{lem: Relations:h>0u>0:u>0ExistsCountableM}) and walks $\{\wa^m\}_{m\in M} \subseteq \sdRoutes$ together with indices $j_m\leq |\wa^m|$ and measurable sets $\mathfrak D^m, \leb(\mathfrak D^m)>0$ for all $m \in M$ such that $\wa^m[j_m] = \arc$, $h_{\wa^m}(t) > 0$ for a.e.~$t \in \mathfrak D^m$ and $\mathfrak T^{m} := \arr_{\wa^m,j_m}(\mathfrak D^m), m \in M$ are disjoint with $\bigcup_{m \in M} \mathfrak T^m$ equaling $\mathfrak T_c$ up to a null set.
			Furthermore, for any walk $\wa \in \sdRoutes$, there are only finitely many $m \in M$ with $\wa^m=\wa$ and hence ${\norm{\wa}_M} :=|\{m \in M\mid \wa^m = \wa\}|< \infty$.   
			 
			The \stwalk{s} to which we shift flow from any $\wa^m$ are then defined by 
			 $\hat{\wa}^{m,n}:=(\wa^m_{<j_m},c^{m,n},\wa^m_{\geq j_m})$ where $c^{m,n}$ is the concatenation of $n\cdot{\norm{\wa^m}_M}$ copies of~$c$.

			Next, we define $\hat{h}^m$ as a function fulfilling $\ell_{\wa^m,j_m}(\hat{h}^m) = \Indi_{\mathfrak T^m}\cdot {h}_c$, which exists by~\Cref{lem: 1to1:h-f}, and consider an arbitrary representative. 
			Now, $\hat h^m$ is exactly the amount of flow we would like to add to the inflow into $(\wa^m_{<j_m},c,\wa^m_{\geq j_m})$ in order to replace all the flow on~$c$ induced by the zero-cycle inflow rate~$h_c$. If we had $\hat h^m \leq h_{\wa^m}$, this could be achieved by a single shift (of $\hat h^m$) from $\wa^m$ to $(\wa^m_{<j_m},c,\wa^m_{\geq j_m})$. Since this need not be the case in general, we instead consider the fraction $\hat{h}^m / h_{\wa^m}$ and aim to shift only $\frac{\hat{h}^m}{\lceil\hat{h}^m / h_{\wa^m}\rceil} \leq h_{\wa^m}$ from $\wa^m$ to the walk that contains the cycle $c$ exactly $\lceil\hat{h}^m / h_{\wa^m}\rceil$-many times. A final technicality now is that we might have $\wa^m = \wa^{m'}$ for different $m \neq m'$ and, hence, we have to be even more careful when removing flow from~$\wa^m$. We fix this by considering how many times $\wa^m$ appears in $M$ (i.e.\ $\norm{\wa^m}_M$) and then only remove $\frac{\hat{h}^m}{\lceil\hat{h}^m / h_{\wa^m}\rceil\cdot\norm{\wa^m}_M}$ from $\wa^m$ for each~$m$ and shift it to a walk, that traverses~$c$ even more often, namely $\lceil\hat{h}^m / h_{\wa^m}\rceil\cdot\norm{\wa^m}_M$ many times.
			
			To formalize this, we start by defining a function $\hat{\rho}^m$ via  
			\begin{align*}
				\hat\rho^m: \hori \to \Rnn, t \mapsto \begin{cases}
					\frac{\hat h^m(t)}{h_{\wa^m}(t)}, &\text{if } t \in \mathfrak D^m \\
					0,                                 &\text{else}
				\end{cases}
			\end{align*}
			where we choose a representative of $h$ with $h_{\wa^m}(t)>0$ for all $t \in \mathfrak D^m$, which is possible by the definition of the set $\mathfrak D^m$. 
			Furthermore, set $\rho^m := \lceil \hat{\rho}^m \rceil $
			where $\lceil \cdot\rceil$ denotes the standard ceiling function, i.e.~$\lceil x\rceil$ is the smallest integer that is greater or equal to $x$. Remark that  $\rho^m$ is measurable as the ceiling function and $\hat{\rho} ^m$ are likewise.
			
			We now define $h'$ as follows: For any $\wa \in \sdRoutes$, set 
			\begin{align*}
				h_\wa' &\coloneqq  h_\wa - \sum_{m: \wa ^m = \wa} \sum_{n \in \N} \Indi_{(\rho^{m})^{-1}(n)} \cdot \frac{\hat{\rho}^m}{n {\norm{\wa^m}_M}} h_{\wa^m} + 
				\sum_{(m,n): \hat{\wa}^{m,n} = \wa} \Indi_{(\rho^{m})^{-1}(n)}  \cdot \frac{\hat{\rho}^m}{n {\norm{\wa^m}_M}}h_{\wa^{m}}
			\end{align*}
			and define $h'_{\tilde{c}} \coloneqq h_{\tilde{c}}$ for $\tilde{c}\neq c$ and $h'_c \coloneqq 0$. 
			Here and from now on, we assume that every index $m$ is an element of $M$ unless it is explicitly stated otherwise. Note also, that $(\rho^{m})^{-1}(n)$ is just the set of times $t \in \mathfrak D^m$ where $\hat\rho^m(t) \in (n-1,n]$.

			As the zero-cycle inflow rates under $h'$ are already as required in the claim, it only remains to show that $h'$ is a (well-defined) walk-decomposition of $\eflow$ and still satisfies~\eqref{eq: FlowDecompEachCycle}:
			\begin{subclaim} \label{subclaim: FlowDecompSingleCycleh'}
				$h'$ fulfills the following:
				\begin{enumerate}[label = \roman*)]
					\item $h'$ is well-defined and  $h'\in L_+(\hori)^{\sdRoutes\cup\SimpCyc}$. \label{subclaim: FlowDecompSingleCycleh':geq0}
					\item $\ell_\wa(h'_\wa)$ exists for all $\wa \in \sdRoutes$. \label{subclaim: FlowDecompSingleCycleh':l(h')ex}
					\item $\sum_{\wa \in \sdRoutes} \ell_{\wa,\arc}({h}'_\wa) + \sum_{c \in \SimpCyc: \arc \in c}h_c' = \eflow_\arc$ for all $\arc \in \GA$.  \label{subclaim: FlowDecompSingleCycleh':equal} 
					\item $h'$ fulfills \eqref{eq: FlowDecompEachCycle}. \label{subclaim: FlowDecompSingleCycleh':Prop} 
				\end{enumerate}
			\end{subclaim}
			\begin{proofClaim}
				\begin{structuredproof}
					\proofitem{\ref{subclaim: FlowDecompSingleCycleh':geq0}} 
					Let $\wa \in \sdRoutes$ be arbitrary.

					The well-definedness of the first sum follows immediately by the aforementioned property that 
					${\norm{\wa^m}_M} :=|\{m \in M\mid \wa^m = \wa\}|< \infty$ is finite together with the obvious observation that 
					$(\rho^{m})^{-1}(n) \cap (\rho^{m})^{-1}(n') = \emptyset$ for any two $n\neq n'$. 
					The second sum is  also well-defined as it contains only finitely many summands. To see this, 
					note that there are only finitely many walks $\wa'$ in $\sdRoutes$ that 
					can be extended to $\wa$ by adding copies of~$c$, i.e.~for which there exists $m' \in M$ and $n \in \N$ with $\wa' = \wa^{m'}$ and $\hat{\wa}^{m',n} = \wa$. 
					This, together with the fact that $\norm{\wa'}_M < \infty$ for all $\wa' \in \sdRoutes$, shows the finiteness of the sum. 
					In particular, this together with the estimate 
					\begin{align*}
						h_\wa'  &\leq  h_\wa + 
						\sum_{(m,n): \hat{\wa}^{m,n} = \wa} \Indi_{(\rho^{m})^{-1}(n)}  \cdot \frac{\hat{\rho}^m}{n {\norm{\wa^m}_M}}h_{\wa^{m}} \\
						&\leq  h_\wa + \sum_{(m,n): \hat{\wa}^{m,n} = \wa}  h_{\wa^{m}} 
					\end{align*}
					shows that $h_\wa'$ is bounded from above by an integrable function. Hence, $h_\wa' \in L_+(\hori)$ follows by the nonnegativity of $h'_\wa$, which can be shown by a direct computation for any $t \in \R$:
                    
					\begin{align*}
						h'_\wa(t) 
						&\geq h_\wa(t) - \sum_{m:\wa^m=\wa}\sum_{n \in \N}\CharF[(\rho^m)^{-1}(n)](t)\cdot\frac{\hat\rho^m(t)}{n\norm{\wa^m}_M}h_{\wa^m}(t) \\
						&= h_\wa(t) - \sum_{m:\wa^m=\wa}\frac{\hat\rho^m(t)}{\rho^m(t)\norm{\wa}_M}h_{\wa}(t) \\
						&= h_\wa(t) - \frac{h_\wa(t)}{\norm{\wa}_M} \sum_{m:\wa^m=\wa}\frac{\hat\rho^m(t)}{\ceil{\hat\rho^m(t)}} \\
						&\geq h_\wa(t) - \frac{h_\wa(t)}{\norm{\wa}_M} \cdot \norm{\wa}_M = 0.
					\end{align*}

					\proofitem{\ref{subclaim: FlowDecompSingleCycleh':l(h')ex}} This proof works similarly to the one of \Cref{subclaim: FlowDecompZshKomps: h'props}\ref{subclaim: FlowDecompZshKomps: h'props: l(h')ex}. 
					
					Let $\wa \in \sdRoutes, m \in M$ with $\hat{\wa}^{m,n} = \wa$ or $\wa^m =\wa$ and $n \in \N$. We aim to show that   $\ell_{\wa}(\Indi_{(\rho^{m})^{-1}(n)}  \cdot \frac{\hat{\rho}^m}{n{\norm{\wa ^m}_M}} h_{\wa^m})$ exists. Then,  
					by linearity of $\ell_\wa$ and the existence of $\ell_\wa(h_\wa)$, the existence of  $\ell_\wa(h_\wa')$ follows. 
					
					In the case of $\wa^m =\wa$, the required existence follows as $\Indi_{(\rho^{m})^{-1}(n)}  \cdot \frac{\hat{\rho}^m}{n{\norm{\wa ^m}_M}} h_{\wa^m} \leq h_{\wa^m} =h_\wa$ and $\ell_{\wa}(h_\wa)$ exists. 
					Hence, consider the case of $\hat{\wa}^{m,n} = \wa$. The existence of $\ell_{\wa,j}(h_\wa),j\leq \abs{\wa}+1$ together with the observation that $\hat{\wa}^{m,n}$ ($=\wa$) and $\wa^m$ share the first $j_m-1$ edges shows the existence and equality of 
					\begin{align}\label{eq: subclaim: FlowDecompSingleCycleh':InducedFlow3}
						\ell_{\wa,j}\Big(\Indi_{(\rho^{m})^{-1}(n)}  \cdot \frac{\hat{\rho}^m}{n{\norm{\wa ^m}_M}} h_{\wa^m}\Big) = \ell_{\wa^m,j}\Big(\Indi_{(\rho^{m})^{-1}(n)}  \cdot \frac{\hat{\rho}^m}{n{\norm{\wa ^m}_M}} h_{\wa^m}\Big) \text{ for all }j\leq j_m.
					\end{align}

					For $j = j_{m} +z$ with $z \in \{0,\ldots,|c^{m,n}|\}$, we argue in the following for the existence and equality:
					\begin{align}\label{eq: subclaim: FlowDecompSingleCycleh':InducedFlow}
						\ell_{\wa,j}\Big(\Indi_{(\rho^{m})^{-1}(n)} \cdot \frac{\hat{\rho}^m}{n{\norm{\wa ^m}_M}} h_{\wa^m}\Big) =  \ell_{\wa^m,j_m} \Big( \Indi_{(\rho^{m})^{-1}(n)}  \cdot \frac{1}{n{\norm{\wa ^m}_M}} \Indi_{\mathfrak D^m} &\hat{h}^m \Big).
					\end{align}
					Note that the right hand side exists since $\ell_{\wa^m,j_m}(\hat{h}^m)$ exists. 
					We calculate  for an arbitrary $\mathfrak T \in \mathcal{B}(\hori)$:
					\begin{align*}
						 \int_{\mathfrak T} \ell_{\wa^m,j_m} \Big( \Indi_{(\rho^{m})^{-1}(n)}  \cdot \frac{1}{n{\norm{\wa ^m}_M}} \Indi_{\mathfrak D^m} &\hat{h}^m \Big)\di\leb \\  
						&\symoverset{3}{=}  \int_{\arr_{\wa^m,j_m}^{-1}(\mathfrak T)} \Indi_{(\rho^{m})^{-1}(n)}  \cdot \frac{1}{n{\norm{\wa ^m}_M}} \Indi_{\mathfrak D^m} \hat{h}^m \di\leb \\
						&\symoverset{4}{=}  \int_{\arr_{\wa,j_m}^{-1}(\mathfrak T)} \Indi_{(\rho^{m})^{-1}(n)}  \cdot \frac{1}{n{\norm{\wa ^m}_M}} \Indi_{\mathfrak D^m} \hat{h}^m \di\leb \\
						&\symoverset{5}{=}   \int_{\arr_{\wa,j}^{-1}(\mathfrak T)} \Indi_{(\rho^{m})^{-1}(n)}  \cdot \frac{1}{n{\norm{\wa ^m}_M}} \Indi_{\mathfrak D^m} \hat{h}^m \di\leb \\    
						&= \int_{\arr_{\wa,j}^{-1}(\mathfrak T)} \Indi_{(\rho^{m})^{-1}(n)}  \cdot \frac{\hat{\rho}^m}{n{\norm{\wa ^m}_M}} h_{\wa^m}\di\leb
					\end{align*}
					Here, \refsym{3} holds simply by definition of $\ell_{\wa^m,j_m}$.  
					The equality  \refsym{4} follows as $\hat{\wa}^{m,n}$ ($= \wa$) and $\wa^m$ share the first $j_m-1$ edges and 
					\refsym{5} 
					follows from  
					$\arr_{\wa,j} = \arr_{\wa,j_{m}}$ on $\mathfrak D^{m}$ due to $\arr_{\wa,j_{m}}(\mathfrak D^{m}) = \mathfrak T^m\subseteq \mathfrak T_{c}$. Note that equality on $\mathfrak D^m$ is sufficient  
					as  a part of the integrand is a indicator function over $\mathfrak D^m$. 
					The last equality is due to the definition of $\hat{\rho}^m$. 
					Hence, we have shown that $ \Indi_{\arr_{\wa^m,j_m}((\rho^{m})^{-1}(n))} \cdot  \Indi_{\mathfrak T^m} \cdot \frac{1}{n{\norm{\wa ^m}_M}}h_c$ fulfills \eqref{eq: Deflwj} for $\Indi_{(\rho^{m})^{-1}(n)}  \cdot \frac{\hat{\rho}^m}{n{\norm{\wa ^m}_M}} h_{\wa^m}$ which shows the claimed existence and equality.

					Finally, 
					the exact same argumentation as carried out for \eqref{eq: identity} shows the existence and equality of
					\begin{align}\label{eq: subclaim: FlowDecompSingleCycleh':InducedFlow2}
						\ell_{\wa,j_{m} +|c^{m,n}| + z}\left(\Indi_{(\rho^{m})^{-1}(n)}  \cdot \frac{\hat{\rho}^m}{n{\norm{\wa ^m}_M}} h_{\wa^m}\right) = \ell_{\wa^m,j_m+z}\left(\Indi_{(\rho^{m})^{-1}(n)}  \cdot \frac{\hat{\rho}^m}{n{\norm{\wa ^m}_M}} h_{\wa^m}\right)
					\end{align}
					for any $z\in \{1,\ldots,\abs{\wa} -j_{m}\}$. 
					
					Hence, we have shown that $\ell_{\wa,j}(h_\wa')$ exists for all $j \leq \abs{\wa}$, from which the existence of  $\ell_{\wa}(h_\wa')$ follows.  
					
					\proofitem{\ref{subclaim: FlowDecompSingleCycleh':equal}} 
					
					We show in the following that $\sum_{\wa \in \sdRoutes}\ell_{\wa,\arc}(h'_\wa)= \sum_{\wa \in \sdRoutes}\ell_{\wa,\arc}(h_\wa) +   \Indi_{c}(\arc) \cdot h_c$ for all $\arc \in \GA$ from which the claimed equality follows immediately as $h'_c = 0$ and $h'_{\tilde c} = h_{\tilde c}$ for all $\tilde c \neq c$ by the definition of~$h'$. Here, we denote by $ \Indi_{c}(\arc)$ the function that is $1$ if $\arc \in c$ and $0$ else. 
					With the help of the identities derived in~\ref{subclaim: FlowDecompSingleCycleh':l(h')ex}, we can calculate for an arbitrary $\arc \in\GA$: 					
					\begin{align}
						\sum_{\wa \in \sdRoutes}&\ell_{\wa,\arc}(h'_\wa) - \sum_{\wa \in \sdRoutes}\ell_{\wa,\arc}(h_\wa)= \notag \\ 
					\notag	&=  
						\sum_{\wa\in \sdRoutes}\Bigg[ -\sum_{m: \wa ^m = \wa} \sum_{n \in \N} \ell_{\wa,\arc}\Big(\Indi_{(\rho^{m})^{-1}(n)} \frac{\hat{\rho}^m}{n \norm{\wa^m}_M} h_{\wa^m} \Big)\\&\quad \quad + 
						\sum_{(m,n): \hat{\wa}^{m,n} = \wa}\ell_{\wa,\arc}\Big( \Indi_{(\rho^{m})^{-1}(n)}  \frac{\hat{\rho}^m}{n \norm{\wa^m}_M}h_{\wa^{m}}\Big)\Bigg] \\
					\notag	&= \sum_{m \in M}\sum_{n \in \N} \ell_{\hat{\wa}^{m,n},\arc}\Big( \Indi_{(\rho^{m})^{-1}(n)}  \frac{\hat{\rho}^m}{n \norm{\wa^m}_M}h_{\wa^{m}}\Big) - \ell_{\wa^{m},\arc}\Big( \Indi_{(\rho^{m})^{-1}(n)}  \frac{\hat{\rho}^m}{n \norm{\wa^m}_M}h_{\wa^{m}}\Big) \\
					\notag	&\underset{\eqref{eq: subclaim: FlowDecompSingleCycleh':InducedFlow3}}{\overset{\eqref{eq: subclaim: FlowDecompSingleCycleh':InducedFlow2}}{=}} \sum_{m \in M}\sum_{n \in \N} \sum_{j\in \{ j_m,\ldots,j_m+\abs{c^{m,n}}-1\}: \hat{\wa}^{m,n}[j] = \arc } \ell_{\hat{\wa}^{m,n},j}\Big(\Indi_{(\rho^{m})^{-1}(n)} \cdot \frac{\hat{\rho}^m}{n{\norm{\wa ^m}_M}} h_{\wa^m}\Big) \\
					\notag	&\overset{\eqref{eq: subclaim: FlowDecompSingleCycleh':InducedFlow}}{=}  \sum_{m \in M}\sum_{n \in \N} \sum_{j\in \{ j_m,\ldots,j_m+\abs{c^{m,n}}-1\}: \hat{\wa}^{m,n}[j] = \arc }
						\ell_{{\wa}^{m},j_m}\Big(\Indi_{(\rho^{m})^{-1}(n)} \cdot \frac{1}{n{\norm{\wa ^m}_M}} \Indi_{\mathfrak D^m} \hat h^m\Big) \\
					\notag	&\symoverset{1}{=}  \Indi_{c}(\arc) \cdot \sum_{m \in M}\sum_{n \in \N}  	\ell_{{\wa}^{m},j_m}\Big(\Indi_{(\rho^{m})^{-1}(n)} \cdot \Indi_{\mathfrak D^m} \hat h^m\Big)\\
					\notag	&\symoverset{2}{=} \Indi_{c}(\arc) \cdot \sum_{m \in M}	\ell_{{\wa}^{m},j_m}\Big(\sum_{n \in \N}\Indi_{(\rho^{m})^{-1}(n)} \cdot \Indi_{\mathfrak D^m} \hat h^m\Big) \\
					\notag	&=  \Indi_{c}(\arc) \cdot \sum_{m \in M}	\ell_{{\wa}^{m},j_m}\Big( \Indi_{\R \setminus(\rho^{m})^{-1}(0)} \cdot\Indi_{\mathfrak D^m} \hat h^m\Big) \\ 
					\notag	&=  \Indi_{c}(\arc) \cdot \sum_{m \in M}	\ell_{{\wa}^{m},j_m}\Big( \Indi_{\mathfrak D^m \setminus(\rho^{m})^{-1}(0)} \hat h^m\Big) \\
						\overset{\text{\Crefshort{lem: elluindi}}}&{=}  \Indi_{c}(\arc) \cdot \sum_{m \in M}	 \Indi_{\arr_{\wa^m,j_m}\big(\mathfrak D^m\setminus(\rho^{m})^{-1}(0)\big)}	\cdot\ell_{{\wa}^{m},j_m}\Big( \hat h^m\Big) \\ 
						&=  \Indi_{c}(\arc) \cdot \sum_{m \in M}	 \Indi_{\arr_{\wa^m,j_m}\big(\mathfrak D^m\setminus(\rho^{m})^{-1}(0)\big)}	\cdot\Indi_{\arr_{\wa^m,j_m}(\mathfrak D^m)} \cdot h_c \label{eq: subclaim: FlowDecompSingleCycleh: 1}
					\end{align}
					 
					where the equality indicated by \refsym{1} holds as the middle part of~$\hat{\wa}^{m,n}$ consists of $n \cdot \norm{\wa^m}_M$ copies of the simple cycle~$c$. 
					In the equality indicated by \refsym{2}, we used the linearity and continuity of $\ell_{\wa^m,j_m}$ as well as the convergence of $\sum_{n \in \N}\Indi_{(\rho^{m})^{-1}(n)} \cdot \Indi_{\mathfrak D^m} \hat h^m$. 
					The final equality is due to the definition of $\hat{h}^m$. 
					
					Now in order to show that the expression in the final line is equal to $\Indi_c(e) \cdot h_c$, we argue in the following that $h_c = 0$ on $\arr_{\wa^m,j_m}(\mathfrak D^m)\setminus \arr_{\wa^m,j_m}\big(\mathfrak D^m\setminus(\rho^{m})^{-1}(0)\big)$ for all $m \in M$. Consider an arbitrary
					  $m \in M$ and observe that 
					\begin{align}\label{eq: subclaim: FlowDecompSingleCycleh}
						\arr_{\wa^m,j_m}(\mathfrak D^m)\setminus \arr_{\wa^m,j_m}(\mathfrak D^m\setminus (\rho^{m})^{-1}(0)) \subseteq \arr_{\wa^m,j_m}(\mathfrak D^m \cap(\rho^{m})^{-1}(0) )
					\end{align}
					which allows us to calculate
					\begin{align*}
						\int_{\arr_{\wa^m,j_m}(\mathfrak D^m)\setminus \arr_{\wa^m,j_m}(\mathfrak D^m\setminus (\rho^{m})^{-1}(0))}h_c\di\leb 
						&= \int_{\arr_{\wa^m,j_m}(\mathfrak D^m)\setminus \arr_{\wa^m,j_m}(\mathfrak D^m\setminus (\rho^{m})^{-1}(0))}\Indi_{\arr_{\wa^m,j_m}(\mathfrak D^m)} h_c\di\leb \\
						&\overset{\eqref{eq: subclaim: FlowDecompSingleCycleh}}{\leq}\int_{\arr_{\wa^m,j_m}(\mathfrak D^m \cap(\rho^{m})^{-1}(0) )} \Indi_{\arr_{\wa^m,j_m}(\mathfrak D^m)}h_c\di\leb \\
						&=\int_{\arr_{\wa^m,j_m}(\mathfrak D^m \cap(\rho^{m})^{-1}(0) )} \ell_{\wa^m,j_m}(\hat h^m)\di\leb \\
						&= \int_{\arr_{\wa^m,j_m}^{-1}(\arr_{\wa^m,j_m}(\mathfrak D^m \cap(\rho^{m})^{-1}(0) ))}  \hat{h}^m\di\leb \\
						\overset{\text{\Crefshort{lem: elluinj}}}&{=} \int_{\mathfrak D^m \cap(\rho^{m})^{-1}(0) }  \hat{h}^m\di\leb = 0
					\end{align*}
					Here, the last equality holds by definition of $\rho^m$. 
					Hence, by nonnegativity of $h_c$, it follows that $h_c = 0$ on $\arr_{\wa^m,j_m}(\mathfrak D^m)\setminus \arr_{\wa^m,j_m}\big(\mathfrak D^m\setminus(\rho^{m})^{-1}(0)\big)$. Using this allows us to simplify  \eqref{eq: subclaim: FlowDecompSingleCycleh: 1} further:
					
					\begin{align*}
						\eqref{eq: subclaim: FlowDecompSingleCycleh: 1} = \Indi_{c}(\arc) \cdot \sum_{m \in M}	 \Indi_{\arr_{\wa^m,j_m}(\mathfrak D^m)} \cdot h_c =  \Indi_{c}(\arc) \cdot h_c
					\end{align*}
					where we used for the last equality that $\bigcup_{m\in M}\arr_{\wa^m,j_m}(\mathfrak D^m)$ equals $\mathfrak T_c$ up to a null set and $h_c = 0$ on $\R \setminus \mathfrak T_c$ by definition of the latter set and $h_c$ being a zero-cycle inflow.

					\proofitem{\ref{subclaim: FlowDecompSingleCycleh':Prop}} This is a direct consequence of the above proof where we showed that 
					$\sum_{\wa \in \sdRoutes}\ell_{\wa,\arc}(h'_\wa) = \sum_{\wa \in \sdRoutes}\ell_{\wa,\arc}(h_\wa) + \Indi_c(\arc)\cdot h_c$ holds for all $\arc \in \GA$. This together with $h_c \geq 0$,  $h$ fulfilling \eqref{eq: FlowDecompEachCycle} and the definition of $h'_{\tilde{c}} = h_{\tilde{c}}, \tilde{c}\neq c$ and $h'_c = 0$ implies that $h'$ fulfills \eqref{eq: FlowDecompEachCycle} as well.\qedhere
				\end{structuredproof}
			\end{proofClaim}
			From \Cref{subclaim: FlowDecompSingleCycleh'}, the statement in \Cref{claim: FlowDecompSingleCycle} follows immediately. 
		\end{proofClaim}
		As argued before, the existence of a pure flow decomposition of~$\eflow$ now follows by successively applying \Cref{claim: FlowDecompSingleCycle} for all $c\in \SimpCyc$ to the flow decomposition~$h$ satisfying~\eqref{eq: FlowDecompEachCycle}. This concludes the proof of the \namecref{thm: PureFlowDecomp}. \qedhere
	\end{structuredproof}
\end{proof}

We finish this section with 
the following corollary implied by \Cref{thm: PureFlowDecomp}. It  
states that  for any   $\source$,$\dest$-flow $\eflow$, we can find a ``maximally pure'' flow decomposition $h'_\wa,\wa \in  {\hat{\Routes}},h'_c,c \in \mathcal{C}$ in the sense that any flow on a zero-cycle induced by some $h_c'$ for a $c \in \mathcal{C}$ is not induceable via a flow on an \stwalk.
 We only state the version for \auto setting here -- as before, the non-\auto case then follows directly from it by \Cref{lem: NonParaIsPara}.

\begin{corollary} \label{cor: PureFlowDecompIntuitive}
    Any \auto $\source$,$\dest$-flow $\eflow \in L_+(\hori)^\GA$ with outflow rate $\inflow_\dest$ has 
    a maximally pure \auto flow decomposition $h' \in L_+(\hori)^{\hat{\Routes}\cup \mathcal{C}}$, i.e.~\aauto flow decomposition which fulfills  for all $c \in \mathcal{C }$ and almost all $t \in \hori$ with $h_c'(t)>0$ that neither of the following conditions is fulfilled:
\begin{thmparts}
    \item $\dest \in c$ and  $\inflow_\dest (t)<0$. \label[thmpart]{cor: PureFlowDecompIntuitive: DestCor} 
    \item There exists an edge $\arc=(v,v') \notin c$ with $v \in c$ and  $\sum_{\wa \in \hat{\Routes}}\ell_{\wa,\arc}(h'_{\wa})(t)>0$. \label[thmpart]{cor: PureFlowDecompIntuitive: NotDestCor}
\end{thmparts}
\end{corollary}

Similar to \Cref{thm: PureFlowDecompIntuitive,thm: PureFlowDecomp}, we prove the above corollary by first showing  the following analogue version involving the sets defined in \Cref{def: ConnectedComp}.

\begin{corollary} \label{cor: PureFlowDecomp}
    Consider \aauto $\source$,$\dest$-flow $\eflow \in L_+(\hori)^\GA$ with outflow rate $\inflow_\dest$ and an arbitrary representative of a corresponding \auto flow decomposition~$h$ together with the sets defined in \Cref{def: ConnectedComp}. 
    Then, there exists 
    another \auto flow decomposition $h'_\wa,\wa \in  {\hat{\Routes}},h'_c,c \in \mathcal{C}$ with a corresponding representative and sets $\mathcal{C}'(t),t \in \hori,C'_{\n1'},\mathfrak T'_{C'_{\n1'}},\n1' \in \capn1'$ from \Cref{def: ConnectedComp} such that for every $\n1' \in \capn1'$ 
    and almost every $t \in \mathfrak T_{C'_{\n1'}}$ neither of the following statements holds:
\begin{thmparts}
    \item $\dest \in \GV_{C_{\n1'}'}$ and  $\inflow_\dest (t)<0$. \label[thmpart]{cor: PureFlowDecomp: DestCor} 
    \item There exists an edge $\arc=(v,v') \in \edgesFrom{C_{\n1'}'}$ with  $\eflow_\arc(t)>0$.  \label[thmpart]{cor: PureFlowDecomp: NotDestCor}
\end{thmparts}
Moreover, $h_c'(t)\leq h_c(t)$ for all $t \in \hori$ and $c \in \mathcal{C}$, and, for every $\n1' \in \capn1'$ there exists $\n1 \in \capn1$ with ${C}'_{{\n1}'} = C_{\n1}$ and ${\mathfrak T}'_{ C'_{{\n1'}}} \subseteq \mathfrak T_{C_{\n1}}$.
\end{corollary}

The proof idea is relatively straight forward: We define  new zero-cycle inflow rates $\hat{h}_c,c \in \mathcal{C}$ by setting $h$ to zero whenever neither \ref{cor: PureFlowDecomp: DestCor} or \ref{cor: PureFlowDecomp: NotDestCor} is fulfilled. These zero-cycle inflow rates combined with the original $\source$,$\dest$-walk inflow vector of $h$ induce an  $\source$,$\dest$-flow which has a pure flow decomposition $\tilde{h}$ as 
the zero-cycle inflow rates $\hat{h}_c,c \in \mathcal{C}$  fulfill, by construction, the condition in  \Cref{thm: PureFlowDecomp}.  
Adding to $\tilde{h}$ the difference $h_c - \hat{h}_c$  then yields the desired flow decomposition.

\begin{proof}
    For all $\n1 \in \capn1$  and $t \in \mathfrak T_{C_{\n1}}$ for which either \ref{cor: PureFlowDecomp: DestCor} or \ref{cor: PureFlowDecomp: NotDestCor} is fulfilled \wrt $h$, define 
    $\hat{h}_c (t):= h_c(t), c \in C_{\n1}$. 
    Otherwise, i.e.~for any pair $(c,t) \in \mathcal{C}\times \hori$ for which there does not exist an $\n1 \in \capn1$ such that $t \in \mathfrak T_{C_{\n1}}$ and $c \in C_{\n1}$ and one of \ref{cor: PureFlowDecomp: DestCor} or \ref{cor: PureFlowDecomp: NotDestCor} is fulfilled \wrt $h$, set $\hat{h}_c(t) :=0$. 
    Furthermore define $\hat{h}_\wa(t) := h_\wa(t),t \in \hori, \wa \in \hat{\Routes}$. 
    Then, $\hat{h}_\wa,\wa \in \hat{\Routes},\hat{h}_c, c \in \mathcal{C}$ are a flow decomposition of $\hat{\eflow}_\arc:= \sum_{\wa \in \hat{\Routes}}\ell_{\wa,\arc}(\hat{h}_\wa) +  \sum_{c \in \mathcal{C}:\arc \in c}\hat{h}_c,\arc \in \GA$. Moreover, 
    we can make the following observation: 
    \begin{claim}\label{claim: MaxPureFlow1}
        Consider the sets $\hat{\mathcal{C}}(t),t \in \hori$ and  $\hat{C}_{\hat{\n1}}, \hat{\mathfrak T}_{\hat{ C}_{\hat{\n1}}}, \hat{\n1} \in \hat{\capn1}$ defined in \Cref{def: ConnectedComp} \wrt $\hat{h}$. Then, for every $\hat{\n1 }\in \hat{\capn1}$, there exists $\n1 \in \capn1$ such that $\hat{C}_{\hat{\n1}} = C_{\n1}$ and $\hat{\mathfrak T}_{\hat C_{\hat{\n1}}} \subseteq \mathfrak T_{C_{\n1}}$.
    \end{claim}
\begin{proofClaim}
        Consider an arbitrary  $\hat{\n1}\in \hat{\capn1}$ and $t \in \hat{\mathfrak T}_{\hat{C}_{\hat{\n1}}}$. 
As $\hat{C}_{\hat{\n1}}$  is a connected component of $\hat{\mathcal{C}}(t)$, it is also connected in 
 $\mathcal{C}(t)$ since $\hat{\mathcal{C}}(t) \subseteq \mathcal{C}(t)$ due to $\hat{h}_{c}(t) \leq h_{c}(t),c \in \mathcal{C}$. 
    Hence, there has to exist an $\n1 \in \capn1$ such that $\hat{C}_{\hat{\n1}} \subseteq C_{\n1}$ and $t \in \mathfrak T_{C_{\n1}}$. 
    With this, taking any cycle $\hat c \in \hat C_{\hat{\n1}}$ gives us a cycle $\hat c \in C_{\n1}$ with $\hat h_{\hat c}(t) > 0$. Due to the definition of~$\hat h$, this implies that at least one of \ref{cor: PureFlowDecomp: DestCor} or \ref{cor: PureFlowDecomp: NotDestCor} must have been satisfied for $n$ and $t$ \wrt $h$. But then, the definition of~$\hat h$ also ensures that we have $\hat h_c(t) = h_c(t) > 0$ for all $c \in C_n$ and, therefore, $C_{\n1} \subseteq \hat C_{\hat{\n1}}$ and subsequently $C_{\n1} =\hat C_{\hat{\n1}}$.

    Hence, we have shown that  we can find for any  pair of $\hat{\n1}\in \hat{\capn1}$ and $t \in \mathfrak T_{\hat{C}_{\hat{\n1}}}$ a corresponding $\n1 \in \capn1$ with  $\hat{C}_{\hat{\n1}} = C_{\n1}$ and $t \in \mathfrak T_{C_{\n1}}$. 
     The claim then follows since $C_{\n1_1}\neq  C_{\n1_2}$ for any $\n1_1\neq \n1_2 \in \capn1$. 
\end{proofClaim}

As an immediate consequence of this claim and the definition of $\hat{h}$, we get that $\hat{\eflow}$ and $\hat{h}$ fulfill the conditions stated in \Cref{thm: PureFlowDecomp}. The latter implies that $\hat{\eflow}$ has a pure flow decomposition $\tilde{h}$, i.e.\ $\tilde{h}_{c}=0,c\in\mathcal{C}$. 
Define $h'_\wa := \tilde{h}_\wa,\wa \in \sdRoutes$ and $h_c':= h_c - \hat{h}_c, c \in \SimpCyc$. Then, $h'$ is a flow decomposition of~$\eflow$ by definition of $\hat{h},\hat{\eflow}$ and $\tilde{h}$. Next, we require the following claim:
\begin{claim}\label{claim: MaxPureFlow2}
     Consider the sets ${\mathcal{C}'}(t),t \in \hori$ and  ${C'}_{{\n1'}}, {\mathfrak T'}_{ { C'}_{ {\n1'}}}, {\n1'} \in  {\capn1'}$ defined in \Cref{def: ConnectedComp} with respect to~${h'}$. Then, for every ${\n1 '}\in  \capn1'$, there exists $\n1 \in \capn1$ such that $C'_{{\n1'}} = C_{\n1}$ and $\mathfrak T'_{C'_{\n1'}} \subseteq \mathfrak T_{C_{\n1}}$.
\end{claim}
\begin{proofClaim}
Consider an arbitrary $\n1' \in \capn1'$ and $t \in \mathfrak T'_{C'_{\n1'}}$. 
As ${C'}_{{\n1'}}$  is a connected component of ${\mathcal{C}'}(t)$, it is also connected in 
 $\mathcal{C}(t)$ since ${\mathcal{C}'}(t) \subseteq \mathcal{C}(t)$ due to $h'_{c}(t) \leq h_{c}(t),c \in \mathcal{C}$. 
    Hence, there has to exist a $\n1 \in \capn1$ such that $C'_{\n1'} \subseteq C_{\n1}$ and $t \in \mathfrak T_{C_{\n1}}$.  

We argue now that even ${C}'_{{\n1}'} = C_{\n1}$ holds. 
Assume for the sake of a contradiction that this would not hold. Then, there exists a cycle $c \in C_{\n1}\setminus {C}'_{{\n1}'}$ with   $h_c(t)>0$ but also $h'_c(t)=0$. This, in turn, implies $\hat{h}_c(t) = h_c(t)>0$ and hence $c \in \hat{C}_{\hat{n}}$ for some $\hat{n}$. \Cref{claim: MaxPureFlow1} together with $c \in C_{\n1}$ then implies $ \hat{C}_{\hat{n}} = C_{\n1}$ since we have $C_{\n1_1}\neq  C_{\n1_2}$ for any $\n1_1\neq \n1_2 \in \capn1$.
 From this, we get  for all $c \in C_{\n1}$ that  $\hat{h}_c(t)>0$  and, by definition of $\hat{h}$, even $\hat{h}_c(t) = h_c(t)$. This, however, shows that $h'_c(t)=0$ for all $c \in C_{\n1}$ which is not possible as ${C}'_{{\n1}'} \subseteq C_{\n1}$ and ${C}'_{{\n1}'} \neq \emptyset$ (by definition). 
\end{proofClaim}
The ``moreover'' part of the \namecref{cor: PureFlowDecomp} is now a direct consequence of the above claim and the definition of $h'$. 
For the first part of the statement, consider an arbitrary  $\n1' \in \capn1'$ 
    as well as  $t \in \mathfrak T'_{C'_{\n1'}}$.  
    By \Cref{claim: MaxPureFlow2}, there exists $\n1 \in \capn1$ such that $C'_{{\n1'}} = C_{\n1}$ and $t \in \mathfrak T_{C_{\n1}}$. 
    By  $t \in \mathfrak T'_{C'_{\n1'}}$, we have  $h'_c(t)>0$ and subsequently $h_c(t) > \hat{h}_c(t)$ has to hold. 
 In particular, by the definition of $\hat{h}$, this shows that neither \ref{cor: PureFlowDecomp: DestCor} nor \ref{cor: PureFlowDecomp: NotDestCor}  is fulfilled \wrt $h$ for $\n1,t,C_{\n1}$. 
 But since $C_{\n1} = C'_{{\n1'}}$,  this shows that neither \ref{cor: PureFlowDecomp: DestCor} nor \ref{cor: PureFlowDecomp: NotDestCor}  is fulfilled \wrt $h'$ for $\n1',t,C'_{\n1'}$ either. This completes the proof.  
 \end{proof}
With \Cref{cor: PureFlowDecomp} at hand, the proof of \Cref{cor: PureFlowDecompIntuitive} is straight forward as we only need to show that the flow decomposition constructed in \Cref{cor: PureFlowDecomp} is a maximally pure one, i.e.~fulfills the condition in \Cref{cor: PureFlowDecompIntuitive}:

\begin{proof}[Proof of \Cref{cor: PureFlowDecompIntuitive}] 
	Let $h$ be (an arbitrary representative of) a flow decomposition of $\eflow$ which exists by \Cref{thm: FlowDecomp}. Furthermore, let $h'$ be  the representative of the flow decomposition constructed in \Cref{cor: PureFlowDecomp}. 
	We argue in the following via a  proof by contradiction that this $h'$ fulfills the conditions stated in \Cref{cor: PureFlowDecompIntuitive}: 
	Assume for the sake of a contradiction that there exists $c \in \SimpCyc$ and a measurable set $\mathfrak T$ with $\leb(\mathfrak T)>0$ for which $h'_c(t)>0$ and either of the conditions \Cref{cor: PureFlowDecomp: DestCor} or \Cref{cor: PureFlowDecomp: NotDestCor} holds. It is sufficient to consider the cases where either one of the conditions is satisfied for all of $\mathfrak T$. 
	Similarly, since there are only finitely many sets $C'_{\n1'},\n1'\in\capn{1}'$, there exists $\n1' \in \capn1'$ with $\leb(\mathfrak T_{C'_{\n1'}}\cap \mathfrak T)>0$ and we can assume \wlg that $\mathfrak T \subseteq \mathfrak T_{C'_{\n1'}}$.

	\begin{structuredproof}
		\proofitem{\Cref{cor: PureFlowDecomp: DestCor}} If \Cref{cor: PureFlowDecomp: DestCor} was fulfilled during $\mathfrak T$, then $\dest \in \GV_{C'_{\n1'}}$ 
		and subsequently \Cref{cor: PureFlowDecompIntuitive: DestCor} is fulfilled for all $t \in\mathfrak T$, contradicting that $h'$ fulfills the 
		properties stated in \Cref{cor: PureFlowDecompIntuitive}.

		\proofitem{\Cref{cor: PureFlowDecomp: NotDestCor}}
		Assume that 
		 \Cref{cor: PureFlowDecomp: NotDestCor} was fulfilled during $\mathfrak T$ but \Cref{cor: PureFlowDecomp: DestCor} is not. 
		 We can assume \wlg (by finiteness of $\GA$) that \Cref{cor: PureFlowDecomp: NotDestCor} is fulfilled \wrt one edge $\arc =(v,v')\in \edgesFrom{c}$ with  $\sum_{\wa \in \sdRoutes}\ell_{\wa,\arc}(h'_{\wa})(t)>0$ for all $t\in \mathfrak T$. Since $\sdRoutes$ is countable, there has to exist  $\wa \in \sdRoutes$ and $j\leq \abs{\wa}$ with 
		$\ell_{\wa,j}(h'_\wa)(t)>0$ for a non-null subset of $\mathfrak T$ and we can assume \wlg on all of $\mathfrak T$. 
		Let $j^*>j$ be the first index with $\wa[j^*]\notin \GA_{C_{\n1'}'}$ (and subsequently $\wa[j^*] \in \edgesFrom{C_{\n1'}'}$). Such a $j^*$ has to exists since $\wa$ is an \stwalk{} and 
		\Cref{cor: PureFlowDecomp: DestCor} does not hold. 
		By the travel times on all arcs $\{\wa[j],\ldots,\wa[j^*-1]\}\subseteq \GA_{C'_{\n1'}}$ being equal to $0$ (due to $\mathfrak T \subseteq  \mathfrak T_{ C_{\n1'}'}$), 
		it follows by \Cref{lem: FLowOnZeroTrav,lem: elluPropagation} that 
		$\ell_{\wa,j^*}(h'_\wa)(t)=\ell_{\wa,j}(h'_\wa)(t)>0$ for almost all $t \in \mathfrak T$.  Since $h'$ is in 
		particular a flow decomposition of $\g$, we hence arrive at $\g_{\wa[j^*]}(t) \geq \ell_{\wa,j^*}(h'_\wa)(t) >0$ for almost all $ t \in \mathfrak T$. That is, \Cref{cor: PureFlowDecomp: NotDestCor} 
		is fulfilled for almost all $t \in \mathfrak T$ and edge $\wa[j^*]$, in contradiction to $h'$ fulfilling the 
		properties stated in \Cref{cor: PureFlowDecompIntuitive}.  
		\qedhere
	\end{structuredproof}
\end{proof}

\section{Equivalence of Dynamic Equilibrium Definitions}\label{sec:EquivalenceOfDE}

\newcommand{\cost}{C}
\newcommand{\setAE}{E}

As an application of our decomposition theorem, we now discuss the equivalence of walk- and edge-based definitions for dynamic equilibria. Specifically, we consider a network where, in addition to the traversal time functions $\trav_\arc(\cdot,\cdot)$, every edge also comes with a (flow and time dependent) cost function $\cost_\arc(\cdot,\cdot)$ such that $\cost_\arc(\g,t) \in \Rnn$ denotes the (nonnegative) cost of traversing edge~$\arc$ under the edge flow~$\g$ when entering edge~$\arc$ at time~$t$.

For any edge $\source$,$\dest$-flow $\g \in L_+(\R)^\GA$ and walk $\wa$ we then define the cost of traversing this walk at time~$t$ via 
    \[\cost_\wa(\g,t) \coloneqq \sum_{j=1}^{\abs{\wa}}\cost_{\wa[j]}(\g,\arr_{\wa,j}(\g,t)).\]
Furthermore, we define remaining cost labels $R_v(\g,t)$ for any node $v \in V$ by
    \[R_v(\g,t) \coloneqq \inf\Set{\cost_\wa(\g,t) | \wa \text{ a $v$,$\dest$-walk}}.\]

For walk inflows the most common equilibrium definition is then to require Wardrop's first principle (\cite[p.~345]{Wardrop52}) to hold, i.e.\ flow particles only ever enter cheapest \stwalk s (cf.\ e.g.\ \cite[Definition~4.1]{ZhuM00}, \cite[Definition~1]{Koch11}):

\begin{definition}
    We call a walk inflow $h \in L_+(\R)^{\Routes}$ a \emph{walk-based dynamic equilibrium} if it induces an edge $\source$,$\dest$-flow $\g$ and satisfies
        \[h_\wa(t) > 0 \implies \cost_\wa(\g,t) = R_s(\g,t)\]
    for all $\wa \in \Routes$ and almost all $t \in \R$.
\end{definition}

A natural way of transferring this concept to edge flows is to require flow particles to only ever enter edges which are the start of some cheapest walk towards the destination.

\begin{definition}
    We call an edge $\source$,$\dest$-flow $\g \in L_+(\R)^\GA$ an \emph{edge-based dynamic equilibrium} if it satisfies
        \[\g_\arc(t) > 0 \implies \arc \in \setAE(\g,t)\]
    for all $\arc \in \GA$ and almost all $t \in \R$, where
        \[\setAE(\g,t) \coloneqq \Set{\arc = (v,v') \in \GA | R_v(\g,t)(t) \geq \cost_\arc(\g,t) + R_{v'}(\g,\exit_\arc(\g,t))}\]
    denotes the set of all \emph{active} edges at time~$t$.
\end{definition}

Note that, for the case of costs being equal to travel times, the labels are typically defined starting at the source (i.e.\ denoting the earliest arrival time at intermediate nodes instead of the shortest remaining travel time), cf.\ e.g.\ \cite[Definition~2]{Koch11} or \cite[Lemma~1]{CominettiCL15}. For the case of more general costs as considered here, however, we have to use the version with destination-based labels. 

We now aim to show that these two definitions are equivalent in the sense that for every walk-based dynamic equilibrium there exists a corresponding edge-based dynamic equilibrium and vice versa. Here, we say that a walk inflow~$h$ and an edge flow~$\g$ \emph{correspond to each other} if they satisfy equation~\eqref{eq: DefFlowDecomp}, i.e.\ if $h$ is a flow decomposition of~$\g$.

In order to show this equivalence, we first need to establish equivalence of the definitions for given corresponding flows. The desired result then follows by a simple application of our flow decomposition theorem. 

\begin{lemma}\label{lemma:EquivalenceEquilibriumDefinitionsGivenFlows}
    Consider a network with edge cost functions such that for any edge flow~$\g$ there are no cycles of cost~$0$. Let $h$ be a walk inflow and $\g$ a corresponding edge $\source$,$\dest$-flow. Then $h$ is a walk-based dynamic equilibrium if and only if $\g$ is an edge-based equilibrium.
\end{lemma}

Note that, a similar statement specifically for the Vickrey queuing model and dynamic equilibria traversal times as cost functions can be found in \cite[Theorem 4.13 and Lemma~4.14]{KochThesis} and \cite[Theorem~1]{Koch11}.

\begin{proof}
    We start by showing that cheapest walks are exactly those walks consisting only of active edges:
    \begin{claim}\label{claim:activePathsAreMinimalPaths}
        Let $\g$ be an edge $\source$,$\dest$-flow and $\wa$ any $v$,$\dest$-walk. Then, for any time $t \in \hori$ we have
            \[\cost_\wa(\g,t) = R_v(\g,t) \iff \forall j \in [\abs{\wa}]: \wa[j] \in \setAE(\g,\arr_{\wa,j}(\g,t)).\]
    \end{claim}
    
    \begin{proofClaim}
    For all edges $\arc=({\tilde{v}},v')$, we have 
        by the definition of the node labels $R_{\tilde{v}}$ (and triangle-inequality) that $R_{\tilde{v}}(g,t)(t) \leq \cost_\arc(g,t) + R_{v'}(g,\exit_\arc(g,t))$   for all times $t \in \hori$ with equality if and only if $({\tilde{v}},v')$ is active at that time. Hence, using $\wa[j] \eqqcolon (v_{j-1},v_j)$ we get
            \begin{align*}
                \cost_\wa(g,t) 
                    &= \sum_{j=1}^{\abs{\wa}}\cost_{\wa[j]}(\g,\arr_{\wa,j}(\g,t)) \\
                    &\symoverset{1}{\geq} \sum_{j=1}^{\abs{\wa}}\left(R_{v_{j-1}}(\g,\arr_{\wa,j}(\g,t))-R_{v_j}(\g,\exit_{\wa[j]}(\g,\arr_{\wa,j}(\g,t)))\right) \\
                    &= \sum_{j=1}^{\abs{\wa}}\left(R_{v_{j-1}}(\g,\arr_{\wa,j}(\g,t))-R_{v_j}(\g,\arr_{\wa,j+1}(\g,t))\right) \\
                    &= R_{v_0}(\g,\arr_{\wa,1}(\g,t))-R_{v_{\abs{\wa}+1}}(\g,\arr_{\wa,\abs{\wa}+1}(\g,t)) \\
                    &= R_v(\g,t) - R_\dest(\g,\arr_{\wa,\abs{\wa}+1}(\g,t)) = R_v(\g,t) - 0
            \end{align*}
        with equality at~\refsym{1} if and only if all edges along that walk are active at the respective entry times.
    \end{proofClaim}

    With this, we can now deduce the \namecref{lemma:EquivalenceEquilibriumDefinitionsGivenFlows}'s statement:
    
   \begin{structuredproof}
        \proofitem{``if''} Let $\wa \in \Routes$ be any $\source$,$\dest$-walk. Fix some representatives of $h$ and $\g$ such that we have 
            \[h_\wa(t) > 0 \implies \g_{\wa[j]}(\arr_{\wa,j}(\g,t) > 0\]
        for all $j \in [\abs{\wa}]$ and \emph{all} times $t \in \hori$ (this is possible by \Cref{lem: Relations:h>0u>0:Pointwise}). Then, define the sets
            \[\mathfrak T_j \coloneqq \Set{t \in \hori | g_{\wa[j]}(t) > 0 \text{ and } \wa[j]\notin\setAE(g,t)}\]
        of times where $\g$ violates the equilibrium condition on the $j$-th edge of walk~$\wa$. Since $\g$ is an edge-based dynamic equilibrium, all these sets must be null sets. Hence, by \Cref{lem: elluExistenceProperties:ExistenceInducedFlowOnAllEdges}, we get
            \[h_\wa(t) = 0 \text{ for almost all } t \in \bigcup_{j=1}^{\abs{\wa}}\arr_{\wa,j}(\g,\cdot)^{-1}(\mathfrak T_j).\]
        In other words, for every $j \in [\abs{\wa}]$ and almost all $t \in \hori$ with $h_\wa(t) > 0$ we have $t \notin \arr_{\wa,j}(g,\cdot)^{-1}(\mathfrak T_j)$ and, therefore, $g_{\wa[j]}(\arr_{\wa,j}(\g,t)) = 0$ or $\wa[j] \in E(g,\arr_{\wa,j}(g,t))$. But since our choice of representatives for $h$ and $\g$ precludes the former, the latter must hold for almost all $t \in \hori$ with $h_\wa(t) > 0$ (and all $j \in [\abs{\wa}]$). With this, \Cref{claim:activePathsAreMinimalPaths} implies that we have $\cost_\wa(g,t) = R_s(g,t)$ for almost all such times~$t$ and, thus, $h$ is a walk-based dynamic equilibrium.

        \proofitem{``only if''} We fix representatives of $h$ and $\g$ such that the equilibrium condition for $h$ hold for \emph{all} times. Now, let $\arc \in \GA$ be any edge and $\mathfrak T \coloneqq \set{t \in \hori | g_\arc(t) > 0}$ the set of times where flow enters this edge. If $\mathfrak T$ has measure zero, there is nothing to show. Otherwise, \Cref{lem: Relations:h>0u>0:u>0ExistsHw>0} gives us for almost every $t \in \mathfrak T$ some \stwalk~$\wa\in\Routes$ and $j \in [\abs{\wa}]$ with $\wa[j]=\arc$ such that there exists some $\tilde t \in \arr_{\wa,j}(g,\cdot)^{-1}(t)$ with $h_\wa(\tilde t) > 0$. Since $h$ satisfies the equilibrium condition everywhere, this implies $\cost_\wa(g,\tilde t) = R_s(g,\tilde t)$ and, hence, $\wa[j] \in \setAE(g,\arr_{\wa,j}(g,\tilde t)) = \setAE(g,t)$ by \Cref{claim:activePathsAreMinimalPaths}. Therefore, $\g$ is an edge-based dynamic equilibrium. \qedhere
    \end{structuredproof}   
\end{proof}

\begin{corollary}\label{thm:def-equilibrium-equivalent}
    Given a network where the destination~$\dest$ is reachable from every node~$v \in \GV$ and with nonnegative edge cost functions such that for any fixed edge flow there are no cycles of cost~$0$. Then the following holds:
    \begin{enumerate}
        \item For any edge-based dynamic equilibrium $\g$, there exists a corresponding walk inflow~$h$ which is a walk-based equilibrium.
        \end{enumerate}
        Moreover, if   every walk inflow has a corresponding network loading, the following holds:
        \begin{enumerate}[resume]
        \item For any walk-based dynamic equilibrium $h$, there exists a corresponding edge flow~$\g$ which is an edge-based dynamic equilibrium.
    \end{enumerate}
\end{corollary}

\begin{proof}
    \Cref{lemma:EquivalenceEquilibriumDefinitionsGivenFlows} together with the assumption that network loadings exist, directly implies the second part of the corollary. 
    For the first part we apply our decomposition theorem (\Cref{thm: FLowDecompModel}) to obtain a decomposition of $\g$ into walk- and (zero)-cycle inflow rates $h$. If $h$ is in fact a pure walk inflow, we are done by \Cref{lemma:EquivalenceEquilibriumDefinitionsGivenFlows}. Thus, assume for contradiction that this is not the case. That is, there is some positive zero-cycle inflow $h_c$ on some set of positive measure $\mathfrak T$. Then, we must have 
        \[g_\arc(t) = \ell_\arc(h)(t) \geq \ell_\arc(h_c)(t) \Croverset{lem: FLowOnZeroTrav}{=} h_c(t) > 0\]
    for almost all $t \in \mathfrak T$ and $\arc \in c$. Since $\g$ is an edge-based equilibrium, this implies that all edges on~$c$ are active for almost all $t \in \mathfrak{T}$. Since $\dest$ is reachable from every node, the node labels~$R_v(g,t)$ are all finite and, hence, a zero-cycle of only active edges is only possible if the costs on all these edges are zero -- a contradiction to the assumptions of this \namecref{thm:def-equilibrium-equivalent}.
\end{proof}

\section{Structural Properties of \Auto Network Loadings} 
\label{sec: ANLResultsFlowDecomp}

In this section, we develop the theory of \auto network loadings required to derive our results. In particular, we provide formal proofs of the central insights on \auto network loadings stated in \Cref{sec:uBasedNetworkLoadings}.
We start with an example demonstrating that  not every walk inflow rate necessarily induces an edge flow \wrt arbitrary fixed traversal times $\trav$ that can be described via a vector ${\g} \in L_+(\hori)^\GA$. This is true even in case of traversal times $\trav$ being induced by some other flow in a well known flow propagation model like the Vickrey model.

 \begin{example}\label{exa: noarcflow}
Consider the  network depicted in \Cref{fig: NLexa}. As flow propagation model we use the Vickrey queuing model  with free flow travel times and service rates given by  $\tau_\arc$ and $\nu_\arc$ on the edges.
 \begin{figure}[h!]     
        \begin{center}
        \BigPicture[1]{%
\begin{tikzpicture}
	\node[namedVertex] (v) at (0,0) {$s$};
	\node[namedVertex] (w) at (4,0) {$v$};	
 \node[namedVertex] (t) at (8,0) {$\dest$};
	\draw[edge, bend left] (v) to node[below]{$\arc_1$} node[above,pos =0.5]{$\tau_{\arc_1} = 1,\nu_{\arc_1}  = 1$} (w);
\draw[edge,bend right] (v) to node[above]{$\arc_2$}node[below,pos =0.5]{$\tau_{\arc_2} = 1,\nu_{\arc_2}  = 2$}  (w);
\draw[edge] (w) to node[below]{$\arc_3$}node[above,pos =0.5]{$ \tau_{\arc_3} = 1,\nu_{\arc_3}  = 4$}  (t);
\node[above left =1cm of v,xshift=-1cm,  anchor=west, node distance=1cm,blue]  {$h_{\wa_1} = 2_{[0,1]}$};
\node[below left =1cm of v,xshift=-1cm,  anchor=west, node distance=1cm,blue]  {$h_{\wa_2} = 2_{(1,2]}$};

    \begin{scope}[xshift=-2.5cm,yshift=3cm]
        \begin{axis}[anchor=center, xmin=0,xmax=3.2,ymax=4.2, ymin=0, samples=500,width=4.5cm,height=3.5cm,
				        axis y line*=left, axis lines=left, xtick={1,2,3},ytick={0,1,2,3}]
            \addplot[blue,  thick,samples at={0,...,1},-]  (\x,\x)  ;
            \addplot[blue,  thick,samples at={1,...,2},-]  (\x,2-\x)node[above right,pos=0.5]{$q_{\arc_1}(\g,\cdot)$};
            \addplot[blue, very thick,samples at={2,...,6},-]  (\x,0)  ;
        \end{axis}
    \end{scope}

    \begin{scope}[xshift=2cm,yshift=3cm]
        \begin{axis}[anchor=center, xmin=0,xmax=3.2,ymax=4.2, ymin=0, samples=500,width=4.5cm,height=3.5cm,
				         axis y line*=left, axis lines=left,x label style={at={(axis description cs:1,-0.1)},anchor=north}, xtick={1,2,3}, ytick={0,1,2,3}]
            \addplot[blue,  thick,samples at={0,...,1},-]  (\x,1 +\x)  ;
            \addplot[blue,  thick,samples at={1,...,2},-]  (\x,3-\x) ;
            \addplot[blue, thick,samples at={2,...,6},-]  (\x,1) node[above ,pos=0.1]{$\trav_{\arc_1}(\g,\cdot)$} ;
        \end{axis}
    \end{scope}

    \begin{scope}[xshift=6.5cm,yshift=3cm]
        \begin{axis}[anchor=center,xmin=0,xmax=3.2,ymax=4.2, ymin=0, samples=500,width=4.5cm,height=3.5cm,
				    axis y line*=left, axis lines=left,x label style={at={(axis description cs:1,-0.1)},anchor=north}, xtick={1,2,3}, ytick={0,1,2,3}]
            \addplot[blue,  thick,samples at={0,...,1},-]  (\x,1 +2*\x)  ;
            \addplot[blue,  thick,samples at={1,...,2},-]  (\x,3)node[below ,pos=0.8]{$\arr_{\wa_1,2}(\g,\cdot)$};
            \addplot[blue,  thick,samples at={2,...,6},-]  (\x,1+\x)  ;
        \end{axis}
    \end{scope}
\end{tikzpicture}
}
\end{center}
    \caption{Example for the non-existence of edge flows under fixed traversal times.}
    \label{fig: NLexa}
\end{figure}

Consider the edge flow $\g_{\arc_1} = 2_{[0,1]}$, $\g_{\arc_2} = 2_{(1,2]}$, $\g_{\arc_3} = 1_{[1,3]} + 2_{(2,3]}$ induced by the walk inflow rates $h_{\wa_1} = 2_{[0,1]}$, $h_{\wa_2} = 2_{(1,2]}$ (where $\wa_1 = (\arc_1,\arc_3)$ and $\wa_2=(\arc_2,\arc_3)$).  

On edge~$\arc_1$, a queue starts to build in the time interval from $0$ to $1$ with 
the volume $q_{\arc_1}(\g,t) = t$ for $t \in[0,1]$ and starts to decrease in the time interval $1$ to $2$ with $q_{\arc_1}(\g,t) = 2-t$. 
The resulting travel times on $\arc_1$ and the corresponding arrival times at $\arc_3$ over $\wa_1$ are given by 
\[
    \trav_{\arc_1}(\g,t) = 1 + q_{\arc_1}(\g,t)/\nu_{\arc_1} = \begin{cases}
        1 + t, &\text{if } t \in [0,1] \\
        3 - t, &\text{if } t \in (1,2]
    \end{cases} 
\]
and
\[
    \arr_{\wa_1,2}(\g,t) = t +\trav_{\arc_1}(\g,t)= \begin{cases}
        2+t, &\text{if } t \in [0,1] \\
        3, &\text{if } t \in (1,2]
    \end{cases},
\]
respectively.

Now, consider the walk inflow rate $\Tilde{h}_{\wa_1} := 2_{[0,2]}$ under the fixed traversal times of $\g$ and assume that the \auto network loading $\ell_{\wa_1}(\Tilde{h}_{\wa_1})$ existed. This would imply that 
\begin{align*}
    0 = \int_{\{3\}} \ell_{\wa_1}(\Tilde{h}_{\wa_1}) \di \sigma =  \int_{\arr_{\wa_1,2}(\g,\cdot)^{-1}(3)} \tilde{h}_{\wa_1} \di \sigma =   \int_{[1,2]} 2 \di \sigma = 2,
\end{align*}
which is a contradiction.
Hence, $\ell_{\wa_1}(\Tilde{h}_{\wa_1})$ does not exist. 
 
The reason for the non-existence of an  edge flow induced by $\Tilde{h}_{\wa_1}$ under the fixed traversal times of~$\g$  
is the fact that $\Tilde{h}_{\wa_1}$ sends a nontrivial amount of particles into the walk $\wa_1$ during $[1,2]$. These particles, then, all arrive at the same time $\arr_{\wa_1,2}(\g,t) = 3, t\in (1,2]$ at $\arc_3$.
\end{example}

\subsection{Existence of \Auto Network Loadings}
\Cref{exa: noarcflow} suggests that a necessary condition for a walk inflow rate $h_\wa \in L_+(\hori)$ to induce an edge flow, is that no flow of positive measure is sent into the walk in such a way that these flow particles all arrive at some edge during a null set of times. That is, $h_\wa$ must satisfy the following condition for all $j \leq |\wa|$:
\begin{equation}\label{eq: nlexists}
	h_\wa = 0 \text{ on } \arr_{\wa,j}^{-1}(\mathfrak T) \text{ for every Borel-measurable  $\leb$-null set }\mathfrak T \in \mathcal{B}(\hori).  
\end{equation}
We show in the following \namecref{lem: elluExistenceProperties} that this is also a sufficient condition for the existence 
 $h_\wa$ to induce a corresponding \auto  edge flow.

In \cite{KochThesis} the same condition (called compatibility of $h$ and $\arr_{\wa,j}(h,\cdot)$ there) is stated as an assumption on the flow model itself and is required to hold for \emph{all} walk inflows and  corresponding induced arrival time functions (see \cite[Definition~3.2]{KochThesis} where, in the last paragraph, $\tau_P$ seems to be a typo and should be replaced by~$\ell_P$). 
For \auto network loadings this assumption does not hold in general (see \Cref{exa: noarcflow} for a simple instance with traversal times induced by the Vickrey queuing model where not every walk inflow has a \auto network loading).
However, as the following \namecref{lem: elluExistenceProperties} will show, here this condition can instead be used to exactly characterize those inflow rates $h_\wa \in L_+(\hori)$ into a walk~$\wa$ that do induce a corresponding \auto  edge flow.

To formally state this \namecref{lem: elluExistenceProperties} we will need some additional notation: 
For any walk $\wa$, $j \in [|\wa|]$ and $h_\wa \in L_+(\hori)$, we denote by $\ell_{\wa,j}(h_\wa) \in L_+(\hori)$ the (\auto[)] flow induced by $h_{\wa}$ on the  $j$-th edge of $\wa$ if it exists, i.e.\ a function satisfying
\begin{align}\label{eq: Deflwj}
    \int_{\mathfrak T} \ell_{\wa,j}(h_\wa) \di\sigma=    \int_{\arr_{\wa,j}^{-1}(\mathfrak T)}h_\wa\di\sigma
\end{align}
  for all $\mathfrak T \in \mathcal{B}(\hori)$.  
Remark that the stated equality is equivalent to the equality holding for all intervals of the form ${\startint}t]$.
Analogously, we also define $\ell_{\wa,j}(h_\wa)$ for $j = |\wa|+1$, denoting the inflow 
into the last node of the walk~$\wa$. 
Note that, if an edge~$\arc$  occurs multiple times on~$\wa$, then $\ell_{\wa,j}(h_\wa)$ with $\wa[j] = \arc$ is different to $\ell_{\wa,\arc}(h_\wa)$ (the flow induced by $h_\wa$ on edge~$\arc$) but related to it by $\ell_{\wa,\arc}(h_\wa) = \sum_{j:w[j]=\arc}\ell_{\wa,j}(h_\wa)$.  
From now on, 
we denote by $\edom{\Routes'}$, $\edom{\wa}$, $\edom{\wa,\arc}$ and $\edom{\wa,j}$ 
the maximal domains of the corresponding \auto network loading operator. 
 Equivalently to writing $h_\wa \in \edom{\wa,j}$, we will also say that $\ell_{\wa,j}(h_\wa)$ exists and 
 adopt the analogue convention for  the sets $\edom{\Routes'}$, $\edom{\wa}$ and  $\edom{\wa,\arc}$.

\begin{rsttheorem}{\ref{lem: elluExistencePropertiesShort}}\label{lem: elluExistenceProperties}
    Consider an arbitrary countable collection of walks $\Routes'$, $h \in L_+(\hori)^{\Routes'}$,  $\wa\in \Routes'$, $j \in[|\wa|+1]$ and $\arc \in \GA$. Then, the following holds:
    \begin{thmparts}
        \item $h_\wa \in \edom{\wa,j}$   if and only if $h_\wa$ satisfies~\eqref{eq: nlexists}.  In this case $\ell_{\wa,j}(h_\wa)$ is uniquely determined. \label[thmpart]{lem: elluExistenceProperties:ExistenceInducedFlowOnJthEdge}

        \item 
         $h_\wa \in \edom{\wa,\arc}[]$ if and only if  $h_\wa \in \edom{\wa,j}$ for all $j\leq \abs{\wa}$ with $\wa[j] = \arc$. In this case, $\ell_{\wa,\arc}(h_\wa)$ is uniquely determined by $\ell_{\wa,\arc}(h_\wa) = \sum_{j:\wa[j]=\arc}\ell_{\wa,j}(h_\wa)$.
          \label[thmpart]{lem: elluExistenceProperties:ExistenceInducedFlowOnEdgeE}

        \item  $h_\wa \in \edom{\wa}[]$ if and only if  $h_\wa \in \edom{\wa,j}$ for all $j\leq \abs{\wa}$. In this case, $\ell_{\wa}(h_\wa)$ is uniquely determined by $\ell_{\wa,\arc}(h_\wa) = \sum_{j:\wa[j]=\arc}\ell_{\wa,j}(h_\wa)$ for all $\arc \in \GA$.\label[thmpart]{lem: elluExistenceProperties:ExistenceInducedFlowOnAllEdges}

        \item The maximal domains $\edom{\wa,j}$ and $\edom{\wa,\arc}$ of $\ell_{\wa,j}$ and $\ell_{\wa,\arc}$  are sequentially weakly closed convex cones of $L_+(\hori)$, that is, if $h_\wa^n\wto h_\wa$ and $\ell_{\wa,j}(h_\wa^n),n\in \N$ exist, then so does $\ell_{\wa,j}(h_\wa)$ and, analogously, for $\ell_{\wa,\arc}$.

        Moreover, both functions are linear on their respective domains.
        \label[thmpart]{lem: elluExistenceProperties:Linearity}

        \item  If the aggregated edge flow $\ell_{\Routes'}(h) \in L_+(\hori)^\GA$ exists, then it is uniquely determined. On its maximal domain $\edom{\Routes'} \subseteq L_+(\hori)^{\Routes'}$, the function  $\ell_{\Routes'}$ is  linear. \label[thmpart]{lem: elluExistenceProperties:AggLinearity}

        \item The aggregated edge flow $\ell_{\Routes'}(h) \in L_+(\hori)^\GA$ exists if and only if $\ell_{\wa}(h_\wa)$ exists for all $\wa \in \Routes'$ and $(\ell_{\wa}(h_\wa))_{\wa\in \Routes'} \in \seql[1][\Routes'][L_+(\hori)^\GA]$ holds.   In this case, $\Nl[]_{\Routes'}(h)$ is uniquely determined by $\Nl[]_{\Routes'}(h) = \sum_{\wa \in \Routes'}\Nl[]_\wa(h_\wa)$.
        \label[thmpart]{lem: elluExistenceProperties:ExistenceInducedFlow}

 
    \end{thmparts} 
\end{rsttheorem}

\begin{proof}
We will only prove the statements about $\ell_{\Routes'}$ and $\ell_{\wa,j}$ as the analogues ones about $\ell_{\wa,\arc}$ follow directly from the statements of $\ell_{\wa,j}$. 
    \begin{structuredproof}
        \proofitem{\ref{lem: elluExistenceProperties:ExistenceInducedFlowOnJthEdge},``$\Rightarrow$''} 
        This is a direct consequence of the nonnegativity of $h_\wa$ and the equality in \eqref{eq: Deflwj} as the left-hand side is $0$ for all null sets $\mathfrak T$.

        \proofitem{\ref{lem: elluExistenceProperties:ExistenceInducedFlowOnJthEdge},``$\Leftarrow$''} 
        By setting $\mu_{h}^j({\startint}t]) :=  \int_{\arr_{\wa,j}^{-1}({\startint}t])}h_\wa\di\sigma$ for all $t \in \hori$,  one arrives at a uniquely determined finite Borel measure  $\mu_{h}^j$ that fulfills $\mu_{h}^j(\mathfrak T) =   \int_{\arr_{\wa,j}^{-1}(\mathfrak T)}h_\wa\di\sigma$ for all $\mathfrak T \in \mathcal{B}(\hori)$. 
        Due to the assumption, it follows that $\mu_{h}^j$ is absolutely continuous \wrt $\sigma$. 
        Hence, there exists a uniquely determined  $  \ell_{\wa,j}(h_\wa)  \in L_+(\hori)$ (the Radon-Nikodym derivative of $\mu_{h}^j$) fulfilling the equality   
            \[\int_{\startint t]}   \ell_{\wa,j}(h_\wa)  \di\sigma = \int_{\arr_{\wa,j}^{-1}({\startint}t])}h_\wa\di\sigma\]
        for all $t \in \hori$ (\cite[Theorem 3.2.2]{Bogachev2007I}). Hence, $  \ell_{\wa,j}(h_\wa)$ exists and is uniquely determined. 

 \proofitem{\ref{lem: elluExistenceProperties:Linearity}} Regarding the sequential weak closedness, let $h_\wa^n\wto h_\wa$  with $\ell_{\wa,j}(h_\wa^n),n\in \N$ existing. Consider an arbitrary null set $\mathfrak T\subseteq \hori$. Then $0 = \int_{\arr_{\wa,j}^{-1}(\mathfrak T)} h_\wa^n\di\sigma \to \int_{\arr_{\wa,j}^{-1}(\mathfrak T)} h_\wa \di\sigma$, showing that $h_\wa = 0$ on $\arr_{\wa,j}^{-1}(\mathfrak T)$. Thus, $\ell_{\wa,j}(h_\wa)$ exists  by the first part of the \namecref{lem: elluExistenceProperties}. 

        The property of the domain of $\ell_{\wa,j}$ being a linear convex cone follows immediately by the characterization given in 
        \ref{lem: elluExistenceProperties:ExistenceInducedFlowOnJthEdge}. 
        Furthermore, it is clear that $\ell_{\wa,j}$ is linear on its domain.

        \proofitem{\ref{lem: elluExistenceProperties:AggLinearity}} The uniqueness follows immediately by the fact that for any $\g^1_\arc,\g^2_\arc \in L_+(\hori)$ the implication 
        \begin{align*}
            \int_{\startint t]} \g_\arc^1\di\sigma =  \int_{\startint t]} \g_\arc^2\di\sigma \text{ for all } t\in \hori \implies \g^1_\arc = \g^2_\arc
        \end{align*}
        holds. For the linearity, consider two $h^1,h^2 \in \edom{\Routes'}$ and a scalar $\lambda \in \R$. Then, the claimed linearity follows from the following chain of equalities  for all $\arc  \in \GA$ and $t \in \hori$: 
        \begin{align*}
            \int_{\startint t]}\lambda\cdot \ell_{\Routes',\arc}(h^1) +  \ell_{\Routes',\arc}(h^2) \di\sigma &= \lambda\cdot\int_{\startint t]} \ell_{\Routes',\arc}(h^1) \di\sigma + \int_{\startint t]}\ell_{\Routes',\arc}(h^2) \di\sigma \\
            &= \lambda\cdot\sum_{\wa\in \Routes'} \sum_{j: \wa[j] = \arc}\int_{\arr_{\wa,j}^{-1}({\startint}t])}h^1_\wa\di\sigma \\
            &\quad+ \sum_{\wa\in \Routes'} \sum_{j: \wa[j] = \arc}\int_{\arr_{\wa,j}^{-1}({\startint}t])}h^2_\wa\di\sigma \\
            &\symoverset{1}{=}\sum_{\wa\in \Routes'} \sum_{j: \wa[j] = \arc}\int_{\arr_{\wa,j}^{-1}({\startint}t])}\lambda\cdot h^1_\wa + h^2_\wa\di\sigma .
        \end{align*}
        Here, \refsym{1} holds by the absolute convergence of both series which is implied by the existence of $\ell_{\Routes'}(h^j),j=1,2$, cf.~the discussion after \Cref{def: FlowDecomp}. 
 
        \proofitem{\ref{lem: elluExistenceProperties:ExistenceInducedFlow}, ``$\Leftarrow$''} 
        We first observe that  if $\ell_{\wa}(h_\wa), \wa \in \Routes'$ exist and $(\ell_{\wa}(h_\wa))_{\wa\in \Routes'} \in \seql[1][\Routes'][L_+(\hori)^\GA]$, then the sum $\sum_{\wa \in \Routes'}\ell_\wa(h_\wa)$ is well defined, i.e.~for any order $(\wa_l)_{l\in \N}$ of the walks in $\Routes'$, the induced 
        sequence of partial sums $(\sum_{l=1}^s\ell_{\wa_l}(h_{\wa_l}))_{s \in \N}$ converges (\wrt the strong topology of $L(\hori)$) to the same element $\sum_{\wa \in \Routes'}\ell_\wa(h_\wa)$. 
        This is a direct consequence of the definition of $\seql[1][\Routes'][L_+(\hori)^\GA]$ which requires the absolute convergence of 
        $\sum_{\wa \in \Routes'}\ell_\wa(h_\wa)$ in order for $(\ell_{\wa}(h_\wa))_{\wa\in \Routes'} \in \seql[1][\Routes'][L_+(\hori)^\GA]$. With this, we calculate for arbitrary $t \in \hori$ and $\arc \in \GA$: 
                \begin{align*}
            \int_{\startint t]}\sum_{\wa \in \Routes'}\ell_{\wa,\arc}(h_\wa)\di\sigma 
            & = \int_{\startint t]}\sum_{\wa \in \Routes'}\sum_{j:\wa[j]=\arc}\ell_{\wa,j}(h_\wa)\di\sigma \\
            & \symoverset{1}{=} \sum_{\wa \in \Routes'}\sum_{j:\wa[j]=\arc}\int_{\startint t]}\ell_{\wa,j}(h_\wa)\di\sigma \\
            &= \sum_{\wa \in \Routes'}\sum_{j:\wa[j]=\arc}\int_{\arr_{\wa,j}^{-1}(\startint t])}h_\wa\di\sigma.
        \end{align*}
        where we used in \refsym{1} that $\sum_{\wa \in \Routes'}\ell_\wa(h_\wa)$ exists. Hence, we have shown that the latter series fulfills \eqref{eq: DefUEdgeFlow}  which shows the claim.

        \proofitem{\ref{lem: elluExistenceProperties:ExistenceInducedFlow}, ``$\Rightarrow$''} 
If $\ell_{\Routes'}(h) \in L_+(\hori)^\GA$ exists and $h \in L_+(\hori)^{\Routes'}$, then it follows for an arbitrary $\wa^* \in \Routes'$, $\arc \in \GA$, $j^*\leq \abs{\wa^*}$ and  set $\mathfrak T \in \mathcal{B}(\hori)$ that 
\begin{align*}
     \int_{\mathfrak T} \ell_{\Routes',\arc}(h) \di\sigma = \sum_{\wa\in \Routes'} \sum_{j: \wa[j] = \arc}\int_{\arr_{\wa,j}^{-1}(\mathfrak T)}h_\wa\di\sigma   \geq \int_{\arr_{\wa^*,j}^{-1}(\mathfrak T)}h_{\wa^*}\di\sigma
\end{align*}
where the inequality holds by $h$ being nonnegative. 
The above shows in particular  that the series $\sum_{\wa\in \Routes'} \sum_{j: \wa[j] = \arc}\int_{\arr_{\wa,j}^{-1}(\mathfrak T)}h_\wa\di\sigma $ is well-defined, i.e.~absolute convergent. 
Now, for any null set $\mathfrak T$, the left hand side is $0$ by which the existence of 
 $\ell_{\wa^*,j}(h_{\wa^*})$ follows by~\ref{lem: elluExistenceProperties:ExistenceInducedFlowOnJthEdge} and the nonnegativity of $h_{\wa^*}$. Since $\wa^* \in \Routes'$ and $j^*\leq \abs{\wa^*}$ was arbitrary, we can conclude that $\ell_{\wa}(h_\wa),\wa \in \Routes'$ exist.  
 Thus, the above yields 
 \begin{align*}
     \int_{\mathfrak T} \ell_{\Routes',\arc}(h) \di\sigma =  \sum_{\wa\in \Routes'} \sum_{j: \wa[j] = \arc}\int_{ \mathfrak T}\ell_{\wa,j}(h_\wa)\di\sigma 
 \end{align*} 
 which implies $(\ell_{\wa}(h_\wa))_{\wa\in \Routes'} \in \seql[1][\Routes'][L_+(\hori)^\GA]$ and $\ell_{\Routes'}(h) = \sum_{\wa \in \Routes'} \ell_\wa(h_\wa)$. 
 \qedhere
    \end{structuredproof} 
\end{proof}

 Note that, in contrast to $\edom{\wa,j}$ and $\edom{\wa,\arc}$, the domain of the full network loading $\edom{\Routes'}$ is not necessary sequentially weakly closed in $\seql[1][\Routes'][L_+(\hori)]$ in case of $\Routes'$ not being finite:

\begin{figure}[h]
	\centering
	\BigPicture[1]{
		
		\begin{adjustbox}{max width=\textwidth}
			\begin{tikzpicture}
				\coordinate(d)at(4,0);
				\coordinate(s)at(-0,0);

				\node[namedVertexF,inner sep=4](temp-d)at(d){$\dest$};
				\node[namedVertexF,inner sep=4, minimum size = 0.8cm](temp-s)at(s){$\source$};

				\draw[edge] (temp-s)--node[below,sloped]{$\arc_1$}node[above,sloped]{$\trav \equiv 0$} (temp-d);
				\draw[edge] (temp-d) to [in=30,out=150,looseness=15] node[above,sloped]{$\trav \equiv 0$} node[below,sloped]{$\arc_2$}  (temp-d);
				
			\end{tikzpicture}
	\end{adjustbox}			 	}
	\caption{The network considered in \Cref{exa: EdomNotClosed,exa: DiscontinuityNL} with constant travel times of $0$ on both edges.}
	\label{fig: DiscontinuityNL}
\end{figure}

\begin{example}\label{exa: EdomNotClosed}
		Consider the network depicted in \Cref{fig: DiscontinuityNL} with constant travel times of $0$ on all edges. 
	Let $\Routes'=(\wa^k)_{k\in \N\cup\{0\}}$ be the collection of all $\source$,$\dest$-walks where $\wa^k=(\arc_1,\arc_2,\ldots,\arc_2)$ is the walk containing the edge $\arc_2$ exactly $k$ times.
	Define a sequence $(h^n)_{n \in \N}\subseteq \edom{\Routes'}$ via $h^n_{\wa^k} = \nicefrac{1}{k^2}\cdot \Indi_{[0,1]}$ for $1\leq k\leq n$ and $h^n_\wa = 0$ else. Then $\Nl[]_{\Routes',\arc_1}(h^n) = \sum_{k=1}^n \nicefrac{1}{k^2}\cdot \Indi_{[0,1]}$ and  $\Nl[]_{\Routes',\arc_2}(h^n) = \sum_{k=1}^n \nicefrac{1}{k} \cdot \Indi_{[0,1]}$. Moreover, we have $h^n\to h$ in the norm topology of $\seql[1][\Routes'][L_+(\hori)]$ where $h$ is given via $h_{\wa^k} =  \nicefrac{1}{k^2}\cdot \Indi_{[0,1]}$ for all $k\in \N$ and $h_{\wa^0} = 0$. This holds by the fact that $\sum_{k\in \N}\nicefrac{1}{k^2} < \infty$. Yet, $\Nl[]_{\wa^k,\arc_2}(h_{\wa^k})= \nicefrac{1}{n}\cdot \Indi_{[0,1]}$ and since $\sum_{k\in \N}\nicefrac{1}{k} = \infty$, we can conclude that $(\Nl[]_{\wa}(h_\wa))_{\wa\in \Routes'} \notin\seql[1][\Routes'][L_+(\hori)^\GA]$ 
	which shows by \Cref{lem: elluExistenceProperties:ExistenceInducedFlow} that $h\notin\edom{\Routes'}$. 
\end{example}

While  the previous \Cref{lem: elluExistenceProperties:ExistenceInducedFlowOnJthEdge} shows which walk flows $h_\wa$ induce an edge flow $\g^\wa_j \in L_+(\hori)$ on a given edge $j$, the next \namecref{lem: 1to1:h-f} shows the opposite direction, namely which edge flows on a given edge can be induced by a walk-inflow.

\begin{lemma}\label{lem: 1to1:h-f}
	Consider an arbitrary walk  $\wa$, $j\in [|\wa|+1]$ and $\g^{\wa,j} \in L_+(\hori)$. 
	Then, there exists a  ${h}_{\wa,j}\in L_+(\hori)$ with $\ell_{\wa,j}({h}_{\wa,j}) = \g^{\wa,j}$ if and only if 
	$\g^{\wa,j} = 0$ on ${\startint}\arr_{\wa,j}(-\infty)]$ where $\arr_{\wa,j}(-\infty)\coloneq \lim_{t\to -\infty}\arr_{\wa,j}(t)$. In this case, ${h}_{\wa,j}\in L_+(\hori)$ is uniquely determined. 
\end{lemma}
\begin{proof}
	We prove both directions of the claimed equivalence separately and begin with the ``only if''-direction: 
	In case of $ \arr_{\wa,j}(-\infty) = -\infty$, the interval  ${\startint}\arr_{\wa,j}(-\infty)] = \emptyset$ is empty and 
	the claim is trivial. Hence, assume  $\arr_{\wa,j}(-\infty) \in \R$ and consider 
	  the given equality \eqref{eq: Deflwj} for an arbitrary $t< \arr_{\wa,j}(-\infty)$, that is, 
	\begin{align*}
		\int_{\startint t]}\g^{\wa,j} \di \sigma = \int_{\arr_{\wa,j}^{-1}(\startint t])}h_{\wa,j} \di \sigma = \int_{\emptyset}h_{\wa,j}\di\sigma = 0.
	\end{align*}
	Since $\g^{\wa,j} \in L_+(\hori)$ is nonnegative, it follows that $\g^{\wa,j} = 0$ on $(-\infty,t]$ and since 
	$t< \arr_{\wa,j}(-\infty)$ was arbitrary, we can conclude that 	$\g^{\wa,j} = 0$ on ${\startint}\arr_{\wa,j}(-\infty)]$. 
	
	For the ``if''-direction, 
	consider the function 
	$\hori \to \R,    t \mapsto \int_{{\startint}\arr_{\wa,j}(t)]}\g^{\wa,j}\di\sigma$. 
	This function is  absolutely continuous  as the concatenation of an absolutely continuous function $t \mapsto \int_{{\startint}t]}\g^{\wa,j} \di\sigma$ and an absolutely continuous monotone increasing function $\arr_{\wa,j}$  (cf.~\Cref{lem: PropAbsCon:Conca}). 
	Hence, by \Cref{lem: PropAbsCon:Der}, the derivative ${h}_{\wa,j}$ exists almost everywhere and fulfills for all $a<b \in \R$ the equality 
	\begin{align*}
		\int_{[\arr_{\wa,j}(a),\arr_{\wa,j}(b)]}\g^{\wa,j}\di\sigma   = \int_{{\startint}\arr_{\wa,j}(b)]}\g^{\wa,j}\di\sigma - \int_{{\startint}\arr_{\wa,j}(a)]}\g^{\wa,j}\di\sigma =  \int_{[a,b]} {h}_{\wa,j}\di\sigma  . 
	\end{align*} 
	Since this holds for all $a<b \in \R$, we can deduce from the integrability of $\g^{\wa,j} \in L_+(\hori)$ that also ${h}_{\wa,j}\in L_+(\hori)$ is integrable. Moreover, by considering $a \to -\infty$ and using $\g^{\wa,j} = 0$ on ${\startint}\arr_{\wa,j}(-\infty)]$, 
	we get the following equality for all $b \in \hori$: 
	\begin{align}\label{eq: lem: 1to1}
		\int_{{\startint}\arr_{\wa,j}(b)]}\g^{\wa,j}\di\sigma = \int_{(\arr_{\wa,j}(-\infty),\arr_{\wa,j}(b)]}\g^{\wa,j}\di\sigma  = \int_{\startint b]} {h}_{\wa,j}\di\sigma  . 
	\end{align}

	It remains to show that this $h_{\wa,j}$ has $\g_j^\wa$ as its autonomous network loading (i.e., satisfies~\eqref{eq: Deflwj}): For an arbitrary $t > \arr_{\wa,j}(-\infty)$ ($\geq$ in case that $\arr_{\wa,j}^{-1}(\arr_{\wa,j}(-\infty)) \neq \emptyset$), the property of $\arr_{\wa,j}(t) \geq t$ ensures that the set $\arr_{\wa,j}^{-1}(t)$ is non-empty and contained in ${\startint}t]$. This, together with the continuity of $\arr_{\wa,j}$, shows that  the maximum  $b:= \max\{t'\in \hori\mid t' \in \arr_{\wa,j}^{-1}(t)\}$  exists. For this $b$ in the above equality~\eqref{eq: lem: 1to1}, we get
	\begin{align*}
		\int_{{\startint}t]} \g^{\wa,j}\di\sigma =  \int_{{\startint}\arr_{\wa,j}(b)]} \g^{\wa,j}\di\sigma = \int_{\startint b]}    {h}_{\wa,j}  \di\sigma =
		\int_{\arr_{\wa,j}^{-1}({\startint}t])}    {h}_{\wa,j}  \di\sigma.    
	\end{align*}
	Hence, the desired equality holds for all  $t > \arr_{\wa,j}(-\infty)$ ($\geq$ in case that $\arr_{\wa,j}^{-1}(\arr_{\wa,j}(-\infty)) \neq \emptyset$). 
	For any $t <\arr_{\wa,j}(-\infty)$ ($\leq$ in case that $\arr_{\wa,j}^{-1}(\arr_{\wa,j}(-\infty)) = \emptyset$) we have, by  $\g^{\wa,j} = 0$ on ${\startint}\arr_{\wa,j}(-\infty)]$ and $\arr_{\wa,j}^{-1}({\startint}t]) =\emptyset$, that also in this case the equality holds: 
	\begin{align*}
		\int_{{\startint}t]} \g^{\wa,j}\di\sigma =  0=   \int_{\arr_{\wa,j}^{-1}({\startint}t])}    {h}_{\wa,j}  \di\sigma.  
	\end{align*}
	Hence, $\ell_{\wa,j}(h_{\wa,j}) = \g^{\wa,j}$ and the proof of this direction is finished. 
	
	In order to finish the entire proof, it remains to show the uniqueness of $h_{\wa,j}$ for which we 
	  argue   that
	\begin{align*}
		\int_{\startint \arr_{\wa,j}(t)]}\g^{\wa,j}\di\sigma = \int_{\startint t]}h_{\wa,j}\di\sigma \text{ for all } t \in \R  
	\end{align*}
	holds.  
	By the fulfillment of \eqref{eq: Deflwj}, we know that 
	\begin{align*}
		\int_{\startint \arr_{\wa,j}(t)]}\g^{\wa,j}\di\sigma = \int_{\arr_{{\wa},j}^{-1}(\startint \arr_{\wa,j}(t)])}h_{\wa,j}\di\sigma \text{ for all } t \in \R
	\end{align*}
	holds. By $\arr_{\wa,j}$ being continuous and nondecreasing, we have $(-\infty,t] \subseteq \arr_{{\wa},j}^{-1}(\startint \arr_{\wa,j}(t)])$ and the difference $\arr_{{\wa},j}^{-1}(\startint \arr_{\wa,j}(t)]) \setminus (-\infty,t]$ is contained in $\arr_{{\wa},j}^{-1}(\arr_{\wa,j}(t))$. In particular, \Cref{lem: elluExistenceProperties:ExistenceInducedFlowOnJthEdge} 
	implies $h_{\wa,j} = 0$ on this difference, implying
	\begin{align*}
	\int_{\startint \arr_{\wa,j}(t)]}\g^{\wa,j}\di\sigma = \int_{\arr_{{\wa},j}^{-1}(\startint \arr_{\wa,j}(t)])}h_{\wa,j}\di\sigma = \int_{\startint t]}h_{\wa,j}\di\sigma  \text{ for all } t \in \R,
\end{align*}	
which finishes the poof.
\end{proof}

\subsection{Optimization Problems Involving \Auto Network Loadings}\label{app: elluContinuity} 

 In this section, we consider a whole class of  optimization problems involving \auto network loadings which 
 contain the optimization problem needed for the natural flow reduction algorithm (\Cref{alg: FlowDecompositionPseudo}) as a special case. 
 We consider general optimization problems of the following form: 
 
 \begin{align} 
 	\max_{\wflow}\;    &\objfunc(\wflow)  \tag{P} \label{opt: General} \\
 	\text{s.t.: } &\ell_{\Routes'}(h) \leq \eflow  \label{eq: ExOptSolLeq}\\
 	&\wflow \in \ofeas \nonumber
 \end{align}
Here, 
$\Routes'$ denotes an arbitrary countable collection of walks which may contain individual walks multiple but at most finitely many times.  
The constraint vector  $\g$ is an arbitrary  element in $L_+(\hori)^\GA$. 
The objective~$\objfunc$ is some real-valued function on $\ofeas$,  which, in turn, is some subset of $\edom{\Routes'}$ containing at least one $h$ fulfilling~\eqref{eq: ExOptSolLeq}, i.e.~the set of feasible solutions 
$\FeasSol\coloneq \{h \in \ofeas\mid \Nl[]_{\Routes'}(h)\leq \g\}$
is non-empty. 
Remark that 
\eqref{opt: General} is well-defined as $\ofeas \subseteq \edom{\Routes'}$ ensures that $\Nl[]_{\Routes'}(h)$ is well-defined.

 A first challenge in showing that \eqref{opt: General} has an optimal solution is the lack of continuity of the constraint function in \eqref{eq: ExOptSolLeq} with respect to any reasonable topology. 
 We demonstrate in \Cref{exa: DiscontinuityNL} that $\ell_{\Routes'}:\edom{\Routes'}\to L_+(\hori)^\GA$ 
 is not even strong-weak continuous.

 \begin{example}\label{exa: DiscontinuityNL}
 	Consider the setting of \Cref{exa: EdomNotClosed}. 
 	Define a sequence of walk inflow vectors $(h^n)_{n\in \N}$ via 
 		$h^n_{\wa^n}= \nicefrac{1}{n}\cdot \Indi_{[0,1]}$ and $h_\wa = 0$ else. 
 		Then $\Nl[]_{\Routes',\arc_1}(h^n) = \nicefrac{1}{n} \cdot \Indi_{[0,1]}$ and 
 		$\Nl[]_{\Routes',\arc_2}(h^n) = \Indi_{[0,1]}$ for all $n\in \N$. In particular, 
 		$\Nl[]_{\Routes',\arc_2}(h^n)\to \Indi_{[0,1]}$ in the norm topology of $L_+(\hori)$, 
		yet, $h^n\to 0$ in $\seql[1][\Routes'][L_+(\hori)]$ \wrt the norm topology and
		 $\Nl[]_{\Routes',\arc_2}(0) = 0\neq  \Indi_{[0,1]}$.  
 \end{example} 
 
A key insight in order to show that \eqref{opt: General} admits optimal solutions is the semi-continuity of the constraint function  in \eqref{eq: ExOptSolLeq} which we show in the following.

\begin{lemma}
    \label{lem: elluContinuity}
 Consider an arbitrary walk  $\wa$, $j\in [|\wa|+1]$, $\arc \in \GA$ and a countable collection of walks~$\Routes'$. 
     Then, the following statements are true.
    \begin{thmparts}
        \item The mappings $\ell_{\wa,j}$ and $\ell_{\wa,\arc}$ are strong-strong and sequentially weak-weak continuous from their maximal domains to $L_+(\hori)$. \label[thmpart]{lem: elluContinuity:SingleWalk}

        \item For all $h \in \edom{\Routes'}$, we have $\norm{\ell_{\Routes'}(h)} = \sum_{\wa \in \Routes'}\abs{\wa}\cdot \norm{h_\wa}$. In particular, $\edom{\Routes'}\subseteq \seql[1][{\Routes'}][L_+(\hori)]$.
        \label[thmpart]{lem: elluContinuity:Subset}
        
        \item The mapping $\ell_{\Routes'}: \edom{\Routes'} \to L_+(\hori)^\GA$ is sequentially weakly lower semi-continuous in the following sense: 
        For any   sequence $(h^n)_{n\in \N}$ with $h^n\in \edom{\Routes'}, n\in \N$ and   $h^n_\wa \wto h_\wa,\wa \in \Routes'$ for some $h \in L_+(\hori)^{\Routes'}$ and all $\g \in L_+(\hori)^\GA$ for which there exists $N \in \N$ such that for all $n \geq N$ the inequality $\ell_{\Routes'}(h^n) \leq \g$ holds, we have $h \in \edom{\Routes'}$ and 
        $\ell_{\Routes'}(h) \leq \g$.  \label[thmpart]{lem: elluContinuity:Sum} 
    \end{thmparts}
\end{lemma}
 Let us remark that
\Cref{lem: elluContinuity:SingleWalk} is reminiscent of known continuity statements regarding the mapping from edge inflow to outflow and the network loading operator for specific flow propagation models as e.g.\ the Vickrey point queue or linear edge delays (see, e.g., \cite[Section~5.3]{CominettiCL15} and \cite[Section~3]{ZhuM00}, respectively). 
Yet, these are neither generalizations nor special cases of the above \namecref{lem: elluContinuity} since we have to consider the fixed traversal times $\trav$ 
and also allow for sets of countably infinitely many walks~$\Routes'$.
\begin{proof}[Proof of \Cref{lem: elluContinuity}]
\begin{structuredproof}
\proofitem{\ref{lem: elluContinuity:SingleWalk}} 
        We only prove the statements for  $\ell_{\wa,j}$ since the analogue ones for $\ell_{\wa,\arc}$ follow immediately from them.

        The strong-strong continuity is a direct consequence of $\ell_{\wa,j}$ being a linear (\Cref{lem: elluExistenceProperties:Linearity}) and bounded operator. 
        For the boundedness, remark that 
        \begin{align}\label{eq: hellNormsEqual}
            \norm{\ell_{\wa,j}(h_\wa)} = \int_\hori \abs{\ell_{\wa,j}(h_\wa)}\di\sigma =   \int_\hori  {\ell_{\wa,j}(h_\wa)}\di\sigma =  \int_{\arr_{\wa,j}^{-1}(\hori)} h_\wa \di\sigma =\int_\hori h_\wa \di\sigma =  \norm{h_\wa},
        \end{align}
        where we used the nonnegativity of $h_\wa$ and $\ell_{\wa,j}(h_\wa)$. 
Since any strong-strong continuous linear function is, in particular, sequentially weak-weak continuous, the claim of \ref{lem: elluContinuity:SingleWalk} is proven.

\proofitem{\ref{lem: elluContinuity:Subset}} 
By  \Cref{lem: elluExistenceProperties:ExistenceInducedFlow}, we know that $(\ell_\wa(h_\wa))_{\wa \in \Routes'} \in \seql[1][\Routes'][L_+(\hori)^\GA]$ and $\ell_{\Routes'}(h) = \sum_{\wa \in \Routes'}\ell_{\wa}(h_\wa)$ holds. 
Hence, we get the following chain of equalities:
       \begin{align*}
        \norm{\ell_{\Routes'}(h)} =  \sum_{\wa \in {\Routes'}} \sum_{\arc \in \GA} \norm{\ell_{\wa,\arc}(h_\wa)}  &\symoverset{1}{=} \sum_{\wa \in {\Routes'}} \sum_{\arc \in \GA} \sum_{j:\wa[j] = \arc}\norm{\ell_{\wa,j}(h_\wa)} \\ &\symoverset{2}{=} \sum_{\wa \in {\Routes'}}  \sum_{\arc \in \GA} \sum_{j:\wa[j] = \arc}   \norm{{h}_\wa} =\sum_{\wa \in {\Routes'}} \abs{\wa}\cdot\norm{{h}_\wa}.
    \end{align*}
    Here, the equality indicated with \refsym{1} holds by the nonnegativity of  $\ell_{\wa,j}(h_\wa)$. Equality \refsym{2} is valid by \eqref{eq: hellNormsEqual}. 

    \proofitem{\ref{lem: elluContinuity:Sum}}  
    By \Cref{lem: elluExistenceProperties:ExistenceInducedFlow} and $h^n \in \edom{\Routes'}$, we know that   $\ell_\wa(h^n_\wa),\wa \in \Routes',n \in \N$ exist. Hence, \Cref{lem: elluExistenceProperties:Linearity} implies that  the induced flow $\ell_{\wa}(h_\wa)$ exists for 
    all $\wa \in \Routes'$  as $h^n_\wa \wto h_\wa$. 
    Hence, by \Cref{lem: elluExistenceProperties:ExistenceInducedFlow}, it remains  to show that $(\ell_\wa(h_\wa))_{\wa \in \Routes'} \in  \seql[1][{\Routes'}][L_+(\hori)^\GA]$ in order to derive the existence of $\ell_{\Routes'}(h)$ which is then given by $\ell_{\Routes'}(h) = \sum_{\wa \in \Routes'}\ell_\wa(h_\wa)$. 
    
    Consider an arbitrary ordering of $\Routes'=\{\wa_l\}_{l \in \N}$ for the following. 
    By \ref{lem: elluContinuity:SingleWalk}, we know that $\ell_\wa(h_\wa^n)\wto \ell_\wa(h_\wa)$ for all $\wa \in \Routes'$. This allows us to verify for an arbitrary $\mathfrak T \in \mathcal{B}(\hori)$, $\arc \in \GA$ and $L\in \N$: 
    \begin{align*}
         \int_{\mathfrak T} \sum_{l \leq L}\ell_{\wa_l,\arc}(h_{\wa_l}^n) \di\sigma  \to \int_{\mathfrak T} \sum_{l \leq L}\ell_{\wa_l,\arc}(h_{\wa_l})\di\sigma. 
    \end{align*}
    Since the left side is bounded by $\int_{\mathfrak T} \g_\arc \di \sigma$ for $n \geq N$ due to $\ell_{\Routes',\arc}(h^n) \leq \g_\arc$, we can deduce that also the right side is bounded by this quantity. 
    Since $\mathfrak T\in \mathcal{B}(\hori)$ was arbitrary, it follows that $\sum_{l \leq L}\ell_{\wa_l,\arc}(h_{\wa_l}) \leq \g_\arc$. 
    This implies by the monotone convergence theorem by Lebesgue and Beppo Levi (\cite[Theorem 2.8.2]{Bogachev2007I}) that   $\sum_{l \in \N}\ell_{\wa_l,\arc}(h_\wa) \in L_+(\hori)$ exists and is smaller or equal to $\g_\arc$. 
    In particular, again by the nonnegativity $\ell_{\wa,\arc}(h_\wa) \in L_+(\hori),\wa \in \Routes'$, we get that  $\sum_{l \in \N}\norm{\ell_{\wa_l,\arc}(h_{\wa_l})} \leq \norm{\g_\arc}$, i.e.\  the sum  is absolutely convergent, showing that  $(\ell_\wa(h_\wa))_{\wa \in \Routes'} \in  \seql[1][{\Routes'}][L_+(\hori)^\GA]$. Subsequently $\ell_{\Routes'}(h)$ is well-defined and fulfills $\ell_{\Routes'}(h) = \sum_{l \in \N}\ell_{\wa_l}(h_{\wa_l}) \leq \g$.
    \qedhere
\end{structuredproof}    
\end{proof}

 We come now to the main result of this subsection,  showing  that problems of the general form~\eqref{opt: General} have an optimal solution under suitable assumptions. 
 One of those assumptions will require $\ofeas$  to be  sequentially weakly closed in $\edom{\Routes'}$. 
 In this regard, let us remark that this does not necessary imply the sequential weakly closedness of $\ofeas$ in $\seql[1][\Routes'][L_+(\hori)]$. The reason being that $\edom{\Routes'}$ is not necessary sequentially weakly closed in $\seql[1][\Routes'][L_+(\hori)]$ in case of $\Routes'$ not being finite, as already remarked in~\Cref{exa: EdomNotClosed}.
 In the following theorem, 
 $\FeasSol$ and  $\edom{\Routes'}[]$ are equipped with the subspace topology induced by $\seql[1][\Routes'][L_+(\hori)]$. Moreover, we require that 
$\objfunc:\FeasSol\to \R$  is sequentially weakly upper semi-continuous, that is, 
$\limsup_{n \to\infty} \objfunc(h^n) \leq \objfunc(h)$ for any weakly converging sequence $h^n\wto h$ in $\seql[1][\Routes'][L_+(\hori)]$ with all $h^n$ and $h$ contained in the feasible set $\FeasSol$.

\begin{rsttheorem}{\ref{thm: ExistenceOptSol}} 
	Assume that $\objfunc:\FeasSol\to \R$  is sequentially weakly upper semi-continuous and $\ofeas$ is sequentially weakly closed in~$\edom{\Routes'}$.   Then, 
	the optimization problem~\eqref{opt: GeneralOverview} has an optimal solution. 
\end{rsttheorem}
\begin{proof}[Proof of \Cref{thm: ExistenceOptSol}]
We aim to apply Weierstrass' theorem which requires compactness of the feasible set \wrt a suitable topology. 
For this, we first have to switch (via a natural homeomorphism) from the infinite product space $\seql[1][\Routes'][L_+(\hori)]$ to the space $L_+(\hori\times \Routes')$ of nonnegative, integrable functions on $\hori\times \Routes'$ (the latter being equipped with the product measure $\sigma\otimes \eta$ where $\eta$ is the counting measure on~$\Routes'$). This, then, allows us to use standard compactness results for such spaces.

\begin{claim}\label{claim: phiHomo}
Define $\phi:\seql[1][\Routes'][L_+(\hori)]\to L_+(\hori \times \Routes'), {\wflow} \mapsto \phi(\wflow)$  
with  $\phi(\wflow)(t,\wa) = {\wflow}_\wa(t)$ for all $\wa \in \Routes'$ and almost all $t\in \hori$. 
Then, $\phi$   
defines a homeomorphism \wrt both spaces being equipped with their norm induced topologies. 
Furthermore, $\phi$ is also a sequential homeomorphism \wrt both spaces being equipped
with the norm-induced weak topologies. 
\end{claim}
\begin{proofClaim}
    \begin{structuredproof}
        \proofitem{$\phi$ is well-defined} 
         In order to prove well-definedness, we have to show that  $\phi\big(\seql[1][\Routes'][L_+(\hori)]\big) \subseteq L_+(\hori\times \Routes')$. We calculate for an arbitrary $\wflow \in \seql[1][\Routes'][L_+(\hori)]$: 
 
\begin{align}\label{eq: NormsEqual}
    \norm{\wflow} \eqperdef \sum_{\wa \in \Routes'}\norm{\wflow_\wa} = \int_{\Routes'} \int_{\hori}\abs{\phi(\wflow)(t,\wa)}\di\sigma(t)\di\eta(\wa) \symoverset{1}{=}\int_{\hori\times \Routes'} \abs{\phi(\wflow)} \di\sigma\otimes\eta \defpereq \norm{\phi(\wflow)}
\end{align}
Remark that equality \refsym{1} holds by first applying Tonelli's theorem (\cite[Theorem 3.4.5]{Bogachev2007I}) and then Fubini's theorem (\cite[Theorem 3.4.4]{Bogachev2007I}). 

\proofitem{$\phi$ is injective}  Suppose we have $\wflow^1,\wflow^2 \in \seql[1][\Routes'][L_+(\hori)]$ with $\phi(\wflow^1) = \phi(\wflow^2)$, that is, for $\sigma\otimes \eta$-all $(t,\wa)$ we have $\phi(\wflow^1)(t,\wa) = \phi(\wflow^2)(t,\wa)$. Since $\eta$ is the counting measure, this implies that the latter equality is valid for all $\wa \in \Routes'$ and almost all $t \in \hori$. Since furthermore $ \phi(\wflow^j)(t,\wa)=\wflow^j_\wa(t), j = 1,2$ for almost all $t \in \hori$ by definition of $\phi$, it follows that $\wflow^1=\wflow^2$. 

\proofitem{$\phi$ is surjective} For an arbitrary $\hat{\wflow} \in L_+(\hori\times \Routes')$, the  equivalence class ${\wflow} \in L_+(\hori)^{\Routes'}$  
  given by ${\wflow}_\wa(t) := \hat{\wflow}(t,\wa)$ for a.e.~$t \in \hori$
  and all $\wa \in \Routes'$ fulfills $\phi(\wflow) = \hat{\wflow}$.  
Note that it is  a direct consequence of Fubini's theorem that ${\wflow} \in  \seql[1][\Routes'][L_+(\hori)]$. 

\proofitem{$\phi,\phi^{-1}$ are norm continuous} This is an immediate consequence of the equality derived in~\eqref{eq: NormsEqual} and the linearity of both functions.

\proofitem{$\phi,\phi^{-1}$ are sequentially weakly continuous} 
This is a direct consequence of $\phi,\phi^{-1}$ being linear and norm-continuous functionals.   
\qedhere
    \end{structuredproof}
\end{proofClaim}

The above claim allows us to reformulate optimization problem~\eqref{opt: General} via 
\begin{align}
    \sup_{\wflow}\;  &\objfunc(\phi^{-1}(\hat{\wflow})) \tag{$\tilde{\mathrm{P}}$}\label{opt: GeneralRef}
    \\ \text{s.t.: }   &\ell_{\Routes'}(\phi^{-1}(\hat{\wflow}))\leq \eflow \label{eq: opt:GeneralRef}\\
    & \hat{\wflow}\in \phi(\ofeas) \nonumber 
\end{align}

We then proceed by showing the following properties of the above reformulation:
\begin{claim}\label{claim:ExistenceOptSol:MinimizationProblem}
    The maximization problem~\ref{opt: GeneralRef} has  a sequentially  weakly closed feasibility set $\Feas$ which is contained in a sequentially weakly compact set. Moreover, the objective function is  sequentially weakly  upper semi-continuous on $\Feas$, that is, $\limsup_{n \to \infty} \objfunc(\phi^{-1}(\hat{\wflow}^n)) \leq \objfunc(\phi^{-1}(\hat{\wflow}))$ for any weakly converging  sequence $\hat{h}^n\wto \hat{h}$ in $L_+(\hori\times\Routes')$  with $\hat{h}^n,\hat{h} \in \Feas$. 
\end{claim}

From this \namecref{claim:ExistenceOptSol:MinimizationProblem} the \namecref{thm: ExistenceOptSol}'s statement now follows with an argument analogously to the proof of Weierstrass' extreme value theorem: Let $(\wflow^n) \subseteq \Feas$ be a sequence of feasible solution with objective values converging to the supremum of the above problem. Since this sequence is contained in a sequentially weakly compact set, it has a converging subsequence with limit point~$\wflow^*$. As $\Feas$ is sequentially weakly closed, this limit point must also be contained in $\Feas$. Finally, using the upper semi-continuity of the objective function gives us that $\objfunc(\wflow^*)$ is larger or equal to the supremum of the given maximization problem. Hence, $\wflow^*$ is an optimal solution.

\begin{proofClaim}[Proof of \Cref*{claim:ExistenceOptSol:MinimizationProblem}] 
The sequential  weak  upper semi-continuity of the objective function is 
an immediate consequence of the semi-continuity required for $\objfunc$ and the property of $\phi$ being a sequential weak homeomorphism.   Next, we argue for the sequential weak closedness and compactness: 
\begin{structuredproof}
     \proofitem{$\Feas$ is sequentially weakly closed}  
Let $\phi(h^n)\wto \phi(h)$ in $L_+(\hori\times \Routes')$ with $\phi(h^n)\in \Feas,n \in \N$.
Since $\phi$ is a sequential homeomorphism for the weak topologies, we get that $h^n\wto h$ in $\seql[1][\Routes'][L_+(\hori)]$ and in particular that $h^n_\wa \wto h_\wa$ in $L_+(\hori)$ for each $\wa \in \Routes'$. Moreover, by $\phi(h^n) \in \Feas$, we get that $\ell_{\Routes'}(h^n)\leq \g$ for all $n \in \N$. Thus, \Cref{lem: elluContinuity:Sum} is applicable, showing that $h \in \edom{\Routes'}$ and $\ell_{\Routes'}(h) \leq \g$. This, in turn, implies that $\phi(h)$ fulfills the constraint \eqref{eq: opt:GeneralRef}. Moreover, by $h \in \edom{\Routes'}$, we know that $h \in \ofeas$ as the latter set is weakly  sequentially closed in $\edom{\Routes'}$ and $\phi(h^n) \in \Feas \subseteq \phi(\ofeas)$. 
Hence, $\phi(h)$ is also contained in $\phi(\ofeas)$ and, thus, we can conclude that $\phi(h) \in \Feas$. 

\proofitem{$\Feas$ is contained in a sequentially weakly compact set} 
Clearly, $\Feas$ is contained in its weak closure. We will argue in the following that this weak closure is weakly compact. Note that by   the Eberlein–\v{S}mulian Theorem (cf.~\cite[Theorem 6.34]{guide2006infinite}), weak compactness and  sequentially weak compactness are equivalent on $L(\hori \times \Routes')$.

We will verify the equivalent conditions stated in~\cite[Theorem 4.7.20 (iv)]{Bogachev2007I} for $\Feas$ having weakly compact closure. We do so in the following and start by noting that 
$\Feas$ is norm bounded. To see this, observe that 
 for an arbitrary feasible $\phi({\wflow}) \in \Feas$: 
      \begin{align}\label{eq: BoundOfFeasibleSols}
         \sum_{\arc \in \GA} \eflow_\arc \geq 
         \sum_{\arc \in \GA} \ell_{\Routes',\arc}(h) \symoverset{1}{=} 
         \sum_{\arc \in \GA} \sum_{\wa \in {\Routes'}} \sum_{j:\wa[j] = \arc}{\ell_{\wa,j}(\wflow_\wa)}  \symoverset{2}{\geq}\sum_{\wa \in {\Routes'}} {\ell_{\wa,1}(\wflow_\wa)}  = \sum_{\wa \in {\Routes'}}{{\wflow}_\wa},
    \end{align}
    where we used in \refsym{1} the description in \Cref{lem: elluExistenceProperties:ExistenceInducedFlow} and for \refsym{2} the nonnegativity of $\ell_{\wa,j}(h_\wa)$.
Since all appearing functions in the above inequality are nonnegative, the 
inequality remains true when considering the respective norms. Thus, combining~\eqref{eq: BoundOfFeasibleSols} with the  equality in~\eqref{eq: NormsEqual} shows that $\Feas$ is uniformly bounded by $\sum_{\arc \in \GA}\norm{\eflow_\arc}$.

Next, we argue that the elements in 
$\Feas$ have uniformly absolutely continuous integrals. 

    Let $\varepsilon>0$ be arbitrary. 
    Then there exists $\delta> 0$ such that $\int_{\mathfrak T}\sum_{\arc \in \GA}{\eflow_\arc} \di \sigma< \varepsilon$ for all $\mathfrak T \in \mathcal{B}(\hori)$ with $\sigma(\mathfrak T)< \delta$ by the absolute continuity of the Lebesgue integral, cf.~\cite[Theorem 2.5.7]{Bogachev2007I}.   
    
			Now consider an arbitrary $\mathfrak A  \in \mathcal{B}(\hori\times \Routes')$ with $(\leb\otimes\eta)(\mathfrak A)< \delta $ and let  $\mathfrak A = \bigcup_{\wa \in \Routes'} \mathfrak T_\wa \times \{\wa\}$.  
 			For any $\phi(\wflow) \in \Feas$, we have 
 \begin{align*}
 \varepsilon > \int_{ \bigcup_{\wa \in \Routes'}\mathfrak T_\wa} \sum_{\arc \in \GA}{\eflow_\arc} \di\sigma \overset{\eqref{eq: BoundOfFeasibleSols}}{\geq}
\int_{ \bigcup_{\wa \in \Routes'}\mathfrak T_\wa} \sum_{\wa\in \Routes'}  \wflow_\wa  \di \sigma 
 &\symoverset{1}{=} \int_{ \bigcup_{\wa \in \Routes'}\mathfrak T_\wa \times \Routes'} \phi(\wflow) \di\sigma \otimes \eta
 \\&\geq \int_{\mathfrak A}  \phi(\wflow)  \di \sigma\otimes \eta, 
 \end{align*}
Hence, the first inequality follows by $  \delta > (\sigma\otimes\eta)(\mathfrak A) = \sum_{\wa \in \Routes'} \sigma(\mathfrak T_\wa) \geq \sigma(\bigcup_{\wa \in \Routes'}\mathfrak T_\wa)$.  
The equality indicated by \refsym{1} is valid by Fubini's theorem (\cite[Theorem 3.4.4]{Bogachev2007I}). 
The last inequality is valid as $\phi(\wflow)\geq 0$ (by $\wflow \geq 0$) and $\mathfrak A \subseteq  \bigcup_{\wa \in \Routes'} \mathfrak T_\wa \times \Routes'$.

Finally, we show that for every $\varepsilon>0$ there exists some $\mathfrak A$ with $(\sigma\otimes\eta)(\mathfrak A)< \infty$ such that we have
$\int_{\hori\times \Routes\setminus \mathfrak A} \phi(\wflow)\di\sigma\otimes\eta < \varepsilon$ for all $\phi(\wflow) \in \Feas$. 
Let $\varepsilon>0$ be arbitrary. 
We require two observations: 
First, remark that 
we can apply \cite[Theorem 4.7.20 (i) and (iv)]{Bogachev2007I}   to  the singleton 
    $\{\sum_{\arc \in \GA}\g_\arc\}$. This is due to the fact that a singleton is weakly compact in 
    $L_+(\hori \times \Routes')$ as weak topologies are Hausdorff (cf.~\cite[Section 5.14]{guide2006infinite}) and singletons are closed in Hausdorff spaces. 
    Hence, there exists $\mathfrak T_{\g,\varepsilon/2} \subseteq \hori$ with $\sigma(\mathfrak T_{\g,\varepsilon/2})< \infty$ and $\int_{\hori \setminus \mathfrak T_{\g,\varepsilon/2}} \sum_{\arc\in \GA} \g_\arc \di\sigma < \varepsilon/2$. 

Secondly, define 
for every $c \in \mathcal{C}$ and $k \in \N$  the set of all walks in~$\Routes'$ containing the cycle~$c$ at least $k$~times by  $\Routes_{c,k}\subseteq \Routes'$. Consider an arbitrary $c \in \mathcal{C}$, $\arc \in c$ and $\phi(\wflow) \in \Feas$. By feasibility, we get that $\sum_{\wa \in \Routes'}\ell_\wa(\wflow_\wa) \leq \eflow$ and hence 
\begin{align*}
    \norm{\eflow_\arc} =  \int_{\hori} {\eflow}_\arc \di \sigma  \geq \sum_{\wa \in \Routes_{c,k}}\sum_{j: \wa[j] = \arc}\int_{\arr_{\wa,j}^{-1}(\hori)} {\wflow}_\wa\di\sigma \geq   \sum_{\wa \in \Routes_{c,k}} k \int_{\hori} {\wflow}_\wa\di\sigma. 
\end{align*}
Now let $k_{c,\varepsilon} \in \N$ be such that $\min_{\arc \in c}\norm{\eflow_\arc}/k_{c,\varepsilon} < \varepsilon/(2|\mathcal{C}|)$. 
Then, the above shows the following estimate: $\sum_{\wa \in \Routes_{c,k_{c,\varepsilon}}} \int_{\hori} {\wflow}_\wa\di\sigma < \varepsilon/(2|\mathcal{C}|)$. 

With these observations, consider the set
 $\mathfrak{A} :=\mathfrak T_{\g,\varepsilon/2} \times (\Routes' \setminus \bigcup_{c \in \mathcal{C}} \Routes_{c,k_{c,\varepsilon}})$ and calculate: 
\begin{align*}
   \int_{\hori\times \Routes'\setminus \mathfrak A} \phi(\wflow)\di\sigma\otimes\eta 
   &=  \int_{(\hori\setminus\mathfrak T_{\g,\varepsilon/2}) \times( \Routes' \setminus \bigcup_{c \in \mathcal{C}} \Routes_{c,k_{c,\varepsilon}})} \phi(\wflow)\di\sigma\otimes\eta +  \int_{\hori\times \bigcup_{c \in \mathcal{C}} \Routes_{c,k_{c,\varepsilon}}} \phi(\wflow)\di\sigma\otimes\eta  
   \\&= \int_{\hori\setminus\mathfrak T_{\g,\varepsilon/2}} \sum_{\wa \in (\Routes' \setminus \bigcup_{c \in \mathcal{C}} \Routes_{c,k_{c,\varepsilon}}) }h_\wa \di\sigma 
   +\sum_{c \in \mathcal{C}}\sum_{\wa \in  \Routes_{c,k_{c,\varepsilon}}}  \int_{\hori} \wflow_\wa\di\sigma \\
   \overset{\eqref{eq: BoundOfFeasibleSols}}&{\leq}   \int_{\hori \setminus \mathfrak T_{\g,\varepsilon/2}} \sum_{\arc\in \GA} \g_\arc \di\sigma    +\sum_{c \in \mathcal{C}}\sum_{\wa \in  \Routes_{c,k_{c,\varepsilon}}}  \int_{\hori} \wflow_\wa\di\sigma 
   \\ &<  \varepsilon/2 +  |\mathcal{C}| \cdot \varepsilon/(2|\mathcal{C}|)  =\varepsilon
\end{align*}
which shows the claim. Note that $\sigma \otimes \eta(\mathfrak A) = \sigma(\mathfrak T_{\g,\varepsilon/2}) \cdot \abs{\Routes' \setminus \bigcup_{c \in \mathcal{C}} \Routes_{c,k_{c,\varepsilon}}}< \infty$ as 
$\sigma(\mathfrak T_{\g,\varepsilon/2}) < \infty$ 
and $\Routes' \setminus \bigcup_{c \in \mathcal{C}} \Routes_{c,k_{c,\varepsilon}}$ is a finite set by our assumption that $\Routes'$ only contains each walk finitely often. \qedhere
\end{structuredproof}
\end{proofClaim}

With this \namecref{claim:ExistenceOptSol:MinimizationProblem} the \namecref{thm: ExistenceOptSol} now follows as explained before.
\end{proof}

\subsection{\Auto Node Balances and  \boldmath{$\source$,$\dest$}-Flows}\label{app:flowcon}

In this section we consider the concept of \auto node balance, i.e.\ the node balance of an arbitrary vector $\g \in L_+(\hori)^\GA$ as defined in \Cref{def: FlowBalaSDFlow} \wrt the fixed traversal time function~$\trav$.
As previously, we omit in the following the term ``\auto[''] whenever it is clear from context. 

For the subsequent proofs we require the following notation for the two types of standard (Borel) measures  which are involved in the definition of node balances. 
Firstly, for any measurable function $\g:\hori\to \R$, we denote by $\g \cdot \sigma$ the measure on $\mathcal{B}(\hori)$
given by $\g\cdot\sigma(\mathfrak T) := \int_{\mathfrak T} \g\di\sigma,\mathfrak T \in \mathcal{B}(\hori)$. Secondly, for any measurable function $A: \hori\to\hori$ and any measure $\mu$ on $\mathcal{B}(\hori)$ we denote by $\mu\circ A^{-1}$ the image measure of $\mu$ under $A$ which is defined by $\mu\circ A^{-1}(\mathfrak T) := \mu(A^{-1}(\mathfrak T))$. The latter is again a measure on $\mathcal{B}(\hori)$. 
Finally, we also introduce the notation of $\mu\leq \mu'$ for two measures, meaning that $\mu(\mathfrak T) \leq \mu'(\mathfrak T)$ for all $\mathfrak T \in \mathcal{B}(\hori)$.
We refer to~\cite{Bogachev2007I} for a comprehensive overview of measure theory. 
With this, we can now also write the \auto node balance as $\op_v\eflow = \sum_{\arc \in \delta^+(v)}   \eflow_\arc \cdot \sigma -  \sum_{\arc \in \delta^-(v)}  ( \eflow_\arc \cdot\sigma) \circ \exit_\arc^{-1}$. 
 
The following \namecref{lem: flowcon} states that any induced edge flow  $\ell_\wa(h_\wa)$ fulfills \auto flow conservation at all nodes except the start and end node of $\wa$  as well as that the \auto net outflow rate at the start node equals $h_\wa$. 
\Cref{lem: flowconW'}, then, extends this to entire walk inflow rate vectors.

\begin{lemma}
    \label{lem: flowcon}
Consider an arbitrary $v_1,v_2$-walk $\wa$ for some $v_1,v_2 \in \GV$, a corresponding walk inflow rate $h_\wa \in \edom{\wa}$ with $\g^\wa:=\ell_\wa(h_\wa)$   and a node $v \in \GV$.
Then we have 
\begin{align*} 
    \op_v \g^\wa = 1_{v_1}(v)\cdot  (h_\wa \cdot \sigma)   -1_{v_2}(v) \cdot (h_\wa \cdot\sigma)\circ\arr_{\wa,|\wa|+1}^{-1}.
\end{align*}
If $\ell_{\wa,\abs{\wa}+1}(h_\wa)$ exists as well, then $\op_{v_2} \g^\wa =  \big(1_{v_1}(v_2)\cdot  h_\wa  -\ell_{\wa,\abs{\wa}+1}(h_\wa)\big) \cdot \sigma$, i.e.~$\g^\wa$ has the outflow rate $ 1_{v_1}(v_2)\cdot  h_\wa  -\ell_{\wa,\abs{\wa}+1}(h_\wa)$ at the end node $v_2$.  
\end{lemma}
\begin{proof}
Consider a node $v \in \GV$. 
   We calculate: 
   \begin{align*}
        \sum_{\arc \in \delta^+(v)}\g^\wa_\arc \cdot \sigma  &=  \sum_{\arc \in \delta^+(v)} \sum_{j:\wa[j] = \arc}(h_\wa \cdot\sigma)\circ\arr_{\wa,j}^{-1} \\
        &= 
        1_{v_1}(v) \cdot (h_\wa \cdot\sigma)\circ\arr_{\wa,1}^{-1}  +
       \sum_{\arc \in \delta^+(v)} \sum_{j\in\{2,\ldots,\abs{\wa}\}:\wa[j] = \arc}(h_\wa \cdot\sigma)\circ\arr_{\wa,j}^{-1}  
 \end{align*}      
       For the last part, we furthermore observe that (explanations follow)
 	\begin{align}
 		\sum_{\arc \in \delta^+(v)} \sum_{j\in\{2,\ldots,\abs{\wa}\}:\wa[j] = \arc}(h_\wa \cdot\leb)\circ\arr_{\wa,j}^{-1}  
 		&= \sum_{\arc \in \delta^+(v)} \sum_{j\in\{2,\ldots,\abs{\wa}\}:\wa[j] = \arc}(h_\wa \cdot\leb)\circ\big(\arr_{\wa,j-1}^{-1} \circ \exit_{\wa[j-1]}^{-1}\big) \nonumber \\
 		&= \sum_{\arc \in \delta^+(v)} \sum_{j\in\{1,\ldots,\abs{\wa}-1\}:\wa[j+1] = \arc}(h_\wa \cdot\leb)\circ\big(\arr_{\wa,j}^{-1} \circ \exit_{\wa[j]}^{-1}\big) \nonumber\\
 		&=\sum_{j\in\{1,\ldots,\abs{\wa}-1\}:\wa[j+1] \in \edgesFrom{v}}(h_\wa \cdot\leb)\circ\big(\arr_{\wa,j}^{-1} \circ \exit_{\wa[j]}^{-1}\big) \nonumber\\
 		&=\sum_{j\in\{1,\ldots,\abs{\wa}-1\}:\wa[j] \in \edgesTo{v}}(h_\wa \cdot\leb)\circ\big(\arr_{\wa,j}^{-1} \circ \exit_{\wa[j]}^{-1}\big) \nonumber\\
 		&= \sum_{\arc \in \delta^-(v)} \sum_{j\in\{1,\ldots,\abs{\wa}-1\}:\wa[j] = \arc}(h_\wa \cdot\leb)\circ\big(\arr_{\wa,j}^{-1} \circ \exit_{\wa[j]}^{-1}\big) \label{eq: NodeBalaHelp}
 	\end{align} 
The first equality is valid by the recursive definition of~$\arr_{\wa,j}$. The second equality  results due to an index shift. 
Since $\wa[\abs{\wa}] \in \edgesTo{v}$ only for $v = v_2$, we can further simplify \eqref{eq: NodeBalaHelp} via 
 \begin{align*}
        \eqref{eq: NodeBalaHelp} &= \sum_{\arc \in \delta^-(v)} \sum_{j\in\{1,\ldots,\abs{\wa}\}:\wa[j] = \arc}(h_\wa \cdot\sigma)\circ\arr_{\wa,j}^{-1} \circ \exit_{\wa[j]}^{-1} - 1_{v_2}(v) \cdot (h_\wa \cdot\sigma)\circ\arr_{\wa,\abs{\wa}}^{-1} \circ \exit_{\wa[\abs{\wa}]}^{-1} \\
        &= \sum_{\arc \in \delta^-(v)} (\g^\wa_\arc \cdot \sigma)\circ \exit_\arc^{-1} - 1_{v_2}(v)\cdot  (h_\wa \cdot\sigma)\circ\arr_{\wa,\abs{\wa}}^{-1} \circ \exit_{\wa[\abs{\wa}]}^{-1} .
\end{align*}
Thus, we ultimately arrive at 
\begin{align*}
      \sum_{\arc \in \delta^+(v)}\g^\wa_\arc \cdot \sigma &= 
      1_{v_1}(v) \cdot (h_\wa \cdot\sigma)\circ\arr_{\wa,1}^{-1}  + \sum_{\arc \in \delta^-(v)} (\g^\wa_\arc \cdot \sigma)\circ \exit_\arc^{-1} - 1_{v_2}(v)\cdot  (h_\wa \cdot\sigma)\circ\arr_{\wa,\abs{\wa}}^{-1} \circ \exit_{\wa[\abs{\wa}]}^{-1} \\
      &= 1_{v_1}(v) \cdot (h_\wa \cdot\sigma)  + \sum_{\arc \in \delta^-(v)} (\g^\wa_\arc \cdot \sigma)\circ \exit_\arc^{-1} - 1_{v_2}(v)\cdot  (h_\wa \cdot\sigma)\circ\arr_{\wa,\abs{\wa}+1}^{-1}
\end{align*}
from which the statement of the \namecref{lem: flowcon} follows immediately. 
\end{proof}

 Next, we aim to extend \Cref{lem: flowcon} to entire walk inflow rate vectors. 
 For this, we require the following \namecref{lem: NodeBalaCont} stating that the operator 
 \begin{align*}
     \op:L_+(\hori)^\GA \to \Meas(\R)^\GV, \g \mapsto (\op_v\g)_{v \in \GV}
 \end{align*}
 is linear and  admits a certain continuity. Here, $\Meas(\R)$ denotes the set of all bounded Borel measures on $\hori$.

 \begin{lemma}\label{lem: NodeBalaCont}
 	The node balance operator $\op:L_+(\hori)^\GA \to \Meas(\R)^\GV$  is linear and sequentially continuous 
 	\wrt the weak topology on $L_+(\hori)^\GA$ and the setwise convergence on $\Meas(\R)^\GV$, that is, 
 	$\op_v\g^n(\mathfrak T) \to \op_v\g(\mathfrak T)$ for all $\mathfrak T \in \mathcal{B}(\hori)$ and $v \in \GV$ whenever $\g^n\wto \g$ in $L_+(\hori)^\GA$. 
 \end{lemma}
 \begin{proof}
 	We prove in the following that the function 
 	$L_+(\hori) \to \Meas(\R), \measfunc \mapsto (\measfunc \cdot\leb)\circ T^{-1}$ for any measurable $T:\hori \to \hori$ 
 	is linear and is sequentially continuous \wrt   the weak topology on $L_+(\hori)$ and the setwise convergence on $\Meas(\R)$. 
 	From this, the statement of the \namecref{lem: NodeBalaCont} follows immediately.

 	For linearity, let $\measfunc_1,\measfunc_2 \in L_+(\hori)$,  $\alpha\in \R$ and $\mathfrak T \in \mathcal{B}(\hori)$ be arbitrary and observe: 
 	\begin{align*}
 		\left( \big((\alpha\cdot \measfunc_1 + \measfunc_2) \cdot\leb\big)\circ T^{-1} \right)(\mathfrak T) &\eqperdef 
 		\int_{T^{-1}(\mathfrak T)}  \alpha\cdot \measfunc_1 + \measfunc_2 \di\leb \\
 		&=\alpha\cdot  \int_{T^{-1}(\mathfrak T)}   \measfunc_1 \di\leb + \int_{T^{-1}(\mathfrak T)}  \measfunc_2 \di\leb \\
 		&\defpereq  \alpha \left((\measfunc_1 \cdot\leb)\circ T^{-1}\right)(\mathfrak T) + \left((\measfunc_2 \cdot\leb)\circ T^{-1}\right)(\mathfrak T).
 	\end{align*}
 	For the claimed continuity,  
    consider a sequence $\measfunc_n\wto \measfunc$ in $L_+(\hori)$, an arbitrary $\mathfrak T\in \mathcal{B}(\hori)$ and the following: 
    \begin{align*}
    	\lim_{n \to \infty} \left((\measfunc_n \cdot\leb)\circ T^{-1}\right)(\mathfrak T) = 
    	\lim_{n \to \infty}	\int_{T^{-1}(\mathfrak T)}  \measfunc_n  \di \leb 
    	\symoverset{1}{=} \int_{T^{-1}(\mathfrak T)}  \measfunc  \di \leb =  \left((\measfunc \cdot\leb)\circ T^{-1}\right)(\mathfrak T)
    \end{align*}
    where we used for \refsym{1} that $\measfunc_n\wto \measfunc$.
 \end{proof}
 
 From the above lemma, the extension to entire walk inflow rate vectors follows by a simple limit argument:

\begin{rstlemma}{\ref{lem: flowconW'}}
Consider an arbitrary countable collection of walks $\Routes'$, a corresponding walk inflow rate vector $h \in \edom{\Routes'}$ with   $\g:=\ell_{\Routes'}(h)$  and a node $v \in \GV$. 
Then we have 
\begin{align*} 
    \op_v \g = \sum_{\wa \in \Routes'_{v+}} h_\wa \cdot \sigma   - \sum_{\wa \in \Routes'_{v-}} (h_\wa \cdot\sigma)\circ\arr_{\wa,|\wa|+1}^{-1},
\end{align*}
where $\Routes'_{v+}$ denotes the set of walks in $\Routes'$ starting at $v$ while $\Routes'_{v-}$ denotes the set of walks in $\Routes'$ ending in  $v$. 

If, additionally, $\ell_{\wa,\abs{\wa}+1}(h_\wa)$ exist for all $\wa \in \Routes'_{v-}$, we even  have
\begin{align*}
    \op_v \g = \sum_{\wa \in \Routes'_{v+}} h_\wa \cdot \sigma - \sum_{\wa \in \Routes'_{v-}} \ell_{\wa,\abs{\wa}+1}(h_\wa) \cdot\sigma.
\end{align*}
\end{rstlemma}
\begin{proof}
 	Consider an arbitrary ordering of $\Routes'=\{\wa_k\}_{k \in \N}$ and define 
 	$\g^{k^*} := \sum_{k \leq k^*}\ell_{\wa_k}(h_{\wa_k})$  for any $k^* \in \N$. By \Cref{lem: flowcon,lem: NodeBalaCont}, we have that 
 	\begin{align*} 
 		\op_v \g^{k^*} = \sum_{k\leq k^*} \op_v \ell_{\wa_k}(h_{\wa_k}) =
 		\sum_{k \leq k^*:\wa_k \in \Routes'_{v+}} (h_{\wa_k} \cdot \leb)   - \sum_{k \leq k^*:\wa_{k} \in \Routes'_{v-}} (h_{\wa_k} \cdot\leb)\circ\arr_{\wa,|\wa|+1}^{-1}.
 	\end{align*}    
 	We know by \Cref{lem: elluExistenceProperties:ExistenceInducedFlow} that 
 	$(\ell_{\wa}(h_\wa))_{\wa\in\Routes'} \in \seql[1][\Routes'][L_+(\hori)]$ and thus 
 	$\g^{k^*} \to \g$ in $L_+(\hori)^\GA$ for $k^* \to \infty$.  
 	Hence, we get by \Cref{lem: NodeBalaCont} for arbitrary $\mathfrak T \in \mathcal{B}(\hori)$ that
 	\begin{align*}
 		&\op_v\g(\mathfrak T)\Croverset{lem: NodeBalaCont}{=} \lim_{k^*\to \infty} \op_v \g^{k^*}(\mathfrak T) =\\
 		&\quad \quad = \lim_{k^*\to \infty} \left( \sum_{k \leq k^*:\wa_k \in \Routes'_{v+}} (h_{\wa_k} \cdot \leb)(\mathfrak T)   - \sum_{k \leq k^*:\wa_{k} \in \Routes'_{v-}} (h_{\wa_k} \cdot\leb)  \circ\arr_{\wa,|\wa|+1}^{-1} (\mathfrak T) \right) \\ 
 		&\quad \quad\symoverset{1}{=}\lim_{k^*\to \infty} \sum_{k \leq k^*:\wa_k \in \Routes'_{v+}} (h_{\wa_k} \cdot \leb)(\mathfrak T)   -\lim_{k^*\to \infty}  \sum_{k \leq k^*:\wa_{k} \in \Routes'_{v-}} (h_{\wa_k} \cdot\leb)  \circ\arr_{\wa,|\wa|+1}^{-1} (\mathfrak T) \\ 
 		&\quad \quad=   \sum_{\wa_k \in \Routes'_{v+}} (h_{\wa_k} \cdot \leb)(\mathfrak T)   - \sum_{\wa_{k} \in \Routes'_{v-}} (h_{\wa_k} \cdot\leb)  \circ\arr_{\wa,|\wa|+1}^{-1} (\mathfrak T)  
 	\end{align*}
 	which shows the claim. Note that \refsym{1} holds as the limits of both terms exist in $\R$ since both of them are bounded by $\sum_{k \in \N}  \norm{h_{\wa_k}} $ which in turn is bounded due to \Cref{lem: elluContinuity:Subset} and $h \in \edom{\Routes'}$.  
 \end{proof}

From these insights, it follows directly that any appearing $\g^k$ during the execution of \Cref{alg: FlowDecompositionPseudo} is indeed \aauto $\source$,$\dest$-flow \wrt $\trav(\g,\cdot)$. Remark that any (non-\auto[)] $\source$,$\dest$-flow $\g$ 
with (non-\auto[)] net outflow rates $\inflow_v,v\in \GV$
is in particular \aauto $\source$,$\dest$-flow \wrt $\trav(\g,\cdot)$ with the same \auto net outflow rates at all nodes.

\subsection{Properties of \Auto \boldmath{$\source$,$\dest$}-Flows}\label{app:propertiesOfAutosdFlows}

In this section, we derive  the two main ingredients (\Cref{lem: ZeroCycleDecomposition} and \Cref{lem: ExistenceOfFlowCarryingWalk}) 
of the proof of our main flow decomposition result in \Cref{thm: FlowDecomp}. 
They state that \aauto $\source$,$\dest$-flow  has either a positive net outflow rate at $s$ and  admits 
a  flow-carrying \stwalk{} (\Cref{lem: ExistenceOfFlowCarryingWalk}), or, is a dynamic circulation (\Cref{lem: FlowConEveryNode}) and can be decomposed into zero-cycle inflow rates (\Cref{lem: ZeroCycleDecomposition}). 
From this, the correctness of  \Cref{alg: FlowDecompositionPseudo}  follows as the limit of the \auto $\source$,$\dest$-flows  $\g^k,k\in \N$ can not admit any flow-carrying \stwalk~$\wa_k$ due to the maximality of the corresponding  $\wflow_{\wa_k}$. 
To this end, we require several structural insights into \auto network loadings which we derive in the following.

In our first structural result, \Cref{lem: Relations:h>0u>0}, we consider an edge flow~$\g$ induced by a walk inflow vector~$h$ and validate the  intuition that whenever we have inflow into some walk~$\wa$ at some time~$t$, then this results in inflow into each of the edges on this walk at the corresponding arrival times and vice versa. 
Those implications are shown to hold for whole sets of times with positive inflow in items~\ref{lem: Relations:h>0u>0:Setwise} and~\ref{lem: Relations:h>0u>0:u>0ExistsHw>0} whereas 
items~\ref{lem: Relations:h>0u>0:Pointwise} and~\ref{lem: Relations:h>0u>0:u>0ExistsHw>0Pointwise}  
show their validity 
for almost all points in time. 
 Item~\ref{lem: Relations:h>0u>0:GoodRepres} then shows that there even
 exists a pair of suitable representatives for both walk and edge flow such that these implications hold 
\emph{for all} $t\in \hori$. 
Item~\ref{lem: Relations:h>0u>0:u>0ExistsHw>0D} extends  \ref{lem: Relations:h>0u>0:u>0ExistsHw>0Pointwise} to the sink, i.e., whenever we have positive inflow into the sink, there must be an \stwalk{} sending flow at a corresponding time. 
Finally, item~\ref{lem: Relations:h>0u>0:u>0ExistsCountableM} provides a covering version of item~\ref{lem: Relations:h>0u>0:u>0ExistsHw>0}: Any set of times~$\mathfrak T$ with positive inflow into some edge~$\arc$ can be covered by a countable set of walks -- in the sense that they have inflow at times such that the union of the corresponding arrival times at~$\arc$ covers~$\mathfrak T$.
To prove this, we require \ref{lem: Relations:h>0u>0:u>0ExistsCountableMZwischen} as an intermediate step.

\begin{lemma}\label{lem: Relations:h>0u>0} 
Let $\Routes'$ be an arbitrary countable collection of walks and $h \in \edom{\Routes'}$ with $\g: = \ell_{\Routes'}(h)$.  The following statements are true: 
    \begin{thmparts}
        \item  For all $\mathfrak T \in \mathcal{B}(\hori)$, $\wa \in \Routes$ and $j \leq \abs{\wa}$ the following implication holds: ${h}_\wa(t)>0$ for a.e.~$t \in \mathfrak T$ $\implies $ ${\g}_{\wa[j]}(t)>0$ for a.e.~$t \in \arr_{\wa,j}(\mathfrak T)$. \label[thmpart]{lem: Relations:h>0u>0:Setwise}

        \item For any edge $\arc \in \GA$, any representative of $h$ and for  all $\mathfrak T \in \mathcal{B}(\hori)$ with ${\g}_\arc(t)>0$ for a.e.~$t \in \mathfrak T$, there exists for almost every $t \in \mathfrak T$ a walk $\wa \in {\Routes'}, j \leq |\wa|$ with $\wa[j] = \arc$ and $\Tilde{t} \in \arr_{\wa,j}^{-1}(t)$ such that ${h}_\wa(\Tilde{t})> 0$. 
        \label[thmpart]{lem: Relations:h>0u>0:u>0ExistsHw>0}        
       
        \item For all $\wa \in \Routes'$ and $j \leq \abs{\wa}$, an arbitrary representative of ${\g}_{\wa[j]}$ and  almost all $t \in \hori$ the following implication holds 
        \begin{align*}
             {h}_\wa(t) > 0 \implies {\g}_{\wa[j]} (\arr_{\wa,j}(t)) >0.
        \end{align*}
       \label[thmpart]{lem: Relations:h>0u>0:Pointwise}

        \item For any edge $\arc \in \GA$, any representative of $h$ and almost all $t\in\hori$, the following implication holds 
        \begin{align*} 
 					\g_\arc(t)>0 \implies \exists \wa \in {\Routes'}, j \leq |\wa| \text{ with } \wa[j] = \arc \text{ and } \Tilde{t} \in \arr_{\wa,j}^{-1}(t) \text{ such that } {h}_\wa(\Tilde{t})> 0. 
 		\end{align*}
        \label[thmpart]{lem: Relations:h>0u>0:u>0ExistsHw>0Pointwise}

        \item \label[thmpart]{lem: Relations:h>0u>0:GoodRepres} There exist representatives of $h$ and $\g$ that fulfill 
            for all $\wa \in \Routes'$,  $j\leq \abs{\wa}$, $\arc \in \GA$  
            and \emph{for all} $t \in \hori$ the implications in \ref{lem: Relations:h>0u>0:Pointwise} and \ref{lem: Relations:h>0u>0:u>0ExistsHw>0Pointwise}.

 		\item If $\Routes'$ is a collection of \stwalk s, then, for any representative of $h$,  
 		there exists for $\op_\dest\g$-almost every $t$ a walk $\wa \in {\Routes'}$  and $\Tilde{t} \in \arr_{\wa,\abs{\wa}+1}^{-1}(t)$ such that ${h}_\wa(\Tilde{t})> 0$. 
 		\label[thmpart]{lem: Relations:h>0u>0:u>0ExistsHw>0D}

        \item Consider an arbitrary representative of $h$, an arbitrary $\mathfrak T \in \mathcal{B}(\hori)$ with $\sigma(\mathfrak T)> 0$ and a countable index set $\hat{L}$ with corresponding walk-edge-index-pairs $(\wa^{\hat l},j_{\hat l}),\hat l  \in \hat{L}$ with $\wa^{\hat l}\in \Routes'$ and $j_{\hat l}\leq \abs{\wa^{\hat l}}+1$ such that the following holds: For almost every $t \in \mathfrak T$ there exists an index $\hat l \in \hat{L}$ with corresponding $\tilde{t} \in \arr_{\wa^{\hat l},j_{\hat l}}^{-1}(t)$ such that $h_{\wa^{\hat l}}(\tilde{t})>0$. 
        Then, 
        there exists a countable set $L$ with corresponding walks $\wa^{l}\in \Routes'$, indices $j_{l}\leq \abs{\wa^{l}}+1$ and departure time sets $\mathfrak D^{l} \in \mathcal{B}(\hori)$ with $\sigma(\mathfrak D^{l})> 0$ such that
        \begin{itemize}\label[thmpart]{lem: Relations:h>0u>0:u>0ExistsCountableMZwischen}  
            \item for every $l \in L$ there exists $\hat{l}\in \hat{L}$ with $(\wa^l,j_l) = (\wa^{\hat{l}},j_{\hat{l}})$, 
            \item for every $l \in L$ we have $h_{\wa^l}(t)>0$ for a.e.~$t \in \mathfrak D^l$, 
            \item the union  $\bigcup_{l \in L}  \arr_{\wa^{l},j_{l}}(\mathfrak D^{l})$ equals $\mathfrak T$ up to a null set,
            \item the union  $\bigcup_{l \in L}  \arr_{\wa^{l},j_{l}}(\mathfrak D^{l})$ is disjoint, 
            \item all $\arr_{\wa^{l},j_{l}}(\mathfrak D^{l}), l \in L$ have positive measure and 
            \item every pair $(\wa,j), \wa\in {\Routes'}, j\leq\abs{\wa}+1$  appears at most as often in~$L$ as in~$\hat{L}$, that is, $\abs{\{l \in L\mid (\wa^{l},j_{l})=(\wa,j)\}} \leq \abs{\{\hat l \in \hat{L}\mid (\wa^{\hat l},j_{\hat l})=(\wa,j)\}}$.
        \end{itemize} 

        \item \label[thmpart]{lem: Relations:h>0u>0:u>0ExistsCountableM}  
        Consider an arbitrary $\arc \in \GA$, $\mathfrak T \in \mathcal{B}(\hori)$ with $\sigma(\mathfrak T)> 0$ and ${\g}_\arc(t)>0$ for a.e.~$t \in \mathfrak T$. Then, 
        there exists a countable set $L$ with corresponding walks $\wa^{l}\in \Routes', j_{l}\leq \abs{\wa^{l}}$ and departure time sets $\mathfrak D^{l} \in \mathcal{B}(\hori)$ with $\sigma(\mathfrak D^{l})> 0$ such that 
        \begin{itemize}
            \item for every $l \in L$ we have $h_{\wa^l}(t)>0$ for a.e.~$t \in \mathfrak D^l$, 
            \item the union  $\bigcup_{l \in L}  \arr_{\wa^{l},j_{l}}(\mathfrak D^{l})$ equals $\mathfrak T$ up to a null set,
            \item the union  $\bigcup_{l \in L}  \arr_{\wa^{l},j_{l}}(\mathfrak D^{l})$ is disjoint, 
            \item all $\arr_{\wa^{l},j_{l}}(\mathfrak D^{l}), l \in L$ have positive measure and 
            \item every $\wa \in {\Routes'}$ appears at most $\abs{\wa}$ many times, i.e.~$\abs{\{l \in L\mid \wa^{l}=\wa\}} \leq \abs{\wa}$. 
        \end{itemize} 
\end{thmparts}
\end{lemma}
\begin{proof}
    Let  $\wa\in \Routes'$, $j\leq |\wa|$ with $\arc:=\wa[j]$ be arbitrary for the following proofs of \ref{lem: Relations:h>0u>0:Setwise} and \ref{lem: Relations:h>0u>0:Pointwise}. 
\begin{structuredproof}
\proofitem{\ref{lem: Relations:h>0u>0:Setwise}}  Choose an arbitrary representative  of~$\g$. 
    Let $\mathfrak T \in \mathcal{B}(\hori)$ be arbitrary with ${h}_\wa>0$ a.e.\ on~$\mathfrak T$. Assume for the sake of a contradiction that the measurable set $\hat{\mathfrak T}:= \arr_{\wa,j}(\mathfrak T) \cap \{t \in\hori\mid \g_\arc=0 \}$ has positive measure. 
    Note that the latter set is indeed measurable by $\arr_{\wa,j}(\mathfrak T)$ being measurable as the image of a measurable set under an absolutely continuous function (\Cref{lem: PropAbsCon:ImageMeas}).  
    Then, we have ${h}_\wa(t) = 0$ for almost all $t \in   \arr_{\wa,j}^{-1}(\hat{\mathfrak T}) \cap \mathfrak{T}$ by the identity 
    $\int_{\hat{\mathfrak T}}\g_\arc\di\sigma = \sum_{\wa' \in {\Routes'}}\sum_{j': \wa'[j'] = \arc} \int_{\arr_{\wa',j'}^{-1}(\hat{\mathfrak T})}  {h}_{\wa'} \di\sigma$. 
    By assumption that  ${h}_\wa(t) > 0$ for almost all $t \in \mathfrak T$, this implies that the set $ \arr_{\wa,j}^{-1}(\hat{\mathfrak T}) \cap \mathfrak{T}$ has to be a null set. Yet, this is not possible as the image of the latter set under $\arr_{\wa,j}$ yields the set $\hat{\mathfrak T}$ which is not a null set, contradicting the property of absolutely continuous functions to have Lusin's property (\Cref{lem: PropAbsCon:Lus}).

\proofitem{\ref{lem: Relations:h>0u>0:u>0ExistsHw>0}} 
Fix an arbitrary representatives of ${h}$ and any edge~$\arc \in \GA$. 
Consider an arbitrary $\mathfrak T \in \mathcal{B}(\hori)$ with $\g_\arc(t)>0$ for a.e.~$t\in \mathfrak T$. Define 
\begin{align*}
    \hat{\mathfrak T} :=\{t \in \mathfrak T \mid \nexists \wa \in {\Routes'}, j\leq |\wa| \text{ with } \wa[j] = \arc \text{ and }\Tilde{t} \in \arr_{\wa,j}^{-1}(t): h_\wa(\Tilde{t})>0 \}.
\end{align*}
We have to show that $\hat{\mathfrak T}$ is a null set. 
We start by observing that the latter set is measurable as we can represent it as follows: 
\begin{align*}
    \hat{\mathfrak T} = \mathfrak T \setminus \Bigl( \bigcup_{\wa \in {\Routes'}}\bigcup_{j \leq|\wa|:\wa[j] = \arc} \arr_{\wa,j}\bigl(\{t \in \hori\mid h_\wa(t)>0\}\bigr)\Bigr)
\end{align*}
where we note that the unions are countable and the individual occurring sets are measurable \Cref{lem: PropAbsCon:ImageMeas}

Now, by definition of ${\g}$, we have $\int_{\hat{\mathfrak T}} {\g}_\arc \di\sigma =  \sum_{\wa \in {\Routes'}}\sum_{j:\wa[j] = \arc} \int_{\arr_{\wa,j}^{-1}(\hat{\mathfrak T})} {h}_\wa \di\sigma$ while at the same time, by definition of $\hat{\mathfrak T}$, we have ${h}_\wa(t)= 0$ for all $t \in \arr_{\wa,j}^{-1}(\hat{\mathfrak T})$ and all $\wa \in {\Routes'}, j \leq |\wa|$ with $\wa[j]= \arc$. Hence, the right hand side of the last equality is equal to zero, implying that $\hat{\mathfrak T}$ has to be a null set since, by definition of $\mathfrak T\supseteq \hat{\mathfrak T}$, we have $\g_\arc(t)> 0$ for almost all $t \in \mathfrak T$. 

\proofitem{\ref{lem: Relations:h>0u>0:Pointwise}} 
    Let $\g_{\arc}$ and $h$ be arbitrary representatives. 
    Consider the measurable set $\mathfrak T := \{t \in \hori\mid h_\wa(t)>0, \g_{\arc} (\arr_{\wa,j}(t)) =0\}$. 
    Then we have 
    \begin{align*}
        0 = \int_{\arr_{\wa,j}(\mathfrak T)}  \g_\arc \di\sigma =  \int_{\arr_{\wa,j}^{-1}(\arr_{\wa,j}(\mathfrak T))}h_\wa \di\sigma   \geq  \int_{\mathfrak T}h_\wa \di\sigma  
    \end{align*}
    showing that $h_\wa(t) = 0$ for a.e.\  $t\in \mathfrak T$ (as $h\in L_+(\hori)$). By definition of $\mathfrak T$, this is only possible if $\mathfrak T$ is itself a null set.

\proofitem{\ref{lem: Relations:h>0u>0:u>0ExistsHw>0Pointwise}} 
This is a direct consequence of \ref{lem: Relations:h>0u>0:u>0ExistsHw>0} for $\mathfrak T = \hori$.

\proofitem{\ref{lem: Relations:h>0u>0:GoodRepres}} We first choose an arbitrary representatives of $\g$ and \wrt to the latter a representative of $h$ such that 
 		the implication in \ref{lem: Relations:h>0u>0:Pointwise} holds for all $t \in \hori$ and all $\wa \in \Routes',j\leq \abs{\wa}$. This is possible since  $\Routes'$ is countable, every walk in $\Routes'$ is finite, and, for any fixed pair $\wa \in \Routes'$ and $j\leq \abs{\wa}$, the statement already holds for almost all $t \in \R$ by~\ref{lem: Relations:h>0u>0:Pointwise}. 
 Next, for every $\arc \in \GA$ let $\mathfrak T_\arc$ be the set of times where the second implication does not hold. This set is measurable, cf.~the proof of \ref{lem: Relations:h>0u>0:u>0ExistsHw>0}. 
 Hence,  \ref{lem: Relations:h>0u>0:u>0ExistsHw>0} implies that this set must be a null set. Therefore, adjusting for all $\arc \in \GA$ the representative of $\g_\arc$  on $\mathfrak T_\arc$ by setting it to zero results in another representative  of $\g$ fulfilling the second implication for all $t\in \hori$ and all $\arc \in \GA$. 
 Now, it remains to observe that this adjustment preserves the fulfillment of the implication in~\ref{lem: Relations:h>0u>0:Pointwise}:
Assume for the sake of a contradiction that there exists $\wa \in \Routes'$, $j\leq \abs{\wa}$ and $t \in \hori$ with $h_\wa(t)>0$
 and $\g_{\wa[j]}(\arr_{\wa,j}(t)) =0$ (after the adjustment). 
 Since, prior to the adjustment,  the implication in \ref{lem: Relations:h>0u>0:Pointwise} did hold, it follows that 
 $\arr_{\wa,j}(t) \in \mathfrak T_{\wa[j]}$. But this contradicts the fact that $t \in \arr_{\wa,j}^{-1}(\arr_{\wa,j}(t))$ and $h_\wa(t) > 0$. Hence, the proof is finished. 

\proofitem{\ref{lem: Relations:h>0u>0:u>0ExistsHw>0D}} 
 		Fix an arbitrary representatives of ${h}$. 
 		Define 
 		\begin{align*}
 			\hat{\mathfrak T} \coloneq\{t \in \hori \mid \nexists \wa \in {\Routes'} \text{ with } \Tilde{t} \in \arr_{\wa,\abs{\wa}+1}^{-1}(t): h_\wa(\Tilde{t})>0 \}.
 		\end{align*}
 		Then $\hat{\mathfrak T}\in \mathcal{B}(\hori)$ follows analogously to \ref{lem: Relations:h>0u>0:u>0ExistsHw>0}. 
 			We have to show that $\hat{\mathfrak T}$ is a $\op_\dest\g$-null set. 
 		By \Cref{lem: flowconW'} and $\Routes'$ being a collection of \stwalk s, we know that 
 		\begin{align*}
 		\op_\dest\g(\hat{\mathfrak T}) =	- \sum_{\wa \in \Routes'} (h_\wa \cdot\leb)\circ\arr_{\wa,|\wa|+1}^{-1}(\hat{\mathfrak T}).
 		\end{align*}
 		By definition of $\hat{\mathfrak{T}}$, we know that for every $\wa\in \Routes'$ we have 
 		\begin{align*}
 			(h_\wa \cdot\leb)\circ\arr_{\wa,|\wa|+1}^{-1}(\hat{\mathfrak T}) = \int_{\arr_{\wa,|\wa|+1}^{-1}(\hat{\mathfrak T})}h_\wa\di\leb =  \int_{\arr_{\wa,|\wa|+1}^{-1}(\hat{\mathfrak T})}0\di\leb=0
 		\end{align*}
 		and thus the claim follows.

\proofitem{\ref{lem: Relations:h>0u>0:u>0ExistsCountableMZwischen}} 
Consider $h$, $\mathfrak T$ and $\hat{L}$ as described. 
Define for any $l \in \hat{L}$ the set 
\begin{align*}
    \hat{\mathfrak D}^l := \arr_{\wa^l,j_l}^{-1}(\mathfrak T) \cap \{t \in \hori\mid h_{\wa^l}(t)>0\}.
\end{align*}
 Remark that these sets are measurable due to the measurability of $\arr$, $\mathfrak T$ and $h$. 
We now consider the union $\bigcup_{l\in \hat{L}} \arr_{\wa^l,j_l}\big(\hat{\mathfrak D}^{l}\big)$ 
and claim that this union gives us all of $\mathfrak T$ up to a null set, that is, the set given by 
$\hat{\mathfrak T}:=  \mathfrak T \setminus \bigcup_{l\in \hat{L}} \arr_{\wa^l,j_l}\big(\hat{\mathfrak D}^{l}\big)$ is a null set.  
By  assumption, 
there exists for almost every $t \in \hat{\mathfrak T}$ an index $l \in \hat{L}$ with  walk $\wa^l \in \Routes'$, index $j_l\leq \abs{\wa}+1$ and time $\tilde{t} \in \arr_{\wa^l,j_l}^{-1}(t)$ such that $h_{\wa^l}(\tilde{t})>0$.
This, in particular, implies $\tilde t \in \hat{\mathfrak D}^l$ (by definition of $\hat{\mathfrak D}^l$) and, therefore, $t = \arr_{\wa^l,j_l}(\tilde t) \in \arr_{\wa^l,j_l}\big(\hat{\mathfrak D}^{l}\big)$, which, in turn, implies $t \notin \hat{\mathfrak T}$ by the definition of $\hat{\mathfrak T}$. Hence, almost every $t \in \hat{\mathfrak T}$ is not contained in $\hat{\mathfrak T}$, which is only possible if $\sigma(\hat{\mathfrak T})=0$. 

Since $\hat{L}$ is countable, we can assume \wlg that it is a set of ascending natural numbers $\hat{L}=\{1,2,\ldots\}\subseteq \N$. 
We define for any $l\in \hat{L}$
 \begin{align*}
      \mathfrak D^{l} :=  \hat{\mathfrak D}^{l} \setminus  \arr_{\wa^l,j_l}^{-1}\Big(\bigcup_{l'\in \N:l'< l} \arr_{\wa^{l'},j_{l'}} \big( \hat{\mathfrak D}^{l'} \big) \Big)
 \end{align*} 
 and remark that these sets are measurable by the measurability of the functions $\arr_{\wa^l,j_l}$ and the sets  $\hat{\mathfrak D}^{l}$.

Now let $L\subseteq \hat{L}$ be the subset of all those indices where $\mathfrak D^l$ is not a null set.  
Clearly, this choice guarantees fulfillment of the first, second, fourth and last property required in~\ref{lem: Relations:h>0u>0:u>0ExistsCountableMZwischen}.
Moreover, we get immediately by construction the following:
\begin{align*}
    \bigcup_{l \in L} \arr_{\wa^l,j_l}\big(\mathfrak D^{l}\big) \symoverset{1}{\subseteq} \bigcup_{l \in \hat{L}} \arr_{\wa^l,j_l}\big(\mathfrak D^{l}\big)  = \bigcup_{l \in \hat{L}} \arr_{\wa^l,j_l}\big(\hat{\mathfrak D}^{l}\big)  \symoverset{3}{\subseteq} \mathfrak T. 
\end{align*}
Here, as remarked above, the inclusion  \refsym{3} is in fact an equality up to a null set.  
Moreover, the same is also true for the  inclusion \refsym{1}  since $\hat{L}\setminus L$ is countable and all $\arr_{\wa,j}$ have Lusin's property, which implies (\Cref{lem: PropAbsCon:Lus}) that $\arr_{\wa^l,j_l}(\mathfrak D^l)$ is a null set whenever $\mathfrak D^l$ is a null set. Thus, $\bigcup_{l \in L} \arr_{\wa^l,j_l}\big(\mathfrak D^{l}\big)$ equals $\mathfrak T$ up to a null set. 

Finally, we remark that  by $\ell_{\wa}(h_\wa),\wa \in \hat{\Routes}$ existing and \Cref{lem: elluExistenceProperties:ExistenceInducedFlowOnJthEdge},  $\arr_{\wa^{l},j_{l}}(\mathfrak D^{l})$ is not a null set for any $l \in L$ 
as $\sigma(\mathfrak D^l)> 0$ and $h_{\wa^l}>0$ on $\mathfrak D^l$.

\proofitem{\ref{lem: Relations:h>0u>0:u>0ExistsCountableM}} 
Consider $\arc \in \GA$ and $\mathfrak T$ as given as well as an arbitrary ordering $\{\wa^k\}_{k \in \N}$ on the set $\Routes'$. 
Define $\hat{L}$ as the countable set of all $(k,j) \in \N^2$ for which 
$\wa^k[j]=\arc$. Then, by \ref{lem: Relations:h>0u>0:u>0ExistsHw>0}, $\mathfrak T$ and $\hat{L}$ fulfill the required property of 
\ref{lem: Relations:h>0u>0:u>0ExistsCountableMZwischen} which shows the statement. \qedhere
\end{structuredproof} 
\end{proof}

The  next \namecref{lem: elluPropagation} shows how for a single walk inflow the inflow rates on two edges on this walk are connected via an \auto network loading along the partial walk between these two edges.

\begin{lemma}\label{lem: elluPropagation}
     Consider an arbitrary walk $\wa$, two indices $j_1\leq j_2\leq|\wa|+1$ and $h_\wa \in \edom{\wa,j_1}$. Then, 
     $\ell_{\wa_{\geq j_1},j_2-j_1+1}(\ell_{\wa,j_1}(h_\wa))$ exists if and only if $\ell_{\wa,j_2}(h_\wa)$ exists,  in which case they coincide, i.e.~$\ell_{\wa,j_2}(h_\wa) = \ell_{\wa_{\geq j_1},j_2-j_1+1}(\ell_{\wa,j_1}(h_\wa))$. 
\end{lemma}
\begin{proof}
We observe for an arbitrary $\mathfrak T \in \mathcal{B}(\hori)$: 
\begin{align}\label{eq: elluPropagation}
  \int_{\arr_{\wa,j_2}^{-1}(\mathfrak T)}h_\wa\di\sigma \symoverset{1}{=}\int_{\arr_{\wa,j_1}^{-1}(\arr_{\wa_{\geq j_1},j_2-j_1+1}^{-1}(\mathfrak T))}h_\wa\di\sigma \symoverset{2}{=}  \int_{\arr_{\wa_{\geq j_1},j_2-j_1+1}^{-1}(\mathfrak T)} \ell_{\wa,j_1}(h_\wa)\di\sigma
\end{align}
where we used for \refsym{1}  that $\arr_{\wa_{\geq j_1},j_2-j_1+1} \circ \arr_{\wa,j_1} =  \arr_{\wa,j_2}$ and for \refsym{2} the existence of $\ell_{\wa,j_1}(h_\wa)$. 

Let us argue for the existence statements in the \namecref{lem: elluPropagation} first: If $\ell_{\wa,j_2}(h_\wa)$ exists, then the left hand side of \eqref{eq: elluPropagation} is equal to $\int_{\mathfrak T} \ell_{\wa,j_2}(h_\wa)\di\sigma$ and in particular equal to $0$ for every null set $\mathfrak T$ which shows the existence of $\ell_{\wa_{\geq j_1},j_2-j_1+1}(\ell_{\wa,j_1}(h_\wa))$ by \Cref{lem: elluExistenceProperties:ExistenceInducedFlowOnJthEdge}. 

Similarly, if $\ell_{\wa_{\geq j_1},j_2-j_1+1}(\ell_{\wa,j_1}(h_\wa))$ exists, then the right hand side of \eqref{eq: elluPropagation} 
is equal to the integral $\int_{\mathfrak T} \ell_{\wa_{\geq j_1},j_2-j_1+1}(\ell_{\wa,j_1}(h_\wa))\di\sigma$ 
 and in particular equal to $0$ for every null set $\mathfrak T$ which shows the existence of $\ell_{\wa,j_2}(h_\wa)$ by \Cref{lem: elluExistenceProperties:ExistenceInducedFlowOnJthEdge}. 

The stated equality now follows immediately by \eqref{eq: elluPropagation}. 
\end{proof}

Next, we demonstrate that if the induced flow of a walk inflow rate $h_\wa$ exists, then integrating $h_\wa$ over a set $\mathfrak T$ yields the same result as integrating over the (potentially larger) set $\arr_{\wa,j}^{-1}(\arr_{\wa,j}(\mathfrak T))$.
\begin{lemma}\label{lem: elluinj}
       Consider an arbitrary walk $\wa$, $j\in [|\wa|+1]$ and $h_\wa \in L_+(\hori)$. If $\ell_{\wa,j}(h_\wa)$ exists, then $h_\wa = 0$ (a.e.) on $\arr_{\wa,j}^{-1}(\arr_{\wa,j}(\mathfrak T))\setminus \mathfrak T$ and in particular $\int_{\mathfrak T} h_\wa \di\leb = \int_{\arr_{\wa,j}^{-1}(\arr_{\wa,j}(\mathfrak T))} h_\wa \di\leb$ for any $\mathfrak T\in \mathcal{B}(\hori)$.
\end{lemma}
\begin{proof}
    Define $\mathfrak T_{\mathrm{sing}}$ as the set of times $t \in \hori$ with $\arr_{\wa,j}^{-1}(t)$ being a singleton. 
    Since $\arr_{\wa,j}$ is monotone increasing, 
    the collection of sets $\arr_{\wa,j}^{-1}(t), t \in   \hori\setminus \mathfrak T_{\mathrm{sing}}$ is a collection of non-empty disjoint intervals. This implies that  
    $\hori\setminus \mathfrak T_{\mathrm{sing}}$ 
    is countable and, in particular, a null set.  Furthermore,  $ \arr_{\wa,j}^{-1}(\arr_{\wa,j}(\mathfrak T))\setminus \mathfrak T \subseteq \arr_{\wa,j}^{-1}(\hori \setminus \mathfrak T_{\mathrm{sing}})$ which implies the claim by \Cref{lem: elluExistenceProperties:ExistenceInducedFlowOnJthEdge}. 
\end{proof}  

As a further structural insight, we show that  the mapping $\ell_{\wa,j}$ is  an order embedding, meaning that  a larger walk inflow rate will lead to a larger edge flow  and vice versa.

\begin{lemma}\label{lem: ellOrderPreserving}
    Consider an arbitrary walk  $\wa$, $j\in [|\wa|+1]$ and $h_{\wa},\Tilde{h}_\wa \in \edom{\wa,j}$. Then, we have
        \[\ell_{\wa,j}(h_\wa) \leq \ell_{\wa,j}(\tilde{h}_\wa) \iff h_\wa \leq \Tilde{h}_\wa.\]
    The analogue statement holds for $<$ instead of $\leq$ where $\wflow < \tilde{\wflow}$ for $\wflow,\tilde{\wflow} \in L_+(\hori)$ means that $\wflow \leq \tilde{\wflow}$ and $\wflow \neq \tilde{\wflow}$.  
\end{lemma}
\begin{proof}
    We first start with the statements for $\leq$ and prove both direction separately: 
    \begin{structuredproof}
          \proofitem{``$\Leftarrow$"}
        Let $\mathfrak T \in \mathcal{B}(\hori)$ be arbitrary. Then we have
        \begin{align*}
           \int_{\mathfrak T} \ell_{\wa,j}(h_\wa) \di\sigma = \int_{\arr_{\wa,j}^{-1}(\mathfrak T)} h_\wa \di\sigma
           \leq  \int_{\arr_{\wa,j}^{-1}(\mathfrak T)} \tilde{h}_\wa \di\sigma = 
           \int_{\mathfrak T} \ell_{\wa,j}(\tilde{h}_\wa) \di\sigma
        \end{align*}
        which shows $\ell_{\wa,j}(h_\wa) \leq \ell_{\wa,j}(\tilde{h}_\wa)$ since $\mathfrak T$ was arbitrary. 
        \proofitem{``$\Rightarrow$"} 
        Let $\mathfrak T \in \mathcal{B}(\hori)$ be arbitrary. Then we have
        \begin{align*}
            \int_{\mathfrak T} h_\wa \di\sigma 
            &\symoverset{1}{=} \int_{\arr_{\wa,j}^{-1}(\arr_{\wa,j}(\mathfrak T))} h_\wa \di\sigma
            =\int_{\arr_{\wa,j}(\mathfrak T)} \ell_{\wa,j}(h_\wa) \di\sigma\\
            &\leq \int_{\arr_{\wa,j}(\mathfrak T)} \ell_{\wa,j}(\tilde{h}_\wa)  \di\sigma 
           = \int_{\arr_{\wa,j}^{-1}(\arr_{\wa,j}(\mathfrak T))} \tilde{h}_\wa \di\sigma\\
          &\symoverset{1}{=} \int_{\mathfrak T} \tilde{h}_\wa \di\sigma 
        \end{align*}     
        where the equalities indicated by \refsym{1} hold due to \Cref{lem: elluinj}. Hence, the claimed inequality is true since $\mathfrak T$ was arbitrary. 
    \end{structuredproof}
    Now the statement for $<$ follows directly by the above and the equality $\int_{\hori} \ell_{\wa,j}(\hat{\wflow}_\wa)\di\sigma = \int_{\hori}\hat{\wflow}_\wa\di\sigma$ for any $\hat{\wflow}_\wa \in L_+(\hori)$. 
\end{proof}

We can even sharpen the previous result and show that if $h_\wa \leq \tilde{h}_\wa$ on some subset $\mathfrak D$ of starting times, then the induced flows fulfill $\ell_{\wa,j}(h_\wa) \leq \ell_{\wa,j}(\tilde{h}_\wa)$  on the arrival times $\arr_{\wa,j}(\mathfrak D)$ at the $j$-th edge. 
In order to show this, we need the following lemma demonstrating that 
 $\ell_{\wa,j}$ commutes with indicator functions. Here, 
we denote for any set $S$ and subset $S'\subseteq S$ the indicator function by $1_{S'}:S\to \{0,1\}$ with $1_{S'}(s) = 1$ if $s \in S'$ and $1_{S'}(s) = 0$ else. 

\begin{lemma} \label{lem: elluindi}
       Consider an arbitrary walk $\wa$, $j\leq |\wa|+1$ and $h_\wa \in\edom{\wa}$.  
       Then for any $\mathfrak T^* \in \mathcal{B}(\hori)$, the flow $\ell_{\wa,j}(1_{\mathfrak T^*}\cdot h_\wa)$ exists and  fulfills $\ell_{\wa,j}(1_{\mathfrak T^*}\cdot h_\wa) = 1_{\arr_{\wa,j}(\mathfrak T^*)}\cdot\ell_{\wa,j}(h_\wa)$.  
\end{lemma}
\begin{proof}
    Let $\mathfrak T \in \mathcal{B}(\hori)$ be arbitrary. We calculate: 
    \begin{align*}
         \int_{\mathfrak T} 1_{\arr_{\wa,j}(\mathfrak T^*)}\cdot\ell_{\wa,j}(h_\wa) \di\sigma &= \int_{\mathfrak T \cap\arr_{\wa,j} (\mathfrak T^*)}   \ell_{\wa,j}(h_\wa)\di\sigma = \int_{\arr_{\wa,j}^{-1}(\mathfrak T \cap\arr_{\wa,j}(\mathfrak T^*))  }   h_\wa \di\sigma  \\
        &=\int_{\arr_{\wa,j}^{-1}(\mathfrak T) \cap \arr_{\wa,j}^{-1}(\arr_{\wa,j}(\mathfrak T^*))}  h_\wa \di\sigma\\
         &\symoverset{1}{=}    \int_{\arr_{\wa,j}^{-1}(\mathfrak T) \cap \mathfrak T^*}  h_\wa \di\sigma  = \int_{\arr_{\wa,j}^{-1}(\mathfrak T)}  1_{\mathfrak T^*} \cdot h_\wa \di\sigma
    \end{align*}
    where equality \refsym{1} holds due to \Cref{lem: elluinj}. Hence, $  1_{\arr_{\wa,j}(\mathfrak T^*)}\cdot\ell_{\wa,j}(h_\wa)$ fulfills \eqref{eq: Deflwj} which shows the claim. 
\end{proof}

As an immediate consequence of the previous two lemmas we get:

\begin{lemma}\label{lem: ellOrderPreservingSharpened}
    Consider an arbitrary walk  $\wa$, $j\in [|\wa|+1]$ and $h_{\wa},\Tilde{h}_\wa \in \edom{\wa}$. 
    For any set $\mathfrak D \in \mathcal{B}(\hori)$ the inequality 
     $\ell_{\wa,j}(h_\wa) \leq \ell_{\wa,j}(\tilde{h}_\wa)$ on $\arr_{\wa,j}(\mathfrak D)$ is equivalent to $h_\wa \leq \Tilde{h}_\wa$ on $\mathfrak D$. The analogue statement holds for $<$ instead of $\leq$. 
\end{lemma}

\begin{proof}
    By \Cref{lem: elluindi}, we get $\ell_{\wa,j}(h_\wa) \cdot 1_{\arr_{\wa,j}(\mathfrak D)} = \ell_{\wa,j}(h_\wa \cdot 1_{\mathfrak D})$ and  $\ell_{\wa,j}(\tilde{h}_\wa) \cdot 1_{\arr_{\wa,j}(\mathfrak D)} = \ell_{\wa,j}(\tilde{h}_\wa \cdot 1_{\mathfrak D})$. Hence, the claim then follows by \Cref{lem: ellOrderPreserving}. 
\end{proof}

Finally, we show that any flow arriving at an subwalk with zero traversal time induces the same flow on each of its edges as well as on the first subsequent edge. 
\begin{lemma}\label{lem: FLowOnZeroTrav}
Consider an arbitrary walk $\wa$, two edge indices $j_1<j_2\leq|\wa|+1$ and $h_\wa \in\edom{\wa,j_1}$.
 Furthermore, let $\mathfrak D \in \mathcal{B}(\hori)$ be a set for which for almost every $t \in \mathfrak D$, we have  $\arr_{\wa,j_1}(t) = \arr_{\wa,j_2}(t)$. 
Then, $\ell_{\wa,j'}(h_\wa\cdot 1_{\mathfrak{D}}), j'\in\{j_1+1,\ldots,j_2\}$ exist and fulfill  $\ell_{\wa,j'}(h_\wa \cdot 1_{\mathfrak{D}} ) = \ell_{\wa,j_1}(h_\wa) \cdot 1_{\arr_{\wa,j_1}(\mathfrak D)}$. 

In particular, if  $h_\wa = 0$ on $\hori\setminus \mathfrak D$, then $\ell_{\wa,j'}(h_\wa)$ exists and is equal to $\ell_{\wa,j_1}(h_\wa)$ on the whole set $\hori$ for all $j'\in\{j_1+1,\ldots,j_2\}$. For a zero-cycle inflow rate $h_c$ into a cycle $c$, the latter statement shows that $\ell_{c,\arc}(h_c) = \sum_{j\leq\abs{c}:c[j] = \arc}h_c$ for every $\arc \in c$. 
\end{lemma}
\begin{proof}
Choose an arbitrary representative of $h_\wa$ and define 
\begin{align*}
    \mathfrak D^*:= \{t \in \mathfrak D\mid   h_\wa(t)>0 \text{ and } \arr_{\wa,j'}(t) = \arr_{\wa,j_2}(t),j' \in\{j_1,\ldots,j_2\} \}.  
\end{align*}
      By assumption, we have $\sigma(\mathfrak D \setminus \mathfrak D^*) = 0$.
    Let $j' \in\{j_1,\ldots,j_2\}$ and   $\mathfrak T \in \mathcal{B}(\hori)$ be arbitrary. Then we have 
    \begin{align}\label{eq: lem:FlowOnZeroTrav}
            \arr_{{\wa},j'}^{-1}(\mathfrak T) \cap \mathfrak D^* = \arr_{{\wa},j_1}^{-1}(\mathfrak T) \cap \mathfrak D^*,
    \end{align}
 which allows us to derive for an arbitrary set $\mathfrak T \in \mathcal{B}(\hori)$: 
\begin{align*}
    \int_{\mathfrak T} \ell_{\wa,j_1}(h_\wa) \cdot  1_{\arr_{\wa,j_1}(\mathfrak D)} \di\sigma &\symoverset{3}{=}  \int_{\mathfrak T} \ell_{\wa,j_1}(h_\wa  \cdot  1_{\mathfrak D})\di\sigma =   \int_{\arr_{\wa,j_1}^{-1}(\mathfrak T)}h_\wa \cdot  1_{\mathfrak D} \di\sigma \\
    &\symoverset{1}{=}  \int_{\arr_{\wa,j_1}^{-1}(\mathfrak T)}h_\wa \cdot  1_{\mathfrak D^*} \di\sigma \\
    &=  \int_{\arr_{\wa,j_1}^{-1}(\mathfrak T)\cap \mathfrak D^*}h_\wa   \di\sigma 
    \overset{\eqref{eq: lem:FlowOnZeroTrav}}{=} \int_{\arr_{\wa,j'}^{-1}(\mathfrak T)\cap \mathfrak D^*}h_\wa   \di\sigma\\
    &= \int_{\arr_{\wa,j'}^{-1}(\mathfrak T) }h_\wa \cdot 1_{\mathfrak D^*}  \di\sigma \\
    &\symoverset{1}{=}   \int_{\arr_{\wa,j'}^{-1}(\mathfrak T) }h_\wa \cdot 1_{\mathfrak D}  \di\sigma 
\end{align*}
where the equality indicated by 
\refsym{3} follow by \Cref{lem: elluindi} and the ones indicated by 
\refsym{1} hold since   $\sigma(\mathfrak D \setminus \mathfrak D^*) = 0$. 
Hence, $\ell_{\wa,j_1}(h_\wa) \cdot  1_{\arr_{\wa,j_1}(\mathfrak D)}$ fulfills \eqref{eq: Deflwj} 
for $h_\wa \cdot 1_{\mathfrak D}$ and thus the claimed equality holds.

Moreover, in case of $h_\wa = 0$ on $\hori \setminus \mathfrak D$, 
we get 
\begin{align*}
 \ell_{\wa,j_1}(h_\wa) \symoverset{2}{=} \ell_{\wa,j_1}(h_\wa \cdot  1_{\mathfrak D}) \symoverset{1}{=}  \ell_{\wa,j_1}(h_\wa) \cdot 1_{\arr_{\wa,j_1}(\mathfrak D)}  \symoverset{3}{=} \ell_{\wa,j'}(h_\wa\cdot 1_{\mathfrak D}) \symoverset{2}{=} \ell_{\wa,j'}(h_\wa) 
\end{align*}
where we used for the equalities indicated with \refsym{2} that $h_\wa = 0$ on $\hori \setminus \mathfrak D$. The equality indicated with \refsym{1} follows by \Cref{lem: elluindi} while we used for equality \refsym{3} the first part the \namecref{lem: FLowOnZeroTrav}.
 \end{proof}

In the follow we now show that any flow satisfying flow conservation at all nodes except for the destination, also satisfies flow conservation at the destination (i.e., is a dynamic circulation) and can be induced by zero-cycle inflow rates only. 

As a first step, we show that we can rewrite the total travel time of a flow in terms of only the node balances.

 \begin{lemma}\label{lem: flowcontrav}    For any  $\eflow \in L_+(\hori)^\GA$, we have
 \begin{align*}
      \dup{\trav}{\eflow} =  \sum_{v \in \GV}\int_{\hori} -\id \di (\op_v   \eflow) \in \R_+ \cup \{\infty\}
 \end{align*}
  with $-\id$ denoting the identity function in $L(\hori)$.  
 \end{lemma}
 
Intuitively, the two sides of the above equality correspond to two different ways of computing the total travel time of all particles without counting any waiting: On the left, one sums the product of current travel time and edge inflow rate over all times and edges. To understand the right side, consider the definition of the node balance as difference between out- and inflow rate at that node. Hence, we get as positive contribution the current time whenever a particle arrives at a node and, analogously, a negative one for particles leaving a node. Thus, in total we get the sum of all node-to-node travel times of all particles.
 
 \begin{proof}[Proof of \Cref{lem: flowcontrav}]
We calculate for an arbitrary  $\eflow \in L_+(\hori)^\GA$ 
\begin{align*}
    \lim_{n \to \infty} \sum_{v \in \GV}\int_{[-n,n]} -&\id \di (\op_v   \eflow) =\\ 
     &\symoverset{2}{=}  \lim_{n \to \infty} \sum_{v \in \GV} \sum_{\arc \in \delta^+(v)} \int_{[-n,n]} -\id \di (\eflow_\arc\cdot \sigma) - \sum_{\arc \in \delta^-(v)} \int_{[-n,n]} -\id \di ((\eflow_\arc\cdot \sigma) \circ \exit_\arc^{-1}) \\
&\symoverset{1}{=} \lim_{n \to \infty} \sum_{v \in \GV} \sum_{\arc \in \delta^+(v)} \int_{[-n,n]} -\id \di (\eflow_\arc\cdot \sigma) - \sum_{\arc \in \delta^-(v)} \int_{[-n,n]} -\exit_\arc  \di (\eflow_\arc\cdot \sigma) \\   
     &= \lim_{n \to \infty}  \sum_{\arc \in \GA} \int_{[-n,n]} -\id + \exit_\arc   \di (\eflow_\arc\cdot \sigma)= \lim_{n \to \infty} \sum_{\arc \in \GA} \int_{[-n,n]} \trav_\arc  \di (\eflow_\arc\cdot \sigma)\\ 
     &= \lim_{n \to \infty}\sum_{\arc \in \GA} \int_{[-n,n]} \trav_\arc\cdot \eflow_\arc  \di \sigma = \dup{\trav}{\eflow}
\end{align*}
where we used in \refsym{2} that $-\id$ is integrable on $[-n,n]$ \wrt $\g_\arc\cdot \sigma$ and $(\g_\arc\cdot \sigma)\circ \exit_\arc^{-1}$ for all $\arc \in \GA$. 
The equality indicated with \refsym{1} holds by the change of variables formula (\cite[Theorem 3.6.1]{Bogachev2007I}). 
 \end{proof}

With this we can now show that \aauto  $\source$,$\dest$-flow fulfilling  flow conservation at all nodes except for the destination, must already be a dynamic circulation and can only use edges of zero travel time. 

\begin{lemma}\label{lem: FlowConEveryNode}
    Let $\eflow \in L_+(\hori)^\GA$ be \aauto $\source$,$\dest$-flow fulfilling flow conservation also at $s$. 
    Then, $\eflow$ is \aauto dynamic circulation  and we have
    \begin{align*}
        {\eflow}_\arc(t) > 0 \implies \trav_\arc(t) = 0 \text{ for almost all } t \in \hori \text{ and all } \arc \in \GA.   
    \end{align*}
\end{lemma}
\begin{proof}
   Consider 
     \begin{align*}
        &\sum_{\arc \in \GA} \Big(   (\eflow_\arc \cdot \sigma) -    (\eflow_\arc \cdot \sigma) \circ \exit_\arc^{-1} \Big) 
        =\sum_{v \in \GV} \Big(\sum_{\arc \in \delta^+(v)} \eflow_\arc \cdot \sigma -  \sum_{\arc \in \delta^-(v)} (\eflow_\arc \cdot \sigma) \circ \exit_\arc^{-1}  \Big)\\
        = & \sum_{v \in \GV} \op_v\eflow  = \op_\dest\eflow   \leq 0.
    \end{align*}
    Furthermore, we observe that each summand  in the first sum assigns a nonnegative measure to any interval of the form $\startint t]$. This is true as $\eflow_\arc\in L_+(\hori)$ and  $\exit_\arc^{-1}({\startint}t]) \subseteq {\startint}t]$ for all $t \in \hori,\arc \in \GA$ since edge traversal times are nonnegative. 
   Hence, we know that $\op_\dest \eflow$ must assign~$0$ to any such interval and, subsequently, is the zero measure. Thus, $\op \eflow = 0$. From this, the second part of the statement follows as well since we have
     $\dup{\trav}{\eflow} =  \sum_{v \in \GV}\int_{\hori} -\id \di (\op_v   \eflow) = 0$ by \Cref{lem: flowcontrav} and both~$\eflow$ and~$\trav$ are nonnegative. 
\end{proof}

In order to deduce from the previous \namecref{lem: FlowConEveryNode} that such a flow can be induced by zero-cycle inflow rates only, we 
require several insights regarding   the (edge) outflow rate $\eflow_\arc^-$ of a corresponding  edge inflow rate  $\eflow_\arc$, i.e.~the equivalence class in $L_+(\hori)$ fulfilling 
\begin{align}\label{eq: outflowCondition}
    \int_{\mathfrak T}\eflow_\arc^- \di\sigma =   \int_{\exit_\arc^{-1}(\mathfrak T)}\eflow_\arc \di\sigma \text{ for all }\mathfrak T \in \mathcal{B}(\hori).
\end{align}


\begin{lemma} \label{lem: outflow}
Let  $\eflow \in L_+(\hori)^\GA$ and $\arc \in \GA$ be arbitrary. The following statements are true:
\begin{thmparts}
    \item \label[thmpart]{lem: outflow: =ell+Cons}  The outflow rate $\eflow_\arc^-$ exists if and only if $\eflow_\arc \in \edom{\wa,2}$ for $\wa = (\arc)$, in which case $\eflow_\arc^- = \ell_{\wa,2}(\eflow_\arc)$. That is, the outflow rate $\eflow_\arc^-$ equals the inflow rate into the end node of the single-edge walk $\wa=(\arc)$. As a direct consequence, we get 
    \begin{enumerate}
        \item \label{lem: outflow: Exis}  The 
    outflow rate $\eflow_\arc^-$ exists and is then uniquely determined if
     and only if 
     \begin{align}\label{eq: outflowExists}
              \eflow_\arc = 0  \text{ on }\exit_\arc^{-1}(\mathfrak T) \text{ for any null set }\mathfrak T\subseteq \hori. 
     \end{align} 
     \item \label{lem: outflow: ExisLeq}  If $\eflow_\arc \leq \tilde{\eflow}_\arc$ for some $\tilde\eflow_\arc \in L_+(\hori)$ and $\tilde{\eflow}_\arc^-$ exists, then so does $\eflow_\arc^-$.
     \item  \label{lem: outflow: ExisLeqInverted} Every $\Tilde{\eflow}_\arc^-\in L_+(\hori)$  with $\Tilde{\eflow}_\arc^- = 0$ on ${\startint}\exit_\arc(-\infty))$ for $\exit_\arc(-\infty) := \lim_{t \to -\infty}\exit_\arc(t) $ has a corresponding inflow rate $\Tilde{\eflow}_\arc\in L_+(\hori)$. 
    \item \label{lem: outflow: ExisLeqInvertedCons} If $\eflow_\arc^-$ exists, then every $\Tilde{\eflow}_\arc^-\in L_+(\hori)$  with $\Tilde{\eflow}_\arc^- \leq \eflow_\arc^-$ has a corresponding inflow rate $\Tilde{\eflow}_\arc\in L_+(\hori)$ with $\Tilde{\eflow}_\arc \leq {\eflow}_\arc$.  
    \end{enumerate}

     \item  \label[thmpart]{lem: outflow: FlowCon} For any $v \in \GV$,  the edge outflow rates  $\eflow_\arc^-,\arc \in \delta^-(v)$ exist, if  $\eflow$ has a net node outflow rate $\inflow_v \in L_+(\hori)$ at $v$, i.e.~if \eqref{eq: FlowBalanceDerivative} holds w.r.t.~$\inflow_v$.  
    
    \item  \label[thmpart]{lem: outflow: In=OutIfD=0} If $\eflow_\arc^-$ exists, then $\eflow_\arc^-(t) = \eflow_\arc(t)$ for almost all $t$ with $\trav_\arc(t) = 0$. 
\end{thmparts}
\end{lemma}
\begin{proof}
\begin{structuredproof}
\proofitem{\ref{lem: outflow: =ell+Cons}} 
By definition  $\ell_{\wa,2}(\eflow_\arc)$ has to fulfill 
     \begin{align*}
         \int_{\mathfrak T} \ell_{\wa,2}(\eflow_\arc) \di \sigma \overset{(*)}{=} \int_{\arr_{\wa,2}^{-1}(\mathfrak T)}\eflow_\arc \di \sigma= \int_{\exit_\arc^{-1}(\mathfrak T)}\eflow_\arc \di \sigma \text{ for all }\mathfrak T \in \mathcal{B}(\hori). 
     \end{align*}
     Hence, the claimed equality follows by \Cref{lem: elluExistenceProperties} where we showed that a function fulfilling the equality $(*)$ is uniquely determined. This \namecref{lem: elluExistenceProperties} then also implies~\ref{lem: outflow: Exis} which in turn implies~\ref{lem: outflow: ExisLeq}. The statement in~\ref{lem: outflow: ExisLeqInverted} follows by \Cref{lem: 1to1:h-f} while~\ref{lem: outflow: ExisLeqInverted} together with \Cref{lem: ellOrderPreserving} imply~\ref{lem: outflow: ExisLeqInvertedCons}.   

    \proofitem{\ref{lem: outflow: FlowCon}} 
    We verify that condition~\eqref{eq: outflowExists} is satisfied. Hence, consider an arbitrary null set $\mathfrak T$ and $\arc \in \delta^-(v)$.  
    By  $\eflow$ having the net node outflow rate $\inflow_v \in L_+(\hori)$ at $v$ and $\eflow \in L_+(\hori)^\GA$ we get
    \begin{align*}
         0 \leq\int_{\exit_\arc^{-1}(\mathfrak T)}\eflow_\arc \di \sigma   \leq \sum_{\arc \in \delta^-(v)} \int_{\exit_\arc^{-1}(\mathfrak T)}\eflow_\arc \di \sigma  \overset{\eqref{eq: FlowBalanceDerivative}}{=} \sum_{\arc \in \delta^+(v)} \int_{\mathfrak T} \eflow_\arc \di \sigma- \int_{\mathfrak T} \inflow_v \di\sigma = 0,
    \end{align*}
    which shows the claim. 
 
    \proofitem{\ref{lem: outflow: In=OutIfD=0}} 
        Let $\mathfrak T \in \mathcal{B}(\hori)$ with $\mathfrak T \subseteq \{t \in \hori \mid \trav_\arc(t) = 0\}$ be arbitrary. 
        Then we have $\mathfrak T \subseteq \exit_\arc^{-1}(\mathfrak T) $. Moreover, 
        by~\ref{lem: outflow: =ell+Cons} and \Cref{lem: elluinj}, we have $\eflow_\arc = 0$ on $\exit_\arc^{-1}(\exit_\arc(\mathfrak T))\setminus\mathfrak T =\exit_\arc^{-1}(\mathfrak T)\setminus\mathfrak T$. 
        Hence, we arrive at    
        \begin{align*}
        \int_{\mathfrak T} \eflow_\arc^- \di\sigma =   \int_{\exit_\arc^{-1}(\mathfrak T)}\eflow_\arc \di\sigma  
        = \int_{\mathfrak T}\eflow_\arc \di\sigma + \int_{\exit_\arc^{-1}(\mathfrak T)\setminus \mathfrak T}\eflow_\arc \di\sigma  = \int_{\mathfrak T}\eflow_\arc \di\sigma  , 
    \end{align*}
    implying the claim. \qedhere
\end{structuredproof}
\end{proof}

We are now in the position to derive the promised statement that any  dynamic circulation can  be expressed in terms of zero-cycle inflow rates.

\begin{rsttheorem}{\ref{lem: ZeroCycleDecomposition}}
	Any  \aauto dynamic circulation $\eflow$ can be decomposed into zero-cycle inflow rates ${h}_c \in L_+(\hori), c \in \SimpCyc$ via $\eflow_\arc = \sum_{c \in \SimpCyc} \ell_{c,\arc}({h}_c) = \sum_{c \in \SimpCyc:\arc \in c}  {h}_c$ for all $\arc \in \GA$.
\end{rsttheorem}
\begin{proof} 
We start by observing that
 all (edge) outflow rates $\eflow_\arc^-,\arc \in \GA$ exist by \Cref{lem: outflow: FlowCon} and 
$\eflow$ satisfying flow conservation at all nodes. 
From this, \Cref{lem: FlowConEveryNode} together with \Cref{lem: outflow: In=OutIfD=0} imply
that $\eflow_\arc = \eflow_\arc^-,\arc \in \GA$. 
Moreover, using flow conservation at all nodes again, we get for any $v \in \GV$ and all $\mathfrak T \in \mathcal{B}(\hori)$ that
\begin{align*}
     0 = \sum_{\arc \in \delta^+(v)} \int_{\mathfrak T} \eflow_\arc \di \sigma -  \sum_{\arc \in \delta^-(v)} \int_{\mathfrak T} \eflow^-_\arc \di \sigma =  \sum_{\arc \in \delta^+(v)} \int_{\mathfrak T} \eflow_\arc \di \sigma -  \sum_{\arc \in \delta^-(v)} \int_{\mathfrak T} \eflow_\arc \di \sigma
\end{align*} which shows 
that  $ \sum_{\arc \in \delta^+(v)}  \eflow_\arc(t) =  \sum_{\arc \in \delta^-(v)}  \eflow_\arc(t) $ for almost all $t \in \hori$.

Let us now fix a nonnegative representative of $\eflow$ fulfilling the latter property as well as the property stated in \Cref{lem: FlowConEveryNode}  for all $t \in \hori$, that is, for all $t \in \hori$, the vector $\eflow(t) \in \R^\GA_+$ is a static flow fulfilling flow conservation at every node and $\eflow_\arc(t)>0 \implies \trav_\arc(t)= 0$ holds for all $\arc \in \GA$. 

Let us denote the set of simple cycles as $\mathcal{C}= \{c_1,\ldots,c_k\}$.
We define in the following recursively nonnegative measurable functions ${h}_{c_j}, j\in\{1,\ldots,k\}$ and $\eflow^j,j\in\{0,\ldots,k\}$ 
with $\eflow^j$ resembling the flow $\eflow^{j-1}$ from which we have subtracted the flow ${h}_{c_j}$ along $c_j$. 
In particular, all $\eflow^j$  fulfill flow conservation at every node and time $t$ and $\eflow^j_\arc(t)>0 \implies \trav_\arc(t)= 0, t \in \hori$ holds for all $\arc \in \GA$.

We start by setting $\eflow^0:= \eflow$ which clearly satisfies the properties required for $g^j$ as well. Now, assume that for some $j \in [k]$ we have already constructed $\eflow^{l},{h}_{c_l}$ with the stated properties for all $l < j$. 
We define ${h}_{c_j} := \min_{\arc \in c_j} \eflow^{j-1}_\arc \geq 0$ and set $\eflow^j_\arc := \eflow^{j-1}_\arc-{h}_{c_j}  \geq 0$ if $\arc \in c_j$ and $\eflow^j_\arc := \eflow^{j-1}_\arc \geq 0$ else. 
Clearly, both are measurable functions by $\eflow^{j-1}$ being likewise. 
Furthermore it is clear that $\eflow^j$ 
 fulfills  $\eflow^j_\arc(t)>0 \implies \trav_\arc(t)= 0, t \in \hori$
 by $\eflow^{j-1}$ fulfilling the latter. Similarly, $\eflow^j$ 
 fulfills also flow conservation at every node and time $t$ 
 as  we subtracted for all $t\in \hori$ from $\eflow^j(t)$ the value $\min_{\arc \in c_j} \eflow^{j-1}_\arc(t)$ along the cycle $c_j$.

We argue in the following that $h_c,c\in \mathcal{C}$ fulfill the properties stated in the \namecref{lem: ZeroCycleDecomposition}:

\begin{structuredproof}
    \proofitem{$\eflow_\arc(t) = \sum_{c \in \mathcal{C}:\arc \in c} {h}_c(t),t \in \hori,\arc\in\GA$}
Let $t \in \hori$ and $\arc \in \GA$ be arbitrary.  
By construction, $\eflow^k_\arc(t) = \eflow_\arc(t) - \sum_{c \in \mathcal{C}:\arc \in c}h_c(t)$. 
Hence, the claim follows by observing that $\eflow^k_\arc(t) = 0$: 
Since $\eflow^k(t)$ is a static flow fulfilling flow conservation at every node,  $\eflow^k_\arc(t) >0$ would imply (via the static flow decomposition theorem) that there has to exist a $j \leq k$ with $\min_{\arc' \in c_{j}} \eflow^k_{\arc'}(t) > 0$ and $\arc \in c_{j}$. This, however, implies that also $\min_{\arc' \in c_{j}} \eflow^{j}_{\arc'}(t) > 0$ which is not possible as $\eflow^{j}_{\arc'}(t) = \eflow^{j-1}_{\arc'}(t) - \min_{\hat{\arc} \in c_{j}} \eflow^{j-1}_{\hat{\arc} }(t)$ for all $\arc' \in c_j$ and hence $\eflow^{j}_{\arc'}(t) = 0$ for any $\arc' \in \argmin_{\hat{\arc} \in c_{j}} \eflow^{j-1}_{\hat{\arc} }(t)$. 

\proofitem{$h_c,c \in \mathcal{C}$ are zero-cycle inflow rates} 
This is an immediate consequence of the construction and $\eflow^j$ fulfilling  $\eflow^j_\arc(t)>0 \implies \trav_\arc(t)= 0$.

\proofitem{$\sum_{c \in \mathcal{C}} \ell_{c,\arc}({h}_c) = \sum_{c \in \mathcal{C}:\arc \in c}  {h}_c,\arc \in \GA $ in $L(\hori)$}
This is an immediate consequence of  $h_c,c \in \mathcal{C}$ being zero-cycle inflow rates,   \Cref{lem: FLowOnZeroTrav} 
and the fact that $\mathcal{C}$ is the set of simple cycles.  \qedhere
\end{structuredproof}
\end{proof}

As our second main result of this section, we show that any \auto $\source$,$\dest$-flow with positive net outflow rate at $s$ admits a flow carrying \stwalk . Moreover, 
 flow  can be send along this walk in such a way that the resulting flow has a net outflow at~$\source$ upper bounded by that of~$\eflow$ and a node balance at the destination that is lower bounded by that of~$\eflow$.

\begin{rsttheorem}{\ref{lem: ExistenceOfFlowCarryingWalk}}
	Let $\eflow \in L_+(\hori)^\GA$ be an \auto edge $\source$,$\dest$-flow with a positive node outflow rate ${\inflow}_s\in L_+(\hori) \setminus\{0\}$ at $\source$. 
Then, there exists 
an \stwalk{} $\wa \in \sdRoutes$ and a walk inflow rate $\wflow_\wa \in L_+(\hori)\setminus\{0\}$ with $\ell_{\wa}(\wflow_\wa) \leq \eflow$, $\wflow_\wa \leq {\inflow}_s$ as well as $\op_\dest\ell_\wa(\wflow_\wa) \geq \op_\dest \eflow$.
\end{rsttheorem}

\begin{proof}
	We determine such an \stwalk{} using the following algorithm (\Cref{algo: sdwalk}) which  constructs a (directed) tree~$\mathcal{T}$ of walks starting at $s$ with flow on them. This tree is iteratively constructed by 
	adding a set of new edges in each iteration $j \in \N$ to each leaf of the current tree $\mathcal{T}^{j-1}$. 
	We will denote by $\GV(G)$ and $\GV(\mathcal{T}^j)$ the nodes belonging to~$G$ and $\mathcal{T}^j$, respectively (similarly for edges).   Note that these trees may contain  many copies of the nodes/edges of the given graph~$G$. We say that a node $\tv \in V(\mathcal{T}^j)$ \emph{corresponds} to a node $v \in \GV(G)$ if it is one of the copies of $v$ and write $\pi(\tv)=v$ (and similarly for edges).  
	Furthermore, we denote by $\tarc(\tv)$ the unique edge in $\GA(\mathcal{T}^j)$ that enters $\tv \in \GV(\mathcal{T}^j)$
	and denote by $L^j$ the set of leafs of the tree $\mathcal{T}^j$.  
	
	\begin{algorithm}
		\caption{Find Flow Carrying \stwalk}\label{algo: sdwalk}
		\Input{\Aauto $\source$,$\dest$-flow  $\g\in L_+(\hori)^\GA$  with positive node outflow rate ${\inflow}_s\in L_+(\hori) \setminus\{0\}$ at $s$}
		\Output{A flow carrying \stwalk{} as described in \Cref{lem: ExistenceOfFlowCarryingWalk}}
		
		$\mathcal{T}^0 \leftarrow (\{\tilde{s}\},\emptyset)$ where $\pi(\tilde{s})= s$ 
		
		$\diffe^{0} \leftarrow \eflow$
		
		$\teflow_{\tarc(\Tilde{s})}^- \leftarrow {\inflow}_s$  \tcp*{Initialize inflow into root of $\mathcal{T}^0$ via artificial edge inflow $\teflow_{\tarc(\Tilde{s})}^-$} 
		
		\ForEach{$j \in \N$ }{ \label{algo: sdwalk: Iterations} 
			
			$L^{j-1} \leftarrow$ leafs of $\mathcal{T}^{j-1}$
			
			$\mathcal{T}^j \leftarrow \mathcal{T}^{j-1}$
			
			$\diffe^j \leftarrow \diffe^{j-1}$
			
			\fetry{$\tv \in L^{j-1}$ \label{algo: sdwalk: lineSetFor}}{
				
				Construct a finite set of outgoing edges $\eset_{\tv}$ with $\pi(\eset_{\tv}) \subseteq \delta^+(\pi(\tv))$  
				and corresponding functions 
				$\teflow_{\tarc} \in L_+(\hori)\setminus\{0\} ,\tarc \in \eset_{\tv}$ 
				such that 
				\begin{align}
					\int_{{\mathfrak T}}\sum_{\tarc \in \eset_{\tv}}\teflow_{\tarc}\di\leb &= \int_{\exit_{\pi(\tarc(\tv))}^{-1}({\mathfrak T})} \teflow_{\tarc(\tv)}\di\leb  
					&\text{ holds  for all }\mathfrak T \in \mathcal{B}(\hori)  \tag{\ref*{algo: sdwalk: lineSet}.1} \label{algo: sdwalk: lineSet1}\\
					\sum_{\tarc \in \eset_{\tv}: \pi(\tarc)=\arc} \teflow_{\tarc} &\leq \diffe^j_\arc &\text{ holds for all } \arc \in \GA(G) \tag{\ref*{algo: sdwalk: lineSet}.2} \label{algo: sdwalk: lineSet2} 
				\end{align}\label{algo: sdwalk: lineSet}	
				
				Add the edges from $\eset_{\tv}$ to $\mathcal{T}^j$
				
				$\diffe^j \leftarrow \diffe^{j} - (\sum_{\tarc \in \eset_{\tv}: \pi(\tarc)=\arc} \teflow_{\tarc})_{\arc \in \GA}$ \label{algo: sdwalk: newDiffe}	
				
			} 
			\catch{ 
				$\tv^* \leftarrow \tv$
				
				$k\leftarrow j$ \label{algo: sdwalk: defk}
				
				\KwRet{$\wa\coloneq(\pi(\tarc_1),\ldots,\pi(\tarc_{k-1}))$ where $\Tilde{\wa}\coloneq(\tarc_1,\ldots,\tarc_{k-1})$ is the \stwalk[\tilde{s}][\tv^*] in $\mathcal{T}^{k}$}}
		} 
	\end{algorithm}

	\begin{figure}[ht]
		\centering
		\BigPicture[1]{%
			\newcommand{\fadingpathe}[3][]{
				\draw[line width=1pt,line cap=round,dash pattern=on 0pt off 2pt,#1](#2) -- (#3);
				\begin{scope}
					\clip($(#2)!.1!(#3)$) rectangle ($(#2)!1!(#3)$);
					\draw[line width=2pt,white, path fading=west]($(#2)!-.01!(#3)$) -- ($(#2)!1.01!(#3)$);
				\end{scope}
			}
			\newcommand{\fadingpathw}[3][]{
				\draw[line width=1pt,line cap=round,dash pattern=on 0pt off 2pt,#1](#2) -- (#3);
				\begin{scope}
					\clip($(#2)!.1!(#3)$) rectangle ($(#2)!1.01!(#3)$);
					\draw[line width=2pt,white, path fading=east]($(#2)!-.01!(#3)$) -- ($(#2)!1.01!(#3)$);
				\end{scope}
			}
			\def\nodedist{2cm}
			
			\begin{tikzpicture}[node distance= \nodedist,
				Knoten/.style ={draw,circle, inner sep = 3pt, fill = black},
				Hilf/.style = {draw,white,circle, inner sep = 3pt, fill = white}]
				
				\node[draw, circle] (s) at (0,0) {$\tilde{s}$}; 
				\node[above of = s, node distance = 3 * \nodedist] (L0) {{\Large $L^0$}};
				
				\node[right of = s, node distance = 0.75*\nodedist] (h3) {};
				\node[Knoten, above of =  h3](v1)  {};
				\node[Knoten, below   of =  h3](v2)  {};
				\node[above of = v1,node distance = 2*\nodedist] (L1) {{\Large $L^1$}};
				\node[right of = L1,node distance = 0.5*\nodedist] {\ldots};
				
				\draw[edge] (s) to(v1);
				\draw[edge] (s) to (v2);
				
				\node[Hilf, above right of =  v1](v3)  {};
				\node[Hilf, below of =  v3,node distance = 0.3*\nodedist](v4)  {};
				\node[Hilf, below of =  v4,node distance = 0.32*\nodedist](v5)  {};
				\node[Hilf, below of =  v5,node distance = 0.33*\nodedist](v6)  {};
				\node[Hilf, below of =  v6,node distance = 0.36*\nodedist](v7)  {};
				\node[Hilf, below of =  v7,node distance = 0.5*\nodedist](v8)  {};
				\node[Hilf, below of =  v8,node distance = 0.66*\nodedist](v9)  {};
				\node[Hilf, below of =  v9,node distance = 0.31*\nodedist](v10)  {}; 
				\node[Hilf, below of =  v10,node distance = 0.31*\nodedist](v14)  {}; 
				
				\fadingpathe{v1}{v3}
				\fadingpathe{v1}{v4}
				\fadingpathe{v1}{v5}
				\fadingpathe{v1}{v6}
				\fadingpathe{v2}{v7}
				\fadingpathe{v2}{v8}
				\fadingpathe{v2}{v9}
				\fadingpathe{v2}{v10}
				\fadingpathe{v2}{v14}
				
				\node[draw,circle,node distance = 0pt, right of = s, node distance = 3.5*\nodedist] (tv) {$\tv$};
				\node[Knoten, above  of =  tv, node distance = 2* \nodedist](v11)  {};
				\node[Knoten, below   of =  tv, node distance = 2* \nodedist](v12)  {};  
				
				\fadingpathw[<-]{v11}{$(v4)+(\nodedist,.5\nodedist)$}
				\fadingpathw[<-]{v12}{$(v14)+(\nodedist,-.5\nodedist)$}
				
				\node[Knoten,left of =tv, node distance = 1.5*\nodedist] (h2) {};
				
				\path (v11) to node[midway] (z1) {\vdots}(tv);
				\path (tv) to node[midway] (z2){\vdots} (v12);
				
				\fadingpathe{h2}{$(h2)!.5!(z1)$}
				\fadingpathe{h2}{$(h2)!.5!(z2)$}
				
				\draw[edge] (h2) to 
				node[pos=0.5,below] {$\tarc(\tv)$}
				node[pos = 0.6,above = .1cm] {
					\begin{tikzpicture}
						\begin{scope}[-]
							\begin{axis}[xmin=0,xmax=6.5,ymax=1.5, ymin=0, samples=500,width=3.5cm,height=2.8cm,anchor=west,
								axis y line*=left, axis lines=left, xtick=\empty,ytick=\empty]
								\addplot[blue,  thick,domain=0:2*pi]  {cos(deg(x))} node[above right = 2pt,pos=.1]{$\teflow_{\tarc(\tv)}$};
								\addplot[blue,  thick,domain=0:2*pi]  {cos(deg(x-pi))};
							\end{axis}            
						\end{scope}  
					\end{tikzpicture}
				} (tv);
				
				\node[above of =  h2, node distance = 3*\nodedist]   {{\Large $L^{j-2}$}};
				\node[above of =  v11]   {{\Large $L^{j-1}$}};
				
				\node[right of = tv, node distance = 2*\nodedist ] (vj2) {};
				\node[Knoten,  above of = vj2, node distance = 1.5*\nodedist ] (vj1) {};
				\node[Knoten,  below of =vj2, node distance = 1.5*\nodedist ] (vj3) {};
				
				\node[above of =  vj2,node distance = 3*\nodedist]   {{\Large $L^j$}};
				
				\draw[edge] (tv) to 
				node[midway, above = 3cm]  {$\eset_{\tv}=\{\tarc_1,\tarc_2\}$} 
				node[midway,above] {$\tarc_1$}
				node[pos=0.7,below right = 0cm] {
					\begin{tikzpicture}
						\begin{scope}[-]
							\begin{axis}[xmin=0,xmax=6.5,ymax=1.5, ymin=0, samples=500,width=3.5cm,height=2.8cm,anchor=west,
								axis y line*=left, axis lines=left, xtick=\empty,ytick=\empty]
								\addplot[blue,  thick,domain=0:2*pi]  {cos(deg(x))} node[above right = -4pt,pos=.1]{$\teflow_{\tarc_1}$};
							\end{axis}            
						\end{scope}  
					\end{tikzpicture}
				} (vj1);

				\draw[edge] (tv) to 
				node[midway,below] {$\tarc_2$}
				node[pos=0.7,above right =0cm] {
					\begin{tikzpicture}
						\begin{scope}[-]
							\begin{axis}[xmin=0,xmax=6.5,ymax=1.5, ymin=0, samples=500,width=3.5cm,height=2.8cm,anchor=west,
								axis y line*=left, axis lines=left, xtick=\empty,ytick=\empty]
								\addplot[blue,  thick,domain=0:2*pi]  {cos(deg(x-pi))}node[above left = -4pt,pos=.4]{$\teflow_{\tarc_2}$};
							\end{axis}            
						\end{scope}  
					\end{tikzpicture}
				} (vj3);
				
		\end{tikzpicture}}
		\caption{Visualization of the tree and functions constructed in \Cref{algo: sdwalk}.}
		\label{fig:AlgoSDWalk}
	\end{figure} 
	
	In order to show that \Cref{algo: sdwalk} is correct, we show the following claim 
	via induction over the number of execution of the inner foreach-loop starting in \Cref{algo: sdwalk: lineSetFor}. 
	\begin{claim} \label{claim: algo: sdwalk: DiffeFlowCon} 
		The inner foreach-loop starting in \Cref{algo: sdwalk: lineSetFor}  retains the invariant that
		the node outflow rate of $\diffe^j$ at any node $v \in \GV(G)\setminus\{\dest\}$ equals the combined edge outflow rates $\teflow_{\tarc(\tv)} ^-$ of edges   entering the representatives of $v$ in $L^j$, i.e.
		\begin{align}\label{eq: claim: algo: sdwalk: DiffeFlowCon1}
			\op_v \diffe^j =   \sum_{\tv \in L^j: \pi(\tv) = v} \left((\teflow_{\tarc(\tv)}\cdot\leb)\circ \exit_{\pi(\tarc(\tv))}^{-1}\right)  = \sum_{\tv \in L^j: \pi(\tv) = v} \teflow_{\tarc(\tv)} ^-\cdot\leb.
		\end{align}
		Similarly, we have for the node balance at $\dest$:
		\begin{align}\label{eq: claim: algo: sdwalk: DiffeFlowCon2}
				\op_\dest \diffe^j =  \op_\dest \eflow + \sum_{\tv \in L^j: \pi(\tv) = \dest} \left((\teflow_{\tarc(\tv)}\cdot\leb)\circ \exit_{\pi(\tarc(\tv))}^{-1}\right) .
		\end{align}
	\end{claim}
	Before coming to the proof, 
	let us remark on two differences between \eqref{eq: claim: algo: sdwalk: DiffeFlowCon1} and \eqref{eq: claim: algo: sdwalk: DiffeFlowCon2}. 
	Firstly, we do not necessarily have the existence of edge outflow rates $\teflow_{\tarc(\tv)} ^-$ for $\pi(\tv) = \dest$ 
	since the edge outflow rates of $\g_\arc^-$ for $\arc \in \edgesTo{\dest}$ do not necessarily exist which in turn is possible as our model does not prevent for a point mass of flow to arrive at the destination. Secondly, we do not have a summand $\op_v\g$ in \eqref{eq: claim: algo: sdwalk: DiffeFlowCon1} as $\op_v\g = 0$ for all $v\neq \source,\dest$ and $\op_\source \g$ is accounted for by the artificial edge inflow $\teflow_{\tarc(\Tilde{s})}^-=  {\inflow}_s$.
	\begin{proofClaim}[Proof of \Cref{claim: algo: sdwalk: DiffeFlowCon}]
		We only show the case of $v \neq \dest$ as   the case of $v = \dest$ follows completely analogously.

		At the start of the algorithm, the invariant is trivially fulfilled by the definitions of  $\diffe^0 = \eflow$ and $\teflow_{\tarc(\Tilde{s})}^-=  {\inflow}_s$, the properties required for $\eflow$ and $L^0 = \{s\}$, $\GA(\mathcal{T}^0) = \emptyset$.  
		
		Hence, assume that the invariant holds before some iteration of the inner foreach-loop where $j$ and $\tv' \in L^{j-1}$ are the index and node, respectively, considered in that iteration. We denote by $\Delta$ the vector $\Delta^j$ before the iteration and by $\Delta'$ the vector $\Delta^j$ after the iteration. Analogously, we use the notation $L$ and $L'$ for the sets of leafs of $\mathcal{T}^j$ before and after the iteration.
		
		Let us start by remarking that $\g_\arc^-,\arc \in \edgesTo{v}$ exist by \Cref{lem: outflow: FlowCon}.  Subsequently, by $\teflow_{\tarc(\tv)} \leq \diffe_{\pi(\tarc(\tv))} \leq \g_{\pi(\tarc(\tv))}$, it follows by \Cref{lem: outflow: ExisLeq} of \Cref{lem: outflow: =ell+Cons} that $\teflow_{\tarc(\tv)}^-$ exists as well.             
		
		We now compute for any $v \in \GV(G)\setminus\{\dest\}$ the node net outflow of $\Delta'$ at this node (note that \eqref{algo: sdwalk: lineSet2} ensures that $\Delta'$ is nonnegative):
		\begin{align*}
			\op_v \diffe' 
			&= \sum_{\arc \in \delta^+(v)}  \diffe'_\arc \cdot \leb -  \sum_{\arc \in \delta^-(v)}  ( \diffe'_\arc \cdot \leb)\circ \exit_\arc^{-1} \\
			&\symoverset{4}{=} \sum_{\arc \in \delta^+(v)}  \bigl(\diffe_\arc - \sum_{\tarc \in \eset_{\tv'}: \pi(\tarc)=\arc} \teflow_{\tarc}\bigr)\cdot \leb -  \sum_{\arc \in \delta^-(v)}  \Bigl(\bigl(\diffe_\arc - \sum_{\tarc \in \eset_{\tv'}: \pi(\tarc)=\arc} \teflow_{\tarc}\bigr)\cdot \leb\Bigr)\circ \exit_\arc^{-1} \\
			&= \op_v \diffe - \sum_{\arc \in \delta^+(v)}\sum_{\tarc \in \eset_{\tv'}: \pi(\tarc)=\arc} \teflow_{\tarc}\cdot \leb + \sum_{\arc \in \delta^-(v)}\sum_{\tarc \in \eset_{\tv'}: \pi(\tarc)=\arc} (\teflow_{\tarc}\cdot \leb)\circ \exit_\arc^{-1} \\
			&= \op_v \diffe -\sum_{\tarc \in \eset_{\tv'}: \pi(\tarc) \in \delta^+(v)} \teflow_{\tarc}\cdot \leb + \sum_{\tarc \in \eset_{\tv'}: \pi(\tarc)\in \delta^-(v)} (\teflow_{\tarc}\cdot \leb)\circ \exit_\arc^{-1} \\
			&\symoverset{1}{=} \sum_{\tv \in L: \pi(\tv) = v} \teflow_{\tarc(\tv)} ^-\cdot\leb -\sum_{\tarc \in \eset_{\tv'}: \pi(\tarc) \in \delta^+(v)} \teflow_{\tarc}\cdot \leb + \sum_{\tarc \in \eset_{\tv'}: \pi(\tarc)\in \delta^-(v)} \teflow_{\tarc}^-\cdot\leb \\
			&=\sum_{\tv \in L\setminus\{\tv'\}: \pi(\tv) = v} \teflow_{\tarc(\tv)}^-\cdot\leb + \CharF[\{\pi(\tv')\}](v)\cdot\teflow_{\tarc(\tv)} ^-\cdot\leb 
			+ \sum_{\tv:\tarc(\tv)\in \eset_{\tv'}\land \pi(\tv) = v} \teflow_{\tarc(\tv)}^-\cdot\leb
			 -\hspace{-0.5cm}\sum_{\tarc \in \eset_{\tv'}: \pi(\tarc) \in \delta^+(v)} \teflow_{\tarc}\cdot \leb  \\ 
			&\symoverset{2}{=} \sum_{\tv \in L': \pi(\tv) = v} \teflow_{\tarc(\tv)}^-\cdot\leb + \CharF[\{\pi(\tv')\}](v)\cdot\teflow_{\tarc(\tv)} ^-\cdot\leb  -\sum_{\tarc \in \eset_{\tv'}: \pi(\tarc) \in \delta^+(v)} \teflow_{\tarc}\cdot \leb    \\
			&=\sum_{\tv \in L': \pi(\tv) = v}   \teflow_{\tarc(\tv)} ^-\cdot\leb.  
		\end{align*}    
		Here, \refsym{4} holds as $\diffe'$ is obtained from $\diffe$ via \Cref{algo: sdwalk: newDiffe} of \Cref{algo: sdwalk}. The equality indicated by~\refsym{1} holds since, for the first term, the invariant is true for $\diffe$. For the third term we used  the definition and existence of the edge outflow rates. The one indicated with~\refsym{2} is true since $L' = (L \setminus \{\tv'\}) \cup \set{\tv | \tarc(\tv) \in \eset_{\tv'}}$. The last equality holds by the fulfillment of~\eqref{algo: sdwalk: lineSet1}.   
	\end{proofClaim}
	
	Next, we show that if \Cref{algo: sdwalk} terminates, then $\tv^*$ corresponds to $\dest$, i.e.~$\pi(\tv^*) = \dest$.
	In particular, the walk $\wa$ returned by \Cref{algo: sdwalk} is an \stwalk. 
	\begin{claim}\label{claim: algo: sdwalk: SetsExist}
		In any execution of \Cref{algo: sdwalk: lineSet},  
		a set  $\eset_{\tv}$ with corresponding $\teflow_{\tarc},\tarc\in \eset_{\tv}$  exists 
		if
		\begin{align}\label{eq: claim: algo: sdwalk: SetsExist}
			\sum_{\arc \in \edgesFrom{\pi(\tv)}}\diffe^j_\arc \cdot \leb \geq (\tilde{\g}_{\tarc(\tv)} \cdot \leb)\circ \exit_{\pi(\tarc(\tv))}^{-1}.
		\end{align}
	     Moreover, if $\pi(\tv)\neq \dest$, $\eset_{\tv}$ with corresponding $\teflow_{\tarc},\tarc\in \eset_{\tv}$  exists and hence,   if \Cref{algo: sdwalk} terminates, $\pi(\tv^*) = \dest$.  
	\end{claim}
	\begin{proofClaim}
		Consider any execution of \Cref{algo: sdwalk: lineSet} with corresponding $\tv'$. We use the same notation $\diffe$, $L$ as in the proof of \Cref{claim: algo: sdwalk: DiffeFlowCon} for the flow and leaf set before the current iteration. 
		
		By \eqref{eq: claim: algo: sdwalk: SetsExist} (for $\diffe^j=\diffe$), it follows that $(\tilde{\g}_{\tarc(\tv)} \cdot \leb)\circ \exit_{\pi(\tarc(\tv))}^{-1}$ is absolutely continuous and, hence by \Cref{lem: outflow: Exis}, $\tilde{\g}_{\tarc(\tv)}^-$ exists and fulfills  $\tilde{\g}_{\tarc(\tv)}^- \leq \sum_{\arc \in \edgesFrom{\pi(\tv')}}\diffe_\arc$. This together with 
		$\diffe \in L_+(\hori)^\GA$ being nonnegative (by the fulfillment of \eqref{algo: sdwalk: lineSet2}) implies the existence  of  a set $\mathcal{E} \subseteq \delta^+(\pi(\tv'))$ and functions $\teflow_\arc \in L_+(\hori)\setminus\{0\},\arc \in \mathcal{E}$ with $\teflow_{\arc} \leq \diffe_{\arc}, \arc \in \mathcal{E}$ as well as $\sum_{\arc \in \mathcal{E}} \tg_\arc =\tilde{\g}_{\tarc(\tv)}^-$. In particular, we have  
		\begin{align*}
		 \int_{{\startint}t]} 	\sum_{\arc \in \mathcal{E}}  \teflow_\arc \di \leb =  \int_{{\startint}t]}   \teflow_{\tarc(\tv')} ^-\di\leb = \int_{\exit_{\pi(\tarc(\tv))}^{-1}({\startint}t])}   \teflow_{\tarc(\tv')}\di\leb \text{ for all }t \in \hori.
		\end{align*}
		Thus, choosing a set $\eset_{\tv'}$ which corresponds  one-to-one via $\pi$ to  $\mathcal{E}$ and setting $\teflow_{\tarc} \coloneq \teflow_{\pi(\tarc)},\tarc \in \eset_{\tv'}$ shows the first part of the claim.

		Now in case of $\pi(\tv') \neq \dest$, we get by \Cref{claim: algo: sdwalk: DiffeFlowCon} applied to $\pi(\tv')$ that: 
		\begin{align*}
			\sum_{\arc \in \edgesFrom{\pi(\tv')}}\diffe_\arc \cdot \leb \geq   \op_{\tv'} \diffe = 
			 \sum_{\tv \in L: \pi(\tv) = \pi(\tv')} \left((\teflow_{\tarc(\tv)}\cdot\leb)\circ \exit_{\pi(\tarc(\tv))}^{-1}\right) \geq (\tilde{\g}_{\tarc(\tv')} \cdot \leb)\circ \exit_{\pi(\tarc(\tv'))}^{-1}. 
		\end{align*} 
		Hence, the first part of the claim is applicable and finishes the proof. 
	\end{proofClaim}
	
	The statement of the \namecref{lem: ExistenceOfFlowCarryingWalk} is now a consequence of the following \namecref{claim: Algosdwalk}: 
	\begin{claim}\label{claim: Algosdwalk}
		The following statements are valid: 
		\begin{thmparts}[label= \roman*)]
			\item \Cref{algo: sdwalk} terminates after at most $\lfloor\norm{\eflow}/\norm{{\inflow}_s}\rfloor$ many iterations. \label[thmpart]{claim: Algosdwalk: termi}
			\item   The walk $\wa$ returned by    \Cref{algo: sdwalk}  is an \stwalk{} and there exists $\wflow_\wa \in L_+(\hori)\setminus\{0\}$  with $\ell_{\wa}(\wflow_\wa) \leq \eflow$ and $\wflow_\wa \leq {\inflow}_s$ as well as $\op_\dest\ell_\wa(\wflow_\wa) \geq \op_\dest \eflow$. \label[thmpart]{claim: Algosdwalk: corr}
		\end{thmparts}
		
	\end{claim}
	\begin{proofClaim}
		\begin{structuredproof}
			\proofitem{\ref{claim: Algosdwalk: termi}} 
			We first show the following subclaim via induction: 
			\begin{subclaim}
				For $j=0$, and, every  $j \in \N$ with the iteration in \Cref{algo: sdwalk: Iterations} being completed, we have $\norm{\sum_{\tv \in L^j} \teflow_{\tarc(\tv)}}= \norm{{\inflow_s}}$.
			\end{subclaim}
			\begin{proofClaim}
				The base case of $j = 0$ is trivial as by definition  $\teflow_{\tarc(\Tilde{s})} ={\inflow}_s$.  Hence, let $j \geq 1$ and assume the claim holds for all $0\leq j'<j$. 
				We first note that for any flow $f_\arc \in L_+(\hori)$ on an $\arc \in \GA$ that admits a corresponding edge outflow rate $f_\arc^-$, we have $\norm{f_\arc} = \norm{f_\arc^-}$.  
				Then, by the induction hypothesis, the fulfillment of~\eqref{algo: sdwalk: lineSet1} and using the equality $\norm{\teflow_{\tarc}} = \norm{\teflow_{\tarc}^-},\tarc \in \GA(\mathcal{T}^j)$ we get 
				\begin{align*}
					\norm{{\inflow}_s} &= \norm{\sum_{\tv \in L^{j-1}} \teflow_{\tarc(\tv)}} =  
					\sum_{\tv \in L^{j-1}} \norm{\teflow_{\tarc(\tv)} } = 
					\sum_{\tv \in L^{j-1}} \norm{\teflow_{\tarc(\tv)}^-} \\
					&= \sum_{\tv \in L^{j-1}} \norm{\sum_{\tarc \in \eset_{\tv}}  \teflow_{\tarc}} = 
					\sum_{\tv \in L^{j-1}} \sum_{\tarc \in \eset_{\tv}}  \norm{\teflow_{\tarc}} = 
					\sum_{\tv \in L^j}\norm{\teflow_{\tarc(\tv)}} = \norm{\sum_{\tv \in L^j} \teflow_{\tarc(\tv)}}
				\end{align*}
				Note that we can interchange any sum and norm since all equivalence classes are nonnegative almost everywhere. 
			\end{proofClaim}
			From this subclaim,~\ref{claim: Algosdwalk: termi} follows directly since we have $0 \leq \norm{\diffe^j} = \norm{\eflow}- \sum_{j'\in[j]} \norm{\sum_{\tv \in L^{j'}} \teflow_{\tarc(\tv)}}$ for all~$j$ with the outer foreach-loop being completed, where we again used that all appearing equivalence classes are nonnegative almost everywhere. 
			
			\proofitem{\ref{claim: Algosdwalk: corr}}
			Let $\wa\coloneq(\arc_1,\ldots,\arc_{k-1})=(\pi(\tarc_1),\ldots,\pi(\tarc_{k-1}))$ be the walk returned by the algorithm resulting from the walk $\Tilde{\wa}\eqcolon (\tarc_1,\ldots,\tarc_{k-1})$ in $\mathcal{T}^k$ and denote by $\tarc_l \eqqcolon (\tv^{l-1},\tv^{l})$ for all $l \in [k-1]$. Note that  $\tarc_l = \tarc(\tv^l)$ and $\tv^{k-1}=\tv^*$.  
			From \Cref{claim: algo: sdwalk: SetsExist}, we immediately know that $\wa$ is an \stwalk.  
			
			The second part of~\ref{claim: Algosdwalk: corr} will follow almost immediately from the following subclaim.  
			Here, the idea is the following: We first construct an edge inflow into the last edge $\arc_{k-1}=(\pi(\tv^{k-2}),\dest)$ of the walk $\wa$, which is upper bounded by $\teflow_{\tarc_{k-1}}$ and which induces a node balance at $\dest$ that is lower bounded by $\op_\dest \g$. 
			This edge out flow is then propagated iteratively backwards along the walk $\wa$. 
			\begin{subclaim}
				There exists for every $j  \in [k-1]$ an equivalence class $\wflow_\wa^j \in \edom{\wa_{\geq j}}\setminus\{0\}$ with  $\ell_{\wa_{\geq j},j'}(\wflow_\wa^j) \leq \teflow_{\tarc_{j+j'-1}}, j' \in  [k-j]$ and  $\op_\dest\ell_{\wa_{\geq j}}(\wflow_\wa^j) \geq \op_\dest \eflow$. 
			\end{subclaim}
			\begin{proofClaim}
				We prove this subclaim via downwards induction, i.e.\ starting with $j=k-1$:
				\begin{proofbyinduction}
					\basecase{$j = k-1$} 
					Define the two measures 
					\[\mu^\diffe \coloneq \sum_{\arc \in \edgesFrom{\pi(\tv^*)}}\diffe^{k-1}_\arc \cdot \leb \quad\text{and}\quad \mu^{\teflow}\coloneq (\tilde{\g}_{\tarc(\tv^*)} \cdot \leb)\circ \exit_{\pi(\tarc(\tv^*))}^{-1},\]
					representing the outflow from~$\pi(\tv^*)$ under $\diffe^{k-1}$ in the original and the inflow into~$\tv^*$ under $\tilde\g$ in the constructed tree, respectively.
					
					Since we have $\tv^{k-1} = \tv^*$,  a set $\eset_{\tv^*}$ as required in \Cref{algo: sdwalk: lineSet} does not exist. 
					Thus, by \Cref{claim: algo: sdwalk: SetsExist}, 
					we have $\mu^\diffe \ngeq  \mu^{\teflow}$. 
					In particular, the difference $\nu\coloneq \mu^\diffe - \mu^{\teflow}$ is a signed finite measure and has a Jordan-Hahn decomposition $\nu = \nu^+-\nu^-$ with $\nu^- \in \Meas_+(\R)\setminus\{0\}$. 
					Then, $\nu^-$ is absolutely continuous \wrt $\mu^{\teflow}$ and admits a 
					Radon-Nikodym derivative (cf.~\cite[Theorem 3.2.2]{Bogachev2007I}). Let 
					  $f$ be a representative of the  latter 
					 with the property that $f(t) \geq 0$ for all $t \in \hori$. Note that this is possible as 
					 the Radon-Nikodym derivative has to be nonnegative $\mu^{\teflow}$-almost everywhere since  $\nu^-$ and $\mu^{\teflow}$ are both nonnegative.

					Define $h_\wa^{k-1} \coloneqq (f \circ \exit_{\pi(\tarc(\tv^*))}) \cdot \tilde{\g}_{\tarc(\tv^*)}$ and note that the latter is measurable and nonnegative as all appearing functions are likewise. 
					Next,  observe the following for arbitrary ${\mathfrak T} \in \mathcal{B}(\hori)$: 
					\begin{equation}\label{eq: subclaim: claim: Algosdwalk: termi}
						\begin{aligned}
							\left((h_\wa^{k-1} \cdot\leb)\circ \exit_{\pi(\tarc(\tv^*))}^{-1}\right)({\mathfrak T}) &\eqperdef 
							\int_{\exit_{\pi(\tarc(\tv^*))}^{-1}({\mathfrak T})} (f\circ \exit_{\pi(\tarc(\tv^*))}) \cdot \tilde{\g}_{\tarc(\tv^*)} \di\leb \\
							&\symoverset{1}{=} \int_{\exit_{\pi(\tarc(\tv^*))}^{-1}({\mathfrak T})} f\circ \exit_{\pi(\tarc(\tv^*))}  \di\big( \tilde{\g}_{\tarc(\tv^*)} \cdot \leb\big) \\
							&\symoverset{2}{=} \int_{{\mathfrak T}} f  \di\Big(\big( \tilde{\g}_{\tarc(\tv^*)} \cdot \leb\big) \circ \exit_{\pi(\tarc(\tv^*))}^{-1}\Big)\\
							&\symoverset{3}{=}\nu^-({\mathfrak T})
						\end{aligned}
					\end{equation}
					where  \refsym{1} follows immediately by definition of $ \teflow_{\tarc(\tv^*)}\cdot \leb$, \refsym{2} by the change of variables formula (\cite[Theorem 3.6.1]{Bogachev2007II}) and \refsym{3} by $f$ being the Radon-Nikodym derivative of $\nu^-$ w.r.t.~$\mu^{\teflow}$. In particular, we get $h_\wa^{k-1} \neq 0$ as $\nu^- \neq 0$. Thus, by the above argued nonnegativity, we have $h_\wa^{k-1}(t) \in L_+(\hori)\setminus\{0\}$.

					Since the equalities in \eqref{eq: subclaim: claim: Algosdwalk: termi} hold for all $\mathfrak T \in \mathcal{B}(\hori)$, we have $(h_\wa^{k-1} \cdot\leb)\circ \exit_{\pi(\tarc(\tv^*))}^{-1} = \nu^-$. 	Moreover, $\ell_{\wa_{\geq k-1}}(\wflow_\wa^{k-1})$ exists and is given by
					$\ell_{\wa_{\geq k-1},\wa[k-1]}(\wflow_\wa^{k-1})= 	\wflow_\wa^{k-1}$ and $\ell_{\wa_{\geq k-1},\arc}(\wflow_\wa^{k-1})= 0$ else. 
					This lets us deduce further:
					\begin{align*}
						\op_\dest\ell_{\wa_{\geq k-1}}(\wflow_\wa^{k-1}) &\symoverset{3}{=} - ((\wflow_\wa^{k-1}  \cdot \leb)\circ \exit_{\pi(\tarc(\tv^*))}^{-1} \overset{\eqref{eq: subclaim: claim: Algosdwalk: termi}}{=} \nu^-\\
						&\eqperdef \left( \sum_{\arc \in \edgesFrom{\pi(\tv^*)}}\diffe^{k-1}_\arc \cdot \leb  - (\tilde{\g}_{\tarc(\tv^*)} \cdot \leb)\circ \exit_{\pi(\tarc(\tv^*))}^{-1} \right)^-  \\
						&\geq   \left(  \op_\dest (  \diffe^k )   -(\tilde{\g}_{\tarc(\tv^*)} \cdot \leb)\circ \exit_{\pi(\tarc(\tv^*))}^{-1} \right)^-  \\
						&\symoverset{1}{=}  \left(   \op_\dest \eflow  + \sum_{\tv \in L^k: \pi(\tv) = \dest} (\teflow_{\tarc(\tv)}\cdot\leb)\circ \exit_{\pi(\tarc(\tv))}^{-1}  -(\tilde{\g}_{\tarc(\tv^*)} \cdot \leb)\circ \exit_{\pi(\tarc(\tv^*))}^{-1} \right)^-\\
						&\symoverset{2}{\geq} \big(   \op_\dest  \eflow \big)^- \symoverset{3}{=} \op_\dest  \eflow 
					\end{align*}
					where the equality indicated by \refsym{3} holds by \Cref{lem: flowcon} and our assumption that~$G$ is loop-free, \refsym{1} by \Cref{claim: algo: sdwalk: DiffeFlowCon} and the inequality \refsym{2} by $\tarc_{k-1} = \tarc(\tv^{k-1}) = \tarc(\tv^*)$ and $\tv^* \in L^k$ with $\pi(\tv^*) = \dest$ by \Cref{claim: algo: sdwalk: SetsExist}. 
					Finally, \refsym{3} holds as $\op_\dest\g$ is nonpositive as $\g$ is an $\source$,$\dest$-flow. 
					
					\inductionstep{$j<k-1$} Assume that the claim is fulfilled for $j+1$, that is, there exists $h_\wa^{j+1}\in\edom{\wa_{\geq j+1}}\setminus\{0\}$ with 
					 $\ell_{\wa_{\geq (j+1)},j'}(\wflow_\wa^{j+1}) \leq \teflow_{\tarc_{(j+1)+j'-1}}, j' \in  [k-(j+1)]$ and  $\op_\dest\ell_{\wa_{\geq (j+1)}}(\wflow_\wa^{j+1}) \geq \op_\dest \eflow$. 
		
					By the validity of \eqref{algo: sdwalk: lineSet1}, we have for all $\mathfrak T \in \mathcal{B}(\hori)$: 
					\begin{align*}
						 \int_{\mathfrak T}h_\wa^{j+1}\di\leb &\symoverset{2}{\leq} \int_{\mathfrak T}\teflow_{\tarc_{j+1}}\di\leb
						\leq \int_{\mathfrak T}\sum_{\tarc \in \eset_{\tv^j}}\teflow_{\tarc}\di\leb\\ \overset{\eqref{algo: sdwalk: lineSet1}}&{=} \int_{\exit_{\pi(\tarc(\tv^j))}^{-1}(\mathfrak T)} \teflow_{\tarc(\tv^j)}\di\leb \symoverset{1}{=} 
						 \int_{\mathfrak T} \teflow_{\tarc(\tv^j)}^-\di\leb
					\end{align*}
					where \refsym{2} is valid as  $h_\wa^{j+1} = \ell_{\wa_{\geq (j+1)},1}(\wflow_\wa^{j+1}) \leq \teflow_{\tarc_{j+1}}$ by induction hypothesis. 
					 The existence of the edge outflow rate $\teflow_{\tarc(\tv^j)}^-$ in \refsym{1} follows by the validity  of \eqref{algo: sdwalk: lineSet1}. 
					The above inequalities show that we have $h_\wa^{j+1} \leq  \teflow_{\tarc(\tv^j)}^-$. 
					 By \Cref{lem: outflow: ExisLeqInvertedCons} of \Cref{lem: outflow: =ell+Cons}, there exists (\wrt edge $\pi(\tarc_{j})$) an inflow rate $\wflow_\wa^j  \in L_+(\hori)\setminus\{0\}$  
						with edge outflow rate $(\wflow_\wa^j)^-= h_\wa^{j+1}$ and 
					$\wflow_\wa^j \leq \teflow_{\tarc(\tv^j)}$.
					 Furthermore,  \Cref{lem: elluPropagation}  implies that $\ell_{\wa_{\geq j},j'}(\wflow_\wa^j) = \ell_{\wa_{\geq j+1},j'-1}(\wflow_\wa^{j+1})$ for all $j' \in \{2,\ldots,k-j\}$ since 
					$\ell_{\wa_{\geq j},j'}(\wflow_\wa^j) = \ell_{\wa_{\geq j+1},j'-1}(\ell_{\wa_{\geq j},2}(\wflow_\wa^j))$ and
					$\ell_{\wa_{\geq j},2}(\wflow_\wa^j) = (\wflow_\wa^j)^- = \wflow_\wa^{j+1}$. 
					Hence, the latter together with the induction hypothesis for $j+1$ shows that $\wflow_\wa^j$ is as required. \qedhere
				\end{proofbyinduction}
			\end{proofClaim}
			Now~\ref{claim: Algosdwalk: corr} follows by setting $\wflow_\wa \coloneq \wflow_\wa^1$ and observing that by the above subclaim we have 
			$\wflow_\wa^1  = \ell_{\wa,1}(\wflow_\wa^1) \leq \teflow_{\tarc_1}\leq \teflow^-_{\tarc(\Tilde{s})} = \inflow_s$. 
			Furthermore, for any $\arc \in \GA$ the chain of inequalities $\ell_{\wa,\arc}(\wflow_\wa) \leq \sum_{\tarc \in \GA(\mathcal{T}^k): \pi(\tarc)=\arc} \teflow_{\tarc} \leq \eflow_\arc$ is valid:  The last inequality holds as $\diffe^k \in L_+(\hori)$ by~\eqref{algo: sdwalk: lineSet2}. 
			The first inequality holds as $\ell_{\wa,j}(\wflow_\wa) \leq \teflow_{\tarc_j}, j \in [k]$ by the above subclaim and $\tarc_j \neq \tarc_{j'},j\neq j'$ since $\Tilde{\wa}$ is a walk in $\mathcal{T}^k$ which is a directed tree. \qedhere
		\end{structuredproof}
	\end{proofClaim}
	
	Since the statement of the \namecref{lem: ExistenceOfFlowCarryingWalk} is exactly \Cref{claim: Algosdwalk: corr}, the proof is finished. 
\end{proof}

\section{Conclusions and Open Problems}

We derived a decomposition theorem for integrable dynamic edge $\source$,$\dest$-flows
stating that any such edge flow can be decomposed into nonnegative linear combinations of \stwalk-inflows and circulations. 
We also gave a characterization of $\source$,$\dest$-edge flows that admit a decomposition purely into \stwalk-inflows without circulations. 
To prove these results, we developed the framework of \auto network loadings which, 
we believe, provides a useful
concept for understanding dynamic flows more broadly, beyond the scope of this paper.  
For this type of network loading, we 
derived several structural results: from existence characterizations and topological properties (\Cref{lem: elluExistenceProperties}), to existence guarantees for a class of associated optimization problems (\Cref{thm: ExistenceOptSol}), and to sufficient conditions for flows to be decomposable into zero-cycle inflow rates or to admit a flow-carrying walk (\Cref{lem: ZeroCycleDecomposition,lem: ExistenceOfFlowCarryingWalk}).

Two closely related problems that we did not address in this paper are related to $\source$,$\dest$-paths, i.e.~cycle-free \stwalk s: Determining which dynamic edge flows can be decomposed into only $\source$,$\dest$-path inflows and, in case of existence, how to compute such a decomposition. 
Regarding the former, one can deduce by the   equivalence of  edge- and walk-based equilibria (\Cref{thm:def-equilibrium-equivalent}) that edge-based equilibria always admit (only) path-based decompositions (assuming that particles minimize travel times and travel times are strictly positive): \Cref{lemma:EquivalenceEquilibriumDefinitionsGivenFlows} states that any decomposition of an edge-based dynamic equilibrium is a walk-based  equilibrium and, hence, flow can only be injected into paths for such a decomposition (and at least one such decomposition must exist by our decomposition theorem). 
Regarding the computation of a path-based decomposition, we remark that 
 the output of the decomposition algorithm presented in this paper heavily depends on the order of walks
 that are chosen in the main flow-reduction step: There are examples in which the edge flow admits a path-based decomposition but 
 a wrong order of paths chosen in the main flow-reduction step requires the algorithm to send flow into a walk including proper cycles.

We further believe that the dynamic flow decomposition results can be the
basis for a better understanding of related infinite dimensional 
optimization problems, such as
the problem of computing a system optimal traffic assignment (minimizing
total travel time) under 
fixed inflow rates and load-dependent travel times.
This quite fundamental problem is not understood at all (except for flow-independent travel times, see the related work), not even
for the Vickrey point queue model.

Let us finally mention that we did not elaborate on the computational complexity
of computing dynamic flow decompositions.
However, even much simpler questions like the computational complexity of the network
loading problem for the well-studied Vickrey queueing model is -- to the best of our knowledge -- not resolved so far (see the open problems raised by Martin Skutella in \cite[Section~4.6]{DagstuhlSeminar22}).

\section*{Acknowledgements}
\hfill\\
This research has been funded by the Deutsche Forschungsgemeinschaft
(DFG) in the project 543678993 (Aggregative gemischt-ganzzahlige
Gleichgewichtsprobleme: Existenz, Approximation und Algorithmen).
We acknowledge the support of the DFG.

\bibliographystyle{plain}
\bibliography{master-bib}

\appendix

\section{Properties of Locally Absolutely Continuous Functions}

In this section, we show several insights into locally absolutely continuous functions $\func:\hori\to\R$, i.e.~functions that are  absolutely continuous on every closed interval $[a,b]\subseteq \hori$ in the sense of \cite[Definition 5.3.1]{Bogachev2007I}.
These results follow almost immediately by their corresponding analogues for absolutely continuous functions on an interval.  
 \Cref{lem: PropAbsCon:ImageMeas} states that the image of a measurable set  under $\func$ is measurable as well. 
\Cref{lem: PropAbsCon:Lus} states that 
$\func$ fulfills Lusin's property, i.e.~takes null sets to null sets. 
\Cref{lem: PropAbsCon:Der} states that
the derivative of 
$\func$ exists almost everywhere. 
\Cref{lem: PropAbsCon:DerNonDec}
shows that for $\func$ non-decreasing, its derivative is larger or equal to $0$ in case of existence. 
\Cref{lem: PropAbsCon:Est} establishes an estimate for 
the measure of the image $\func(\mathfrak T)$ for any measurable set $\mathfrak T \in \mathcal{B}(\hori)$. 
Finally, \Cref{lem: PropAbsCon:Conca} states that the concatenation $\func\circ \hat{\func}$ of two locally absolutely continuous functions $\func,\hat{\func}$ is again locally absolutely continuous, provided that $\hat{\func}$ is non-decreasing. 
\begin{lemma}\label{lem: PropAbsCon} 
    For an arbitrary locally absolutely continuous function $\func:\hori\to\R$, the following statements are valid: 
\begin{thmparts} 
    \item\label[thmpart]{lem: PropAbsCon:ImageMeas} The image $\func(\mathfrak T)$ of a measurable set $\mathfrak T \in \mathcal{B}(\hori)$ under $\func$ is again (Borel-)measurable.
    \item  \label[thmpart]{lem: PropAbsCon:Lus} $\func$ fulfills Lusin's property, i.e.~$\sigma(\func(\mathfrak T)) = 0$ for all null sets $\mathfrak T$.
    \item\label[thmpart]{lem: PropAbsCon:Der} The derivative $\func'$ exists almost everywhere, any measurable extension $\Tilde{\func}'$ of $\func'$ to the whole real line is locally integrable and fulfills for all $a<b \in \R$  the equality 
    \begin{align*}
       \func(b) - \func(a) =  \int_{[a,b]}{\tilde{\func}'}\di\sigma . 
    \end{align*}
    In particular, this allows us to write from now on simply  $\int_{[a,b]}\func'\di\sigma$ instead of the right hand side. 
    \item \label[thmpart]{lem: PropAbsCon:DerNonDec}  If $\func$ is non-decreasing, then $\func'(t)\geq 0$ holds if $\func'(t)$ exists. 
    \item \label[thmpart]{lem: PropAbsCon:Est} The following estimate holds for all $\mathfrak T \in \mathcal{B}(\hori)$: 
    \begin{align*}
        \sigma(\func(\mathfrak T))) \leq \int_{\mathfrak T}\abs{\func'}\di\sigma.
    \end{align*}
    \item \label[thmpart]{lem: PropAbsCon:Conca}   If $\hat{\func}:\hori\to \hori$ is a non-decreasing locally absolutely continuous function, then $\func \circ \hat{\func}$ is also locally absolutely continuous. 
\end{thmparts}
     
\end{lemma}
\begin{proof}
All of the arguments in the following exploit the fact that $\func$ is an absolutely continuous function on any interval of the form $[-n,n]$ for arbitrary $n \in \N$. 
    \begin{structuredproof}
    \proofitem{\ref{lem: PropAbsCon:ImageMeas}} Let $\mathfrak T \in \mathcal{B}(\hori)$ be arbitrary. For every $n \in \N$, 
    the set $\func(\mathfrak T \cap [-n,n])$ is measurable as the image of a measurable set under an absolutely continuous function, cf.~\cite[Theorem 9.9.3 + Exercise 5.8.49]{Bogachev2007II}. Hence, the countable union $\bigcup_{n \in \N} \func(\mathfrak T \cap [-n,n]) = \func(\mathfrak T)$ is also measurable.  
    
     \proofitem{\ref{lem: PropAbsCon:Lus}} Let $\mathfrak T \subseteq \hori$ be a null set. 
     Then we have 
     \begin{align*}
         \sigma(\func(\mathfrak T)) = \lim_{n \to \infty} \sigma(\func(\mathfrak T \cap [-n,n])) = 0
     \end{align*}
     where the first equality holds as for any sequence of measurable sets $\mathfrak T_n, n\in \N$ with $\mathfrak T_n \subseteq\mathfrak T_{n+1},n\in \N$, the equality $\sigma(\bigcup_{n\in \N} \mathfrak T_n) = \lim_{n \to \infty} \sigma(\mathfrak T_n)$ is  true.  The last equality holds by $\func$ fulfilling Lusin's property on every finite interval (\cite[Exercise 5.8.49]{Bogachev2007I}). 
        \proofitem{\ref{lem: PropAbsCon:Der}}  Consider the set $\mathfrak D \in \mathcal{B}(\hori)$ where the derivative of $\func$ does not exist. Then, we have 
        \begin{align*}
            \sigma(\mathfrak D) = \lim_{n\to \infty} \sigma(\mathfrak D \cap [-n,n]) = 0 
        \end{align*}
        where $\sigma(\mathfrak D \cap [-n,n]) = 0$ as for every interval $[-n,n] \subseteq \hori, n \in \N$, the derivative exists almost everywhere by \cite[Theorem 5.3.6]{Bogachev2007I}. Moreover, \cite[Theorem 5.3.6]{Bogachev2007I} shows further that the derivative is integrable on $[a,b]$ for every $a< b  \in \R$ and that it fulfills the stated equality. 
        \proofitem{\ref{lem: PropAbsCon:DerNonDec}} Let $t \in \hori$ with $\func'(t)$ existing, that is, 
        \begin{align*}
            \func'(t)= \lim_{n \to \infty} \frac{\func(t_n)-\func(t)}{t_n-t} 
        \end{align*}
        exists for any sequence $ (t_n)_{n \in \N} \subseteq \hori$ with $t_n \to t$. In particular, for $t_n := t + \frac{1}{n}$ we get
         \begin{align*}
            \func'(t)= \lim_{n \to \infty} \frac{\func(t + \frac{1}{n})-\func(t)} {\frac{1}{n}} \geq  \lim_{n \to \infty} \frac{\func(t)-\func(t)}{\frac{1}{n}} = 0
        \end{align*}       
        where we used that $\func(t + \frac{1}{n}) \geq \func(t)$ as $\func$ is non-decreasing. 
        
        \proofitem{\ref{lem: PropAbsCon:Est}}  
        Denote again by $\mathfrak D$ the null set where $\func'$ does not exist. 
        We calculate for an arbitrary $\mathfrak T \in \mathcal{B}(\hori)$: 
        \begin{align*}
            \sigma(\func(\mathfrak T))) &\symoverset{1}{=}\sigma(\func(\mathfrak T\cap \mathfrak{D}))) = \lim_{n\to \infty} \sigma(\mathfrak T \cap\mathfrak{D}\cap  [-n,n]) 
            \symoverset{2}{\leq}  \lim_{n\to \infty} \int_{\mathfrak T\cap\mathfrak{D}\cap  [-n,n]}\abs{\func'}\di\sigma \\
            &\symoverset{3}{=}  \lim_{n\to \infty} \int_{\mathfrak T \cap  [-n,n]}\abs{\func'}\di\sigma 
            \symoverset{4}{\leq} \int_{\mathfrak T }\abs{\func'}\di\sigma.
        \end{align*}
        Here, \refsym{1} holds by $\func$ fulfilling Lusin's property (cf.~\ref{lem: PropAbsCon:Lus}) and $\hori\setminus \mathfrak D$ being a null set, implying that 
        $\sigma(\func(\mathfrak T \cap (\hori\setminus \mathfrak D))) = 0$. 
        The estimation indicated by \refsym{2} holds due to \cite[Lemma 7.10]{bruckner1997real} while equality \refsym{3} is valid by  $\hori\setminus \mathfrak D$ being a null set. Finally, \refsym{4} holds by Lebesgue's dominated convergence theorem (\cite[Theorem 2.8.1 + Corollary 2.8.6]{Bogachev2007I}) in case of $\func'$ being integrable on $\mathfrak T$ or trivially in case of $\func'$ not being integrable on $\mathfrak T$ (i.e.~$\int_{\mathfrak T }\abs{\func'}\di\sigma = \infty$). 
        
        \proofitem{\ref{lem: PropAbsCon:Conca}} Consider the interval $[-n,n]$ for an arbitrary $n \in \N$. 
        Since $\hat{\func}$ is continuous and non-decreasing, $\hat\func([-n,n])$ is contained in a bounded interval $[-n',n']$ for some $n' \in \N$. Since $\func$ is absolutely continuous on the latter interval, the claim follows by \cite[Exercise 5.8.59]{Bogachev2007I}.  
        \qedhere
    \end{structuredproof}
\end{proof}

\providecommand{\oref}{\Cref}
\section{List of Symbols}

{
	\newcommand{\losEntry}[2]{#1 & #2 \\}
	\renewcommand{\arraystretch}{1.2}

    \ifjournal
	\begin{longtable}{p{4.3cm}p{6.7cm}}
	\else
	\begin{longtable}{p{4.7cm}p{10.7cm}}
	\fi
		Symbol				& Description \\\hline
		
		\hline\multicolumn{2}{l}{\textbf{General}}\\\hline
		
		\losEntry{$L(\hori)$}{space of integrable functions on $\hori$}
		\losEntry{$L_+(\hori)$}{non-negative functions in $L(\hori)$}
        \losEntry{$L^\infty(\hori)$}{space of measurable essentially bounded functions on $\hori$}
		\losEntry{$L^\infty_+(\hori)$}{non-negative functions in $L^\infty(\hori)$}
        \losEntry{$\seql[1][M][L(\hori)]$}{vectors  $(h_m)_{m\in M}\in (L(\hori))^M$ for an arbitrary countable set $M$ whose corresponding series $\sum_{m\in M}h_m$ converges absolutely in $L(\hori)$. }
        \losEntry{$\seql[\infty][M][L^\infty(\hori)]$}{vectors  $(h_m)_{m\in M}\in (L^\infty(\hori))^M$  whose entries are uniformly bounded, i.e.~$\sup_{m \in M}\norm{h_m}_{\infty}< \infty$. }

		\losEntry{$\sigma$}{the Lebesgue measure on $\hori$}
		\losEntry{$\mathfrak{T}, \mathfrak{D}$}{measurable subsets of $\hori$}

        \losEntry{$1_{K}$}{characteristic function of set/walk $K$, i.e. $1_{K}(t)=1$ if $t \in K$  and $0$, otherwise}

        \losEntry{$\dup{f}{g}$}{the bilinear  form  $\dup{f}{g} \coloneqq \sum_{m \in M}\int_\hori f_m\cdot g_m \sigma$ between the dual pair $(\seql[1][M][L(\hori)],\seql[\infty][M][L^\infty(\hori)])$}

        \losEntry{$\eqperdef$}{indicates an equality which holds by definition}
		
		\hline\multicolumn{2}{l}{\textbf{Network}}\\\hline
		
		\losEntry{$G=(\GV,\GA)$}{directed graph with nodes $\GV$ and edges $\GA$}
		\losEntry{$\edgesFrom{v}$}{set of edges starting from node $v$}
		\losEntry{$\edgesTo{v}$}{set of edges ending at node $v$}
		\losEntry{$\source \in V$}{source node}
		\losEntry{$\dest \in V$}{destination node}
		\losEntry{$t \in \hori$}{time}
		\losEntry{$\wa=(\arc_1,\dots,\arc_k)$}{walk consisting of edges $\arc_j$}
		\losEntry{$\edgesFrom{\wa}$}{edges not in $\wa$ but starting from a node visited by $\wa$}
		\losEntry{$\wa[j]$}{$j$-th edge of walk $\wa$} 
		\losEntry{$\hat\wa_{\leq j}$ ($\hat\wa_{< j}$)}{subwalk of $\hat{\wa}$ up to and including (excluding) the $j$-th edge}
		\losEntry{$\hat\wa_{\geq j}$ ($\hat\wa_{> j}$)}{subwalk of $\hat{\wa}$ starting from the $j$-th ($j+1$-th) edge}
		\losEntry{$\hat\Routes_{v_1,v_2}$}{set of (finite) \stwalk[v_1][v_2]s}
		\losEntry{$\Routes'$}{arbitrary collection of (finite)  walks}  
		\losEntry{$\SimpCyc$}{set of all simple cycles}

        \hline\multicolumn{2}{l}{\textbf{Initial Flow-dependent Travel Times}}\\\hline
         \losEntry{$\trav_\arc(\cdot,\cdot): L_+(\hori)^\GA \times \hori \to \R_+$}{initial flow-dependent travel time on edge $\arc$}               
	    \losEntry{$\trav_\arc(\g,t)$}{edge traversal time under $\g$ when entering edge $\arc$ at time $t$ (absolutely continuous in~$t$)}
		\losEntry{$\exit_\arc(\g,t)$}{edge exit time when entering edge $\arc$ at time $t$ under $\g$: $\exit_\arc(\g,t) \coloneqq t+\trav_\arc(\g,t)$ (non-decreasing in~$t$)}
		\losEntry{$\arr_{\wa,j}(\g,t)$}{arrival time in front of the $j$-th edge of walk $\wa$ when entering this walk at time $t$ under $\g$}

		\hline\multicolumn{2}{l}{\textbf{\Auto Network Loadings}}\\\hline
        \losEntry{$\trav_\arc(\cdot): \hori \to \R_+$}{some fixed flow-independent (but still time-dependent) absolutely continuous travel time function on edge $\arc$}   
		\losEntry{$\exit_\arc(t)$}{edge exit time when entering edge $\arc$ at time $t$: $\exit_\arc(t) \coloneqq t+\trav_\arc(t)$ (non-decreasing in~$t$)}
		\losEntry{$\arr_{\wa,j}(t)$}{arrival time in front of the $j$-th edge of walk $\wa$ when entering this walk at time $t$}
		\losEntry{$\ell_{\Routes'}:\edom{\Routes'} \to L_+(\hori)^\GA$}{\auto network loading w.r.t.~(flow-independent) travel time function $\trav(\cdot)$ and walk collection $\Routes'$}        
        \losEntry{$\edom{\Routes'} \subseteq L_+(\hori)^{\Routes'}$}{maximal set of walk inflow rates $h$ into $\Routes'$ whose induced \auto edge flow $\ell_{\Routes'}(h)$ exists (we even have $\edom{\Routes'}\subseteq \seql[1][{\Routes'}][L_+(\hori)]$ by \oref{lem: elluContinuity:Subset})} 
		\losEntry{$h \in L_+(\hori)^{\Routes'}$}{walk-inflow for the walk collection $\Routes'$} 
		\losEntry{$\g = \ell_{\Routes'}(h) \in L(\hori)^\GA$}{\auto edge flow induced by $h$} 
		\losEntry{$\ell_{\wa,j}(h_\wa)\in L_+(\hori)$}{the flow on the $j$-th edge on walk $\wa$ induced by inflow into that walk under $ h_\wa$, without aggregating over multiple occurrences of that edge}
		\losEntry{$\edom{\wa,j} \subseteq L_+(\hori)$}{maximal set of walk inflow rates $h_\wa$ into walk $\wa$ for which $\ell_{\wa,j}(h_\wa)$ exists} 		
		\losEntry{$\ell_{\wa,\arc}(h_\wa)\in L_+(\hori)$}{the flow on edge $\arc$ on walk $\wa$ induced by inflow into that walk under $ h_\wa$,  aggregated over multiple occurrences of that edge}
		\losEntry{$\edom{\wa,\arc} \subseteq L_+(\hori)$}{maximal set of walk inflow rates $h_\wa$ into walk $\wa$ for which $\Nl[\trav(\cdot)]_{\wa,\arc}(h_\wa)$ exists} 		
		\losEntry{$\nabla_v \eflow$}{a measure denoting the \auto flow balance at node $v$ under flow $\eflow$}

		\hline\multicolumn{2}{l}{\textbf{Connected Components of \Cref{def: ConnectedComp} \wrt zero-cycle inflow rates}}\\\hline
		\losEntry{$\SimpCyc(t)$}{set of active (=flow-carrying) simply cycles at $t$}
		\losEntry{$C_n, n\in N$}{all connected components that are active during a non-trivial amount of time}
		\losEntry{$\mathfrak T_{C_n}$}{set of times at which $C_n$ is active} 
		\losEntry{$\mathfrak T_c$}{set of times at which a cycle $c$ is active}
		\losEntry{$\GV_{C_n}$}{nodes contained in $C_n$}
		\losEntry{$\GA_{C_n}$}{edges contained in $C_n$}
		\losEntry{$\edgesFrom{C_n}$}{edges leaving the connected component $C_n$}

        \hline\multicolumn{2}{l}{\textbf{Measure Operations}}\\\hline
		\losEntry{$f\cdot \leb$}{$\leb$-based measure with density $f$ given by $f\cdot \leb(\mathfrak T)\coloneq \int_{\mathfrak T}f\di\leb$} 
		\losEntry{$\mu \circ f^{-1}$}{pushforward/image measure of a measure $\mu$ \wrt function $f$ defined via $\mu \circ f^{-1}(\mathfrak T)\coloneq \mu\big(f^{-1}(\mathfrak T)\big)$ }
	\end{longtable}
}

\end{document}